\documentclass[prodmode]{acmsmall} 

\usepackage{amsmath,amssymb}
\usepackage{verbatim}
\usepackage{color}
\usepackage{array}

\usepackage{tikz}
\usepackage{pgfplots}
\usetikzlibrary{3d}
\usetikzlibrary{patterns}

\usepackage{subfig}

\usepackage{booktabs,multirow,siunitx}

\newcolumntype{R}[1]{S[table-format=#1]}

\newcommand{\vertseclabel}[3]{%
    \multirow{#1}{1.75em}{%
        \rotatebox{90}{%
            \parbox{5em}{%
                \centering #2 \\ (\cref{#3})%
            }%
        }%
    }%
}%

\usepackage[ruled]{algorithm2e}

\SetAlFnt{\small}
\SetAlCapFnt{\small}
\SetAlCapNameFnt{\small}
\SetAlCapHSkip{0pt}
\IncMargin{-\parindent}

\usepackage[naturalnames=false,hypertexnames=false]{hyperref}
\hypersetup{
  colorlinks = true,
  allcolors = blue
}

\usepackage[nameinlink]{cleveref}
\crefname{algorithm}{Alg.\@}{Algs.\@}
\crefname{figure}{Fig.\@}{Figs.\@}
\crefname{table}{Tab.\@}{Tabs.\@}
\crefname{section}{Sec.\@}{Secs.\@}
\crefname{theorem}{Thm.\@}{Thms.\@}
\Crefname{theorem}{Theorem}{Theorems}
\crefname{lemma}{Lem.\@}{Lems.\@}
\Crefname{lemma}{Lemma}{Lemmas}
\crefname{definition}{Def.\@}{Defs.\@}
\Crefname{definition}{Definition}{Definitions}
\crefname{example}{Ex.\@}{Exs.\@}
\Crefname{example}{Example}{Examples}

\newcommand{\SpGEMM}{\text{SpGEMM}}
\newcommand{\SpGEMMs}{\text{SpGEMMs}}
\newcommand{\SpMV}{SpMV}
\newcommand{\cc}{computational cube}

\newcommand{\HG}{{\cal H}}
\newcommand{\U}{{\cal U}}
\newcommand{\V}{{\cal V}}
\newcommand{\N}{{\cal N}}
\newcommand{\A}{\mathbf{A}}
\newcommand{\B}{\mathbf{B}}
\newcommand{\C}{\mathbf{C}}

\newcommand{\M}{\mathbf{M}}
\newcommand{\comp}{\text{comp}}
\newcommand{\mem}{\text{mem}}
\newcommand{\nz}{\text{nz}}
\newcommand{\mm}{\text{m}}
\newcommand{\datafile}{}
\newcommand{\yaxistitle}{Communication Cost}
\newcommand{\plotwidth}{}
\newcommand{\plotheight}{}

\newcommand{\PF}{\mathcal{F}}
\newcommand{\PA}{\mathcal{A}}
\newcommand{\PB}{\mathcal{B}}
\newcommand{\PC}{\mathcal{C}}
\newcommand{\PR}{\mathcal{R}}
\newcommand{\PL}{\mathcal{L}}
\newcommand{\PU}{\mathcal{U}}

\newcommand{\eg}{\emph{e.g.}}

\newif\ifata
\newif\ifgeomrow
\newif\ifgeomouter
\newif\ifnolegend
\atafalse
\geomrowfalse
\geomouterfalse
\nolegendfalse

\acmArticle{ }
\begin{document}

% Input macros for tikz figures
% set some shortcuts
\tikzstyle{nz} = [draw=none,scale=2]
\newcommand{\nzsymm}{\ast}
\newcommand{\nzsym}{$\nzsymm$}
\tikzstyle{matstyle} = [rotate=-90]
\newcommand{\mat}[2]{\draw[black,shift={(-.5,-.5)}] (0,0) grid (#1,#2);}

% define matrix dimensions
\newcommand{\matdimI}{3}
\newcommand{\matdimK}{4}
\newcommand{\matdimJ}{2}

% set numbers of nonzeros 
\newcommand{\nnza}{5}
\newcommand{\nnzb}{5}

% define sparse matrices
\newcommand{\Anz}{
\begin{scope}[matstyle]
	\mat{\matdimI}{\matdimK}
	\node[nz] at (0,0) {\nzsym};
	\node[nz] at (0,2) {\nzsym};
	\node[nz] at (1,0) {\nzsym};
	\node[nz] at (1,3) {\nzsym};
	\node[nz] at (2,1) {\nzsym};
\end{scope}
}
\newcommand{\Bnz}{
\begin{scope}[matstyle]
	\mat{\matdimK}{\matdimJ}
	\node[nz] at (0,1) {\nzsym};
	\node[nz] at (1,0) {\nzsym};
	\node[nz] at (2,0) {\nzsym};
	\node[nz] at (2,1) {\nzsym};
	\node[nz] at (3,1) {\nzsym};
\end{scope}
}
\newcommand{\Cnz}{
\begin{scope}[matstyle]
	\mat{\matdimI}{\matdimJ}
	\node[nz] at (0,0) {\nzsym};
	\node[nz] at (0,1) {\nzsym};
	\node[nz] at (1,1) {\nzsym};
	\node[nz] at (2,0) {\nzsym};
\end{scope}
}

% styles for drawing computational cube
\tikzset{
	CCorient/.style={x={(-0.5cm,-0.4cm)}, y={(1cm,0cm)}, z={(0cm,1cm)},every node/.append style={transform shape}},
	Aface/.style={canvas is yz plane at x=.5},
	Bface/.style={canvas is yx plane at z=.5,yscale=-1,rotate=90},
	Cface/.style={canvas is zx plane at y=(\matdimK-.5),rotate=-90}
}

% computation cube
\newcommand{\compcube}{
	% A face
	\begin{scope}[Aface,red]
		\Anz
		\node[shift={(-.5,.5)}] at (\matdimK/2,-\matdimI/2)  {\LARGE $S_\A$};
	\end{scope}
	% B face
	\begin{scope}[Bface,blue]
		\Bnz
		\node[shift={(-.5,.5)}] at (\matdimJ/2,-\matdimK/2)  {\LARGE $S_\B$};
	\end{scope}
	% C face
	\begin{scope}[Cface,violet]
		\Cnz
		\node[shift={(-.5,.5)}] at (\matdimJ/2,-\matdimI/2)  {\LARGE $S_\C$};
	\end{scope}
}

% Page heads
\markboth{Ballard, Druinsky, Knight, Schwartz}{Hypergraph Partitioning for Sparse Matrix-Matrix Multiplication}

% Title portion
\title{Hypergraph Partitioning for Sparse Matrix-Matrix Multiplication}
\author{GREY BALLARD
\affil{Sandia National Laboratories, {\tt gmballa@sandia.gov}}
ALEX DRUINSKY
\affil{Lawrence Berkeley National Laboratory, {\tt adruinsky@lbl.gov}}
NICHOLAS KNIGHT
\affil{New York University, {\tt nknight@nyu.edu}}
ODED SCHWARTZ
\affil{Hebrew University, {\tt odedsc@cs.huji.ac.il}}}
% NOTE! Affiliations placed here should be for the institution where the
%       BULK of the research was done. If the author has gone to a new
%       institution, before publication, the (above) affiliation should NOT be changed.
%       The authors 'current' address may be given in the "Author's addresses:" block (below).
%       So for example, Mr. Abdelzaher, the bulk of the research was done at UIUC, and he is
%       currently affiliated with NASA.

\begin{abstract}
We propose a fine-grained hypergraph model for sparse matrix-matrix multiplication (SpGEMM), a key computational kernel in scientific computing and data analysis whose performance is often communication bound. This model correctly describes both the interprocessor communication volume along a critical path in a parallel computation and also the volume of data moving through the memory hierarchy in a sequential computation. We show that identifying a communication-optimal algorithm for particular input matrices is equivalent to solving a hypergraph partitioning problem. Our approach is sparsity dependent, meaning that we seek the best algorithm for the given input matrices.

In addition to our (3D) fine-grained model, we also propose coarse-grained 1D and 2D models that correspond to simpler SpGEMM algorithms. We explore the relations between our models theoretically, and we study their performance experimentally in the context of three applications that use SpGEMM as a key computation. For each application, we find that at least one coarse-grained model is as communication efficient as the fine-grained model. We also observe that different applications have affinities for different algorithms.

Our results demonstrate that hypergraphs are an accurate model for reasoning about the communication costs of SpGEMM as well as a  practical tool for exploring the SpGEMM algorithm design space.
\end{abstract}

% (to be added)

%\terms{}

%\keywords{}

%\acmformat{Grey Ballard, Alex Druinsky, Nicholas Knight, and Oded Schwartz, YYYY. Hypergraph partitioning for sparse matrix matrix multiplication.}

%\begin{bottomstuff}
%Authors' addresses: 
%G.~Ballard: {\tt gmballa@sandia.gov}, 
%A.~Druinsky: {\tt adruinsky@lbl.gov}, 
%N.~Knight: {\tt nknight@nyu.edu}, 
%O.~Schwartz: {\tt odedsc@cs.huji.ac.il}.
%\end{bottomstuff}

\maketitle

\section{Introduction}
\label{sec:intro}

Sparse matrix-matrix multiplication (\SpGEMM) is a fundamental computation in scientific computing and data analysis.
It is a key component in applications ranging from linear solvers~\cite{BV11,YL11} and graph algorithms~\cite{RV89,ABG15} to Kohn-Sham theory in computational chemistry~\cite{BVWH14}.
Unlike dense matrix-matrix multiplication, \SpGEMM{} is an irregular computation, and its performance is typically communication bound.
Previous research on \SpGEMM{} algorithms has focused on communication costs, both interprocessor communication and data movement within the memory hierarchy~\cite{BG12,BBDG+13,AB+15-TR}.

We focus on communication costs and argue in this paper that the amount of communication required for sequential and parallel \SpGEMM{} algorithms depends strongly on the matrices' nonzero structures.
That is, algorithms that are communication efficient for \SpGEMMs{} within one application are not necessarily the most efficient in the context of another application.
The simplest \SpGEMM{} algorithms to design, analyze, and implement are based on dense matrix multiplication algorithms: these algorithms are parametrized only by the matrix dimensions and differ from dense matrix-matrix multiplication in that they use sparse data structures and avoid operating on and communicating zeros.
We refer to these algorithms as ``sparsity independent''~\cite[Def.~2.5]{BBDG+13}; two examples are the row-wise algorithm with block partitioning~\cite{FPS14} and Sparse SUMMA (SpSUMMA)~\cite{BG12}.

One of the main goals in this paper is to explore sparsity-dependent \SpGEMM{} algorithms and determine how communication efficient an algorithm can be if it is able to inspect the input matrices' nonzero structures and determine the data distributions and parallelization of arithmetic operations accordingly.
Currently this approach is only theoretical in general: in practice, the cost of nonzero structure inspection can exceed that of even an inefficient \SpGEMM.
However, practical, application-specific algorithmic insight can be gained from the results of this paper and our techniques can be used for other applications not considered here.

We use hypergraphs to model arithmetic operations and dependencies and hypergraph partitioning to model parallelizing and scheduling operations.
Hypergraphs have previously been used to model sparse matrix-vector multiplication (SpMV) (see, \eg,~\cite{VB05,CAU10}), and hypergraph partitioning software has been developed to help design algorithms for SpMV and other computations~\cite{PaToH,BDFHH07}.
They have also been used in the context of \SpGEMM{}~\cite{AA14,BDKS15}; in this paper we present a hypergraph model that generalizes previous models for SpMV and \SpGEMM.
We discuss these and other related works in \cref{sec:related}.

We present a general, fine-grained hypergraph model for \SpGEMM{} in \cref{sec:model} and use it to prove sparsity-dependent communication lower bounds for both parallel (\cref{sec:parallelLB}) and sequential (\cref{sec:sequentialLB}) algorithms.
We reduce identifying a communication-optimal algorithm for given input matrices to solving a hypergraph partitioning problem.
We also compare with previous lower bounds, showing that ours are more general and can be tighter.

In addition to our general \SpGEMM{} model, we also consider several restricted classes of algorithms.
In \cref{sec:simplify-HG} we define a framework for coarsening the fine-grained model to obtain simpler hypergraphs that maintain correct modeling of communication costs.
In particular, we consider the classification of 1D, 2D, and 3D \SpGEMM{} algorithms~\cite{BBDG+13} and examine the relationships among the seven subclasses of algorithms that naturally arise from this classification.
The class of 1D algorithms includes row-wise and outer-product algorithms, and the class of 2D algorithms includes SpSUMMA.
Finding communication-optimal algorithms within each subclass also corresponds to solving simpler hypergraph partitioning problems.

Our experimental results, presented in \cref{sec:expt}, use existing hypergraph partitioning software to compare the communication costs of  algorithms from each subclass for a variety of \SpGEMMs{}, with input matrices coming from three particular applications.
We use representative examples from each application to explore the algorithmic space of \SpGEMM{}; our results can guide algorithmic design choices that are application specific.
In particular, our empirical results lead to different conclusions for each of the three applications we consider.
For example, we find that certain 1D algorithms, although simple, are sufficient in the context of an algebraic multigrid application~\cite{BV11} but that 2D and 3D algorithms are much more communication efficient than 1D algorithms in the context of Markov clustering  applied to social-network matrices~\cite{SP09}.
We discuss overall conclusions in more detail in \cref{sec:conclusion}.

As is the case for other irregular computations, there exists a tradeoff between simplicity and efficiency for \SpGEMM{} algorithms and implementations.
Using hypergraphs to model \SpGEMM{} communication costs, we provide a mechanism for quantifying the relative efficiency of algorithms of varying complexity.
We illuminate this tradeoff in the context of several applications, providing practical insight into the design of algorithms and software.

%\GB{our contributions:
%\begin{itemize}
%	\item a hypergraph model that accurately models communication cost (appeared in SPAA version)
%	\item sequential and parallel communication lower bounds for \SpGEMM{} that are input-sparsity-specific and asymptotically tight
%	\item simplified hypergraph models based on parallelization and data distribution restrictions, which yield smaller (easier-to-partition) hypergraphs and simpler algorithms
%	\item empirical comparisons of the communication costs for different \SpGEMM{} algorithms based on partitions discovered by PaToH for three different applications requiring \SpGEMM{}
%	\item application-specific conclusions for the simplest and most efficient algorithmic approaches
%\end{itemize}
%}

\section{Related Work}
\label{sec:related}

We base our communication models and analysis on classical results for dense matrix multiplication.
Hong and Kung~\citeyear{HK81} used a two-level sequential memory model and established communication lower bounds for dense matrix multiplication and other computations using a graph-theoretic approach.
Irony, Toledo, and Tiskin~\citeyear{ITT04} used a distributed-memory parallel model and established communication lower bounds for dense matrix multiplication using a geometric approach.
See Ballard et al.'s survey~\citeyear{BCDH+14} for more details on these models and results.

Ballard et al.~\citeyear{BDHS11,BDHLS12}
proved communication lower bounds, in both sequential and parallel models, for a general set of matrix computations that includes \SpGEMM{}.
These lower bounds are parameterized by the number of  multiplications involved in the multiplication of the specific input matrices, and they are not tight in general.
There has been other theoretical work in deriving bounds on the communication costs of \SpGEMM{} in the sequential two-level memory model.
Pagh and St\"{o}ckel~\citeyear{PS14} proved matching worst-case upper and lower bounds based on the number of nonzeros in the matrices.
Greiner~\citeyear[Ch.~6]{Greiner12} used a different approach to establish other worst-case lower and upper bounds that apply to a restricted set of matrices and algorithms.
We compare our lower bounds with previous work in more detail in \cref{sec:sequentialLB,sec:parallelLB}.

On the practical side of sequential \SpGEMM{} algorithms, Gustavson~\citeyear{Gustavson78} proposed the first algorithm that exploits sparsity in the inputs and outputs, and Davis~\citeyear{Davis06} implemented a variant of that algorithm in the general-purpose CSPARSE library.
Neither the algorithm nor implementation are designed to specifically reduce communication costs specifically, but the approach is typically efficient in terms of memory footprint and operation count.
Bulu\c{c} and Gilbert~\citeyear{BG08a} proposed an alternative sparse data structure and sequential algorithm that are more memory and operation-count efficient than Gustavson's algorithm for ``hypersparse'' matrices, which can arise in the context of parallel distribution of sparse matrices.

There exists more variation in parallel \SpGEMM{} algorithms.
We use the 1D/2D/3D classification of \SpGEMM{} algorithms defined in~\cite{BBDG+13}.
They surveyed many of the previous algorithms in the literature and theoretically analyzed their communication costs for a particular class of input matrices (Erd\H{o}s-R\'{e}nyi random matrices), comparing with expectation-based communication lower bounds.
Bulu\c{c} and Gilbert~\citeyear{BG12} proposed a general-purpose 2D algorithm called Sparse SUMMA, which uses random permutations to achieve load balance, and they analyzed its communication cost for general inputs.
Ballard, Siefert, and Hu~\citeyear{BSH15-TR} considered the communication costs of 1D algorithms in an algebraic multigrid application.

As a tool for minimizing communication, 
hypergraph partitioning has gained popularity in the parallel computing community in large part due to its application \SpMV{}.
While there is a vast literature on hypergraph partitioning for SpMV, we highlight here only a small sample from which we borrow notation.
\c{C}ataly\"{u}rek and Aykanat~\citeyear{CA99} introduced column-net and row-net models for 1D SpMV algorithms, and in a subsequent paper~\citeyear{CA01a} presented the fine-grain model for 2D SpMV.
\c{C}ataly\"{u}rek, Aykanat, and U\c{c}ar~\citeyear{CAU10} summarize and compare models for 1D and 2D SpMV algorithms (see also~\cite{VB05}).
Akbudak, Kayaaslan, and Aykanat~\citeyear{AKA13} also used hypergraphs to improve cache locality of SpMV on a single processor.

Hypergraph partitioning has been used in the context of \SpGEMM{} as well.
Krishnamoorthy et al.~\citeyear{KCNRS06} introduced a hypergraph model for a general class of computations including \SpGEMM{}, and described a heuristic for scheduling out-of-core algorithms with the goal of minimizing disk I/O.
Akbudak and Aykanat~\citeyear{AA14} introduced a hypergraph model for parallel outer-product algorithms that represents both the computation and the data involved in \SpGEMM{} (see \cref{ex:outer}); they also implemented and benchmarked the performance based on the partitions.
Our earlier conference paper~\citeyear{BDKS15} introduced a slightly simpler version of the \SpGEMM{} hypergraph model given by \cref{def:hp-spgemm}; that paper also presented initial results concerning the application of algebraic multigrid (see \cref{sec:AMG}).

\section{Hypergraph Model for \SpGEMM{}}
\label{sec:model}

\subsection{Notation and Assumptions}
\label{sec:notation}
The notation introduced in this section will be used throughout the rest of the paper.

Let $\mathbb{N} = \{0,1,\ldots\}$; for any $n \in \mathbb{N}$, $[n] = \{1,\ldots,n\}$, with $[0]=\emptyset$.
For any $n \in \mathbb{N}$ and set $X$, an $n$-way partition of $X$, denoted $\{X_1,\ldots,X_n\}$, is a function $[n] \ni i \mapsto X_i \subseteq X$ such that $\bigcup_{i \in [n]}X_i = X$ and $\bigcup_{i \ne j \in [n]} X_i \cap X_j = \emptyset$.
Another partition of $X$, $\{Y_1,\ldots,Y_m\}$, refines $\{X_1,\ldots,X_n\}$ if for every $j \in [m]$ there exists an $i \in [n]$ such that $Y_j \subseteq X_i$.

Let $\A$ and $\B$ be $I$-by-$K$ and $K$-by-$J$ matrices with entries from a set $X$, which contains an element $0$ and is closed under two binary operations called addition (commutative and associative with identity element 0) and multiplication (with absorbing element 0).
For example, $X$ could be a semiring.

Matrix multiplication is the function $(\A,\B) \mapsto \C=\A\cdot\B$, where $\C$ is an $I$-by-$J$ matrix over $X$ defined entrywise by $c_{ij}=\sum_{k \in [K]} a_{ik} b_{kj}$.
We let $S_\A \subseteq [I]\times[K]$, $S_\B \subseteq [K]\times[J]$, and $S_\C \subseteq [I]\times[J]$ denote the nonzero structures of $\A$, $\B$, and $\C$. 
In this work, we study \emph{\SpGEMM{} algorithms}, the class of algorithms that evaluate and sum all \emph{nontrivial} multiplications $a_{ik}b_{kj}$, where both $a_{ik}\ne 0$ and $b_{kj} \ne 0$, and thus depend only on $S_\A$ and $S_\B$.
We do not consider algorithms that exploit additional structure on $X$ or more general relations on the entries of $\A$ and $\B$: in particular, we ignore numerical cancellation, so $S_\A$ and $S_\B$ induce $S_\C$. 
\Cref{fig:notation} illustrates our notation for a particular \SpGEMM{} instance; under our assumptions we do not distinguish the nonzero values of the matrices.
We maintain the nonzero structure of the example in each of \cref{fig:notation,fig:compcube,fig:hypergraph,fig:incidence}.

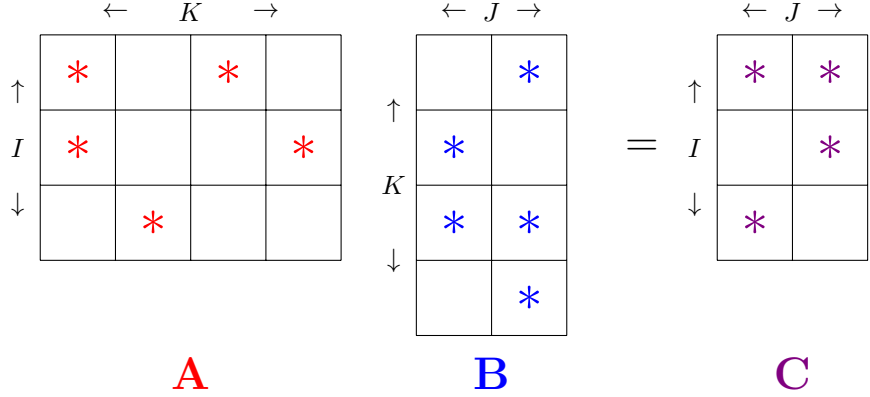
\begin{figure}
\begin{center}
\begin{tikzpicture}

% label matrices at same height (based on largest number of rows)
\pgfmathsetmacro{\matlabely}{-max(\matdimI,\matdimK,\matdimJ)}
\newcommand{\dimlabeloffset}{.3}

\begin{scope}
	\begin{scope}[red]
		\Anz
		\node[draw=none,shift={(-.5,0)}] at (\matdimK/2,\matlabely) {\LARGE $\A$};
	\end{scope}
	\begin{scope}[shift={(-.5,.5)}]
		\node[draw=none] at (-\dimlabeloffset,-\matdimI/2) {$I$};
		\node[draw=none] at (-\dimlabeloffset,-\matdimI/2+\matdimI/4) {$\uparrow$};
		\node[draw=none] at (-\dimlabeloffset,-\matdimI/2-\matdimI/4) {$\downarrow$};
		\node[draw=none] at (\matdimK/2,\dimlabeloffset) {$K$};
		\node[draw=none] at (\matdimK/2+\matdimK/4,\dimlabeloffset) {$\rightarrow$};
		\node[draw=none] at (\matdimK/2-\matdimK/4,\dimlabeloffset) {$\leftarrow$};
	\end{scope}
\end{scope}
\begin{scope}[shift={(\matdimK+1,0)}]
	\begin{scope}[blue]
		\Bnz
		\node[draw=none,shift={(-.5,0)}] at (\matdimJ/2,\matlabely) {\LARGE $\B$};
	\end{scope}
	\begin{scope}[black,shift={(-.5,.5)}]
		\node[draw=none] at (-\dimlabeloffset,-\matdimK/2) {$K$};
		\node[draw=none] at (-\dimlabeloffset,-\matdimK/2+\matdimK/4) {$\uparrow$};
		\node[draw=none] at (-\dimlabeloffset,-\matdimK/2-\matdimK/4) {$\downarrow$};
        		\node[draw=none] at (\matdimJ/2,\dimlabeloffset) {$J$};
		\node[draw=none] at (\matdimJ/2+\matdimJ/4,\dimlabeloffset) {$\rightarrow$};
		\node[draw=none] at (\matdimJ/2-\matdimJ/4,\dimlabeloffset) {$\leftarrow$};
	\end{scope}
	\node[draw=none,shift={(-.5,.5)}] at (\matdimJ+1,-\matdimI/2) {\LARGE $=$};
\end{scope}
\begin{scope}[shift={(\matdimK+\matdimJ+3,0)}]
	\begin{scope}[violet]
		\Cnz
		\node[draw=none,shift={(-.5,0)}] at (\matdimJ/2,\matlabely) {\LARGE $\C$};
	\end{scope}
	\begin{scope}[black,shift={(-.5,.5)}]
		\node[draw=none] at (-\dimlabeloffset,-\matdimI/2) {$I$};
		\node[draw=none] at (-\dimlabeloffset,-\matdimI/2+\matdimI/4) {$\uparrow$};
		\node[draw=none] at (-\dimlabeloffset,-\matdimI/2-\matdimI/4) {$\downarrow$};
        		\node[draw=none] at (\matdimJ/2,\dimlabeloffset) {$J$};
		\node[draw=none] at (\matdimJ/2+\matdimJ/4,\dimlabeloffset) {$\rightarrow$};
		\node[draw=none] at (\matdimJ/2-\matdimJ/4,\dimlabeloffset) {$\leftarrow$};
	\end{scope}
\end{scope}

\end{tikzpicture}
\end{center}
\caption{Notation for a particular \SpGEMM{} instance.}%
\label{fig:notation}
\end{figure}

Given additional structure on $X$, there may exist many algorithms that do not simply compute the sums of products.
For example, if $X$ is a ring, then additive inverses can be exploited to obtain Strassen's algorithm,
 where the multiplicands do not generally equal the input matrix entries. 
As another example, given known relations on the entries of $\A$ and $\B$, one may be able to avoid evaluating some nontrivial multiplications.

To simplify the presentation, we assume that neither $\A$ or $\B$ have any zero rows or columns, so every nonzero input matrix entry participates in at least one nontrivial multiplication. While this can always be enforced, doing so may incur a preprocessing cost.

Our hypergraph terminology borrows from \c{C}ataly\"{u}rek and Aykanat~\citeyear{CA99}.
A \emph{hypergraph} $\HG$ is a generalization of a graph; it consists of a set of \emph{vertices} $\V$ and a set of \emph{nets} $\N$, where each net $n \in \N$ is a subset of the vertices, $n\subseteq \V$.
We use the term \emph{pin} to refer to a vertex in a particular net.
Each vertex $v\in \V$ may have a \emph{weight} $w$ associated with it, and each net $n\in \N$ may have a \emph{cost} $c$ associated with it.
Weights and costs are typically scalar-valued but can be vector-valued; we will use vector-valued weights.

\subsection{Fine-Grained Hypergraph Model}
\label{sec:fg-model}

In this section we define our most general hypergraph model for \SpGEMM{}.
It may be helpful to visualize an instance of \SpGEMM{} geometrically, as the three-dimensional ``iteration space'' of matrix multiplication:
\Cref{fig:compcube} illustrates this geometric perspective.
The set of $IKJ$  multiplications are arranged as a set of cubes so that nontrivial multiplications are distinguished by having nonzero projections on the $\A$- and $\B$-faces.
Nonzero entries of $\C$ are determined by projecting the nontrivial multiplications onto the $\C$-face.

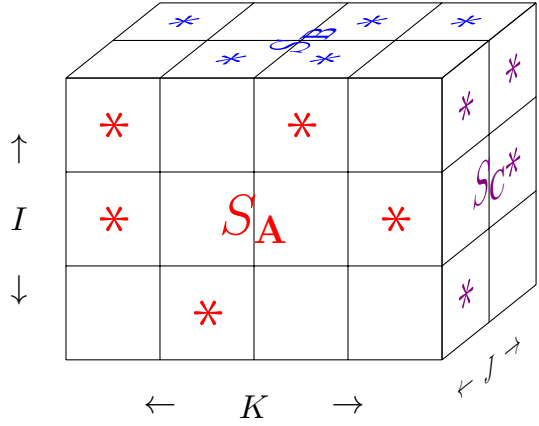
\begin{figure}
\begin{center}
\begin{tikzpicture}[CCorient,scale=1.25]

% base cube 
\compcube

% dimension labels
\begin{scope}[Aface,shift={(-.5,.5)}]
	\node[draw=none] at (-.5,-\matdimI/2) {$I$};
	\node[draw=none] at (-.5,-\matdimI/2+\matdimI/4) {$\uparrow$};
	\node[draw=none] at (-.5,-\matdimI/2-\matdimI/4) {$\downarrow$};
	\node[draw=none] at (\matdimK/2,-\matdimI-.5) {$K$};
	\node[draw=none] at (\matdimK/2+\matdimK/4,-\matdimI-.5) {$\rightarrow$};
	\node[draw=none] at (\matdimK/2-\matdimK/4,-\matdimI-.5) {$\leftarrow$};
\end{scope}
\begin{scope}[Cface,shift={(-.5,.5)}]
        	\node[draw=none] at (\matdimJ/2,-\matdimI-.5) {$J$};
	\node[draw=none] at (\matdimJ/2+\matdimJ/4,-\matdimI-.5) {$\rightarrow$};
	\node[draw=none] at (\matdimJ/2-\matdimJ/4,-\matdimI-.5) {$\leftarrow$};
\end{scope}

\end{tikzpicture}
\end{center}
\caption{Geometric view of the \SpGEMM{} instance in \cref{fig:notation}. correspond to the iteration space of matrix multiplication.  Each sub-cube represents a  multiplication, which is nontrivial if its projections onto the $\A$ and $\B$ faces both correspond to nonzero values.  The projections of nontrivial multiplications onto the $\C$ face define the nonzero structure of $\C$.}
\label{fig:compcube}
\end{figure}

We now define our hypergraph model and connect its definition to its geometric interpretation.

\begin{definition}
\label[definition]{def:hp-spgemm}
Consider matrices $\A$, $\B$, and $\C = \A\cdot\B$, continuing notation.
Define the \emph{fine-grained \SpGEMM{} hypergraph $\HG(\A,\B)=(\V,\N)$}, where $\V = \V^\mm \cup \V^\nz$, with
\begin{align*}
\V^\nz &= \V^\A \cup \V^\B \cup \V^\C \text{,}\\
\V^\mm &= \{ v_{ikj} : (i,k) \in S_\A \wedge (k,j) \in S_\B \} \text{,} \\
\V^\A &= \{ v_{ik}^\A : (i,k) \in S_\A \} \text{,} \\
\V^\B &= \{ v_{kj}^\B : (k,j) \in S_\B \} \text{, and} \\
\V^\C &= \{ v_{ij}^\C : (i,j) \in S_\C \} \text{,}
\end{align*}
and $\N=\N^\A \cup \N^\B \cup \N^\C$, with
\begin{align*}
\N^\A &= \{ n_{ik}^\A : (i,k) \in S_\A \}\text{,} \\
\N^\B &= \{ n_{kj}^\B : (k,j) \in S_\B \}\text{, and} \\
\N^\C &= \{ n_{ij}^\C : (i,j) \in S_\C \}\text{.}
\end{align*}
Net membership is defined such that each $v_{ik}^\A \in n_{ik}^\A$, each $v_{kj}^\B \in n_{kj}^\B$, each $v_{ij}^\C \in n_{ij}^\C$, and 
each $v_{ikj} \in n_{ik}^\A \cap n_{kj}^\B \cap n_{ij}^\C$.
% for every $v_{ikj}\in \V^\mm$, $v_{ikj}\in n_{ik}^\A$, $v_{ikj}\in n_{kj}^\B$, and $v_{ikj}\in n_{ij}^\C$.
Equivalently, pins are defined such that for each $(i,k) \in S_\A$,
\[ n_{ik}^\A = \{ v_{ikj} : (k,j) \in S_\B \} \cup \{ v_{ik}^\A \} \text{,}\]
for each $(j,k) \in S_\B$,
\[ n_{jk}^\B = \{ v_{ikj} : (i,k) \in S_\A \} \cup \{ v_{kj}^\B \} \text{,} \]
and for each $(i,j)\in S_\C$,
\[ n_{ij}^\C = \{ v_{ikj} : (i,k) \in S_\A \wedge (k,j) \in S_\B \} \cup \{ v_{ij}^\C \}\text{.}\]
The vertices $\V$ have two types of \emph{weights}, corresponding to computation and memory.
For all $v_{ikj} \in \V^\mm$, 
\[\begin{array}{r c l}
w_\comp(v_{ikj}) & = & 1 \text{,}\\
w_\mem(v_{ikj}) & = & 0 \text{.}
\end{array}\]
For all $v_{ik}^\A,v_{kj}^\B,v_{ij}^\C \in \V^\mm$, 
\[\begin{array}{r c c c c c l}
w_\comp(v_{ik}^\A) & = & w_\comp(v_{kj}^\B) & = & w_\comp(v_{ij}^\C) & = & 0 \text{,}\\
w_\mem(v_{ik}^\A) & = & w_\mem(v_{kj}^\B) & = & w_\mem(v_{ij}^\C) & = & 1 \text{.} 
\end{array}\]
The \emph{cost} of each net $n \in \N$ is $c(n)=1$.
\end{definition}

In summary, each vertex $v$ of $\HG(\A,\B)$ corresponds either to a nontrivial multiplication or to a nonzero entry of $\A$, $\B$, or $\C$, and each net corresponds to a nonzero of $\A$, $\B$, or $\C$.
Each \emph{multiplication vertex} $v \in \V^\mm$ is a member of the three nets corresponding to nonzeros involved in the  multiplication; each \emph{nonzero vertex} $v \in \V^\nz$ is a member of its corresponding net.
Conversely, each net corresponds to a nonzero value and includes as members its corresponding nonzero vertex as well all associated multiplication vertices.

In the geometric context of \cref{fig:compcube}, we can think of each multiplication vertex $v_{ikj}$ as a point in the three-dimensional iteration space, and we can think of each nonzero vertex $v_{ik}^\A$, $v_{kj}^\B$, or $v_{ij}^\C$ as a point on one of the (two-dimensional) faces of the iteration space.
Likewise, a net $n_{ik}^\A$, $n_{kj}^\B$, or $n_{ij}^\C$ corresponds to a (one-dimensional) fiber perpendicular to one of the faces, and its pins are the vertices that lie along the fiber.

We provide two other visualizations of the hypergraph defined by the \SpGEMM{} instance given in \cref{fig:notation}.
\Cref{fig:hypergraph} depicts the hypergraph as a bipartite graph, where hypergraph vertices are denoted by circles, hypergraph nets are denoted by squares, and net membership is defined by edges in the graph.
\Cref{fig:incidence} shows the incidence matrix of the hypergraph, where rows correspond to vertices, columns correspond to nets, and net membership is defined by the nonzero pattern.

\begin{figure}
\begin{center}
\begin{tikzpicture}[yscale=.75]

\newcommand{\nzvsym}[2]{$v^{#1}_{#2}$}
\newcommand{\netsym}[2]{$n^{#1}_{#2}$}
\newcommand{\mvsym}[1]{$v_{#1}$}
\newcommand{\rotangle}{30}

\tikzset{vs/.style={draw,circle}}
\tikzset{ns/.style={draw}}

% A 
\begin{scope}[red,shift={(0,-3.25)},rotate=-\rotangle]
	% nonzero vertices 
	\node[vs] (va00) at (0,0) {\nzvsym{\A}{00}};
	\node[vs] (va02) at (1,0) {\nzvsym{\A}{02}};
	\node[vs] (va10) at (2,0) {\nzvsym{\A}{10}};
	\node[vs] (va13) at (3,0) {\nzvsym{\A}{13}};
	\node[vs] (va21) at (4,0) {\nzvsym{\A}{21}};
        % nets
        \begin{scope}[shift={(0,1.25)}]
            	\node[ns] (na00) at (0,0) {\netsym{\A}{00}};
            	\node[ns] (na02) at (1,0) {\netsym{\A}{02}};
            	\node[ns] (na10) at (2,0) {\netsym{\A}{10}};
            	\node[ns] (na13) at (3,0) {\netsym{\A}{13}};
            	\node[ns] (na21) at (4,0) {\netsym{\A}{21}};
        \end{scope}
\end{scope}

% B 
\begin{scope}[blue,shift={(6.75,-5.25)},rotate=\rotangle]
	% nonzero vertices
	\node[vs] (vb01) at (0,0) {\nzvsym{\B}{01}};
	\node[vs] (vb10) at (1,0) {\nzvsym{\B}{10}};
	\node[vs] (vb20) at (2,0) {\nzvsym{\B}{20}};
	\node[vs] (vb21) at (3,0) {\nzvsym{\B}{21}};
	\node[vs] (vb31) at (4,0) {\nzvsym{\B}{31}};
        % nets
        \begin{scope}[blue,shift={(0,1.25)}]
            	\node[ns] (nb01) at (0,0) {\netsym{\B}{01}};
            	\node[ns] (nb10) at (1,0) {\netsym{\B}{10}};
            	\node[ns] (nb20) at (2,0) {\netsym{\B}{20}};
            	\node[ns] (nb21) at (3,0) {\netsym{\B}{21}};
            	\node[ns] (nb31) at (4,0) {\netsym{\B}{31}};
        \end{scope}
\end{scope}

% multiplication vertices
\begin{scope}[shift={(5,0)},xscale=1.5]
	\node[vs] (vm020) at (-2.5,0) {\mvsym{020}};
	\node[vs] (vm001) at (-1.5,0) {\mvsym{001}};
	\node[vs] (vm021) at (-.5,0) {\mvsym{021}};
	\node[vs] (vm101) at (.5,0) {\mvsym{101}};
	\node[vs] (vm131) at (1.5,0) {\mvsym{131}};
	\node[vs] (vm210) at (2.5,0) {\mvsym{210}};
\end{scope}

% C 
\begin{scope}[violet,shift={(5,2.5)}]
	% nets
	\node[ns] (nc00) at (-1.5,0) {\netsym{\C}{00}};
	\node[ns] (nc01) at (-.5,0) {\netsym{\C}{01}};
	\node[ns] (nc11) at (.5,0) {\netsym{\C}{11}};
	\node[ns] (nc20) at (1.5,0) {\netsym{\C}{20}};
	% nonzero vertices
	\begin{scope}[violet,shift={(0,1.25)}]
		\node[vs] (vc00) at (-1.5,0) {\nzvsym{\C}{00}};
		\node[vs] (vc01) at (-.5,0) {\nzvsym{\C}{01}};
		\node[vs] (vc11) at (.5,0) {\nzvsym{\C}{11}};
		\node[vs] (vc20) at (1.5,0) {\nzvsym{\C}{20}};
	\end{scope}
\end{scope}

% A edges
\draw (va00) -- (na00);
\draw (va02) -- (na02);
\draw (va10) -- (na10);
\draw (va13) -- (na13);
\draw (va21) -- (na21);

% B edges
\draw (vb01) -- (nb01);
\draw (vb10) -- (nb10);
\draw (vb20) -- (nb20);
\draw (vb21) -- (nb21);
\draw (vb31) -- (nb31);

% C edges
\draw (vc00) -- (nc00);
\draw (vc01) -- (nc01);
\draw (vc11) -- (nc11);
\draw (vc20) -- (nc20);

% multiplication edges
\draw (vm020) -- (na02);
\draw (vm020) -- (nb20);
\draw (vm020) -- (nc00);
\draw (vm001) -- (na00);
\draw (vm001) -- (nb01);
\draw (vm001) -- (nc01);
\draw (vm021) -- (na02);
\draw (vm021) -- (nb21);
\draw (vm021) -- (nc01);
\draw (vm101) -- (na10);
\draw (vm101) -- (nb01);
\draw (vm101) -- (nc11);
\draw (vm131) -- (na13);
\draw (vm131) -- (nb31);
\draw (vm131) -- (nc11);
\draw (vm210) -- (na21);
\draw (vm210) -- (nb10);
\draw (vm210) -- (nc20);

\end{tikzpicture}
\end{center}
\caption{Hypergraph of the \SpGEMM{} instance in \cref{fig:notation}. Circles correspond to vertices and squares correspond to nets.}
\label{fig:hypergraph}
\end{figure}
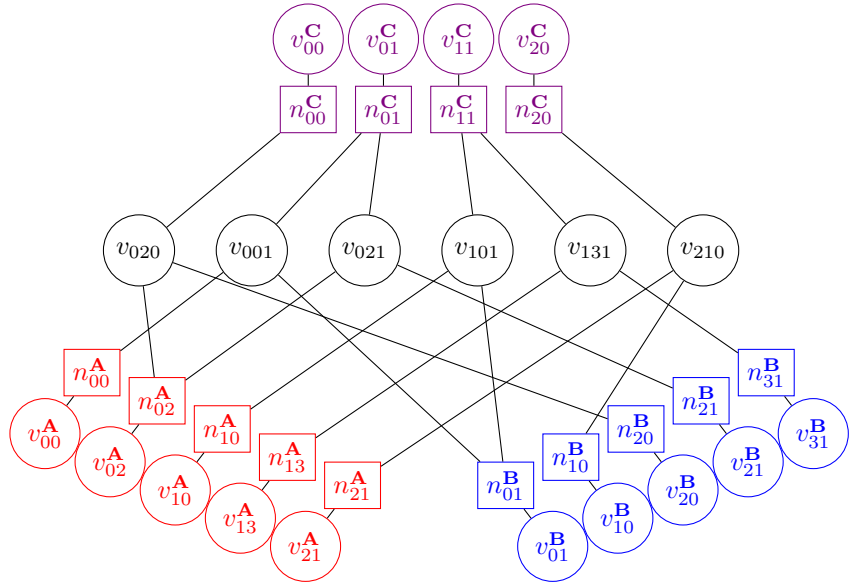

\begin{figure}
\begin{center}
\begin{tikzpicture}

\newcommand{\offsetlabel}{1}

\begin{scope}[matstyle,scale=.75]

	% draw matrix
	\mat{6}{14}
	
	% label rows
	\begin{scope}[shift={(0,-\offsetlabel)}]
		\node at (0,0) {$v_{020}$};
		\node at (1,0) {$v_{001}$};
		\node at (2,0) {$v_{021}$};
		\node at (3,0) {$v_{101}$};
		\node at (4,0) {$v_{131}$};
		\node at (5,0) {$v_{210}$};
	\end{scope}
	
	% label columns
	\begin{scope}[red,shift={(-\offsetlabel,0)}]
		\node at (0,0) {$n^{\A}_{00}$};
		\node at (0,1) {$n^{\A}_{02}$};
		\node at (0,2) {$n^{\A}_{10}$};
		\node at (0,3) {$n^{\A}_{13}$};
		\node at (0,4) {$n^{\A}_{21}$};
	\end{scope}
	\begin{scope}[blue,shift={(-\offsetlabel,5)}]
		\node at (0,0) {$n^{\B}_{01}$};
		\node at (0,1) {$n^{\B}_{10}$};
		\node at (0,2) {$n^{\B}_{20}$};
		\node at (0,3) {$n^{\B}_{21}$};
		\node at (0,4) {$n^{\B}_{31}$};
	\end{scope}
	\begin{scope}[violet,shift={(-\offsetlabel,10)}]
		\node at (0,0) {$n^{\C}_{00}$};
		\node at (0,1) {$n^{\C}_{01}$};
		\node at (0,2) {$n^{\C}_{11}$};
		\node at (0,3) {$n^{\C}_{20}$};
	\end{scope}
	
	% set nonzeros
	\node[nz,red] at (0,1) {\nzsym};
	\node[nz,blue] at (0,\nnza+2) {\nzsym};
	\node[nz,violet] at (0,\nnza+\nnzb+0) {\nzsym};
	
	\node[nz,red] at (1,0) {\nzsym};
	\node[nz,blue] at (1,\nnza+0) {\nzsym};
	\node[nz,violet] at (1,\nnza+\nnzb+1) {\nzsym};
	
	\node[nz,red] at (2,1) {\nzsym};
	\node[nz,blue] at (2,\nnza+3) {\nzsym};
	\node[nz,violet] at (2,\nnza+\nnzb+1) {\nzsym};
	
	\node[nz,red] at (3,2) {\nzsym};
	\node[nz,blue] at (3,\nnza+0) {\nzsym};
	\node[nz,violet] at (3,\nnza+\nnzb+2) {\nzsym};
	
	\node[nz,red] at (4,3) {\nzsym};
	\node[nz,blue] at (4,\nnza+4) {\nzsym};
	\node[nz,violet] at (4,\nnza+\nnzb+2) {\nzsym};
	
	\node[nz,red] at (5,4) {\nzsym};
	\node[nz,blue] at (5,\nnza+1) {\nzsym};
	\node[nz,violet] at (5,\nnza+\nnzb+3) {\nzsym};

\end{scope}

\end{tikzpicture}
\end{center}
\caption{Submatrix of the incidence matrix of the hypergraph of the \SpGEMM{} instance in \cref{fig:notation}. Rows correspond to multiplication vertices and columns correspond to nets. The rows of the incidence matrix not shown here correspond to nonzero vertices; they form a diagonal submatrix.
}%
\label{fig:incidence}
\end{figure}

This hypergraph is nearly the same as the one proposed in our earlier paper~\cite[Definition 1]{BDKS15}.
The difference here is the introduction of nonzero vertices and the inclusion of each in its corresponding nonzero's net.
The purpose of adding these vertices is to enable enforcing memory balance in the parallel case (see \cref{sec:parallelLB}), and for relating the hypergraph model to a computation DAG model in the sequential case (see \cref{sec:sequentialLB}).

\section{Communication Lower Bounds}
\label{sec:LB}

\subsection{Parallel Lower Bound}
\label{sec:parallelLB}

We consider performing \SpGEMM{} with input matrices $\A$ and $\B$ on a parallel machine with $p$ processors with disjoint memories.
Let $\HG=\HG(\A,\B)$.
A \emph{parallelization} is a $p$-way partition of $\V^\mm$ (assigning multiplications to processors) and a \emph{data distribution} is a $p$-way partition of $\V^\nz$ (assigning nonzeros to processors).
Thus, a partition of $\V$ defines both a parallelization and a data distribution.
We define a \emph{parallel \SpGEMM{} algorithm} by not only a parallelization and data distribution, but also by a communication pattern and each processor's local computation algorithm(s).
So, a single partition of $\V$ corresponds to a family of parallel \SpGEMM{} algorithms, and each parallel \SpGEMM{} algorithm corresponds to a single partition of $\V$.
For the rest of \cref{sec:parallelLB}, ``algorithm'' means ``parallel \SpGEMM{} algorithm''.

Given a partition of $\V$, the communication for any corresponding algorithm comprises two phases: the \emph{expand} phase, where the processors exchange nonzero entries of $\A$ and $\B$ (initially distributed according to the partitions of $\V^\A$ and $\V^\B$) in order to perform their multiplications (assigned according to the partition of $\V^\mm$), and the \emph{fold} phase, where the processors communicate to reduce partial sums for nonzero entries of $\C$ (finally distributed according to the partition of $\V^\C$).

We now bound below the parallel communication costs of an algorithm in terms of the cuts of $\HG$ induced by that algorithm's associated partition of $\V$.

\begin{definition}
\label[definition]{def:cut}
Given a partition $\{\V_1,\ldots,\V_p\}$ of $\V$, for each $i \in [p]$, $Q_i$ denotes the subset of $\N$ whose elements (nets) have nonempty intersections with both $\V_i$ and $\V \setminus \V_i$.
\end{definition}

\begin{lemma}
\label[lemma]{lem:per-proc-comm}
Given a partition $\{\V_1,\ldots,\V_p\}$ of $\V$, for any associated algorithm, the number of words each processor $i \in [p]$ sends or receives is at least $|Q_i|$, and the critical-path communication cost is at least $\max_{i \in [p]} |Q_i|$.
\end{lemma}
\begin{proof}
For each processor $i \in [p]$, for each net in $Q_i$, processor $i$ must either receive or send the corresponding nonzero, since at most one processor owns each nonzero at the start and end of the computation.
The bound on the critical-path communication cost is obtained by maximizing over $i \in [p]$, as each processor can send only one word at a time along the critical path \cite{BCDH+14}.
\end{proof}
% noindent
Furthermore, for each partition, there exists an algorithm that attains these lower bounds, within constant factors in the cases of the per-processor costs and within a logarithmic factor in the case of the critical-path cost.

\begin{lemma}
\label[lemma]{lem:parUB}
Given a partition $\{\V_1,\ldots,\V_p\}$, there exists an associated algorithm such that the number of words each processor $i\in[p]$ sends or receives is $O(|Q_i|)$, and the critical-path communication cost is $O( \log p \cdot \max_{i \in [p]} |Q_i|)$. 
\end{lemma}
\begin{proof}
The basic idea of the algorithm is as follows.
For the expand phase, every input nonzero that corresponds to a cut net will be sent to the  processors whose parts $\V_i$ intersect the net via a binary-tree broadcast.
For the fold phase, every output nonzero that corresponds to a cut net will be reduced to the processor whose part includes the nonzero vertex via a binary-tree reduction.

In more detail, consider first the expand phase.
All of the broadcasts are performed synchronously: in the first step, each root sends its nonzero to two other processors; in the second step, every processor that receives a nonzero from a root sends it to two other processors in the second step; and so on.
Since at most $p$ processors can be involved in any one broadcast, the expand phase requires at most $O(\log p)$ steps.
Processor $i$ receives each of its required nonzeros at most once and sends each at most twice, for a total (per-processor) cost of $O(|Q_i|)$.
Furthermore, at each expand step, processor $i$ performs one receive and two sends in each of a subset of $|Q_i|$ broadcast trees,  so each expand step involves $O(\max_{i \in [p]} |Q_i|)$ sends or receives along any critical path.
Analysis of the fold phase is symmetric, using binary-tree reductions.
\end{proof}

For efficient algorithms, we constrain the partitions to enforce load balance of memory and computation.
We define sets of load-balanced partitions in terms of two parameters:
\begin{definition}
\label[definition]{def:load-balance}
For any $\delta,\epsilon \in [0,p-1]$, let $\Pi_{\delta,\epsilon}$ be the set of all partitions $\{\V_1,\ldots,\V_p\}$ of $\V$ where 
\[ w_\mem(\V_i) \le (1+\delta) \frac{|\V^\nz|}{p} \qquad\text{and}\qquad w_\comp(\V_i) \le (1+\epsilon)\frac{|\V^\mm|}{p}\]
for each $i\in[p]$. 
We say an algorithm with partition $\{\V_1,\ldots,\V_p\} \in \Pi_{\delta,\epsilon}$ is \emph{$(\delta,\epsilon)$-load balanced}.
\end{definition}
% noindent
Given $\Pi_{\delta,\epsilon}$, if $\delta = 0$ then the memory requirements are perfectly load-balanced, and if $\epsilon = 0$ then the computations (multiplications) are perfectly load-balanced.
If $\delta=\epsilon = p-1$ then $\Pi_{\delta,\epsilon}$ includes the trivial partition wherein no interprocessor communication is required, since one processor stores all three matrices and performs the whole computation. 

Finally, we state our parallel communication lower bound:
\begin{theorem}
\label[theorem]{thm:1}
The critical-path communication cost of an $(\delta,\epsilon)$-load balanced parallel \SpGEMM{} algorithm is at least
\[ \min_{\{\V_1,\dots,\V_p\}\in \Pi_{\delta,\epsilon}} \max_{i \in [p]} |Q_i|\text{.}\]
This lower bound is tight up to at most a logarithmic factor in the number of processors.
\end{theorem}
\begin{proof}
The lower bound follows directly from \cref{lem:per-proc-comm,def:load-balance}.
The tightness of the lower bound follows from applying \cref{lem:parUB} to any optimal partition.
\end{proof}

We compare \cref{thm:1} with the distributed-memory parallel communication lower bounds for \SpGEMM{} derived by Ballard et al.~\citeyear{BDHS11,BDHLS12}.
In our notation, with the additional assumption that each processor is restricted to using $M$ words of its local memory, the critical-path communication cost is bounded below by a function contained in the union of the sets
\begin{equation}
\label[equation]{eq:GLB}
\Omega\left( \frac{|\V^\mm|}{pM^{1/2}} - \alpha M \right) \qquad\text{and}\qquad \Omega\left(\frac{|\V^\mm|^{2/3}}{p^{2/3}} - \beta\frac{|\V^\nz|}{p} \right)   
\end{equation}
for some $\alpha,\beta > 0$. 
These two asymptotic lower bounds are called memory-dependent and memory-independent, respectively; the memory-independent bound requires an additional asymptotic assumption that the computational load balance parameter $\epsilon = O(1)$; we refer to~\cite{BDHS11,BDHLS12} for details regarding these asymptotic expressions, like the constants $\alpha$, $\beta$, and others suppressed by asymptotic notation, and a discussion of when one lower bound dominates the other. 

The combined lower bound \cref{eq:GLB} is tight (within constant factors) when $\A$ and $\B$ are dense matrices~\cite{ITT04}, but  asymptotically loose in expectation when $I=J=K=n$ and each entry of $\A$ and $\B$ is nonzero with probability $d/n$ for any $d$ in the intersection of the sets $\omega(1)$, $o(\sqrt{M})$, and $o(\sqrt{n})$~\cite{BBDG+13}.
On the other hand, by considering the nonzero structures of $\A$ and $\B$, the conclusion of \cref{lem:per-proc-comm} is always tight, within a logarithmic factor. 

\subsection{Sequential Lower Bound}
\label{sec:sequentialLB}

Now we turn to the communication costs of \SpGEMM{} algorithms executed on a sequential machine with a two-level memory.
In this computation model, communication is data movement between a fast memory of $M$-word capacity, $M \ge 3$, and a slow memory of unbounded capacity; operations are only performed on data in fast memory. 
Communication cost is defined as the total number of loads and stores, i.e., words moved from slow to fast and from fast to slow memory. 
We will derive communication lower bounds in terms of the \SpGEMM{} hypergraph model (see \cref{sec:fg-model}.
For the rest of \cref{sec:sequentialLB}, ``algorithm'' means ``sequential \SpGEMM{} algorithm''.

We follow Hong and Kung's approach~\citeyear[Sec.~3]{HK81} for deriving lower bounds on communication in this sequential model, as well as their graph-theoretic \SpGEMM{} model~\citeyear[Sec.~6]{HK81}.
Hong and Kung model algorithms schematically as uninterpreted directed acyclic graphs (DAGs): the \SpGEMM{} hypergraph model in \cref{def:hp-spgemm} is closely related to Hong and Kung's DAG model of \SpGEMM{}.
In particular, each \SpGEMM{} hypergraph $\HG(\A,\B)$ represents a family of \SpGEMM{} DAGs $\mathcal{G}(\A,\B)$ whose members differ only in the order in which they sum the (nontrivial) multiplications. 
After formalizing this relationship, we exploit it to approximate and simplify Hong and Kung's communication lower bounds framework using hypergraphs.

We define a \emph{summation tree} as any DAG with exactly one zero-outdegree vertex, called the \emph{root}; the zero-indegree vertices are called the \emph{leaves}. 
The members of $\mathcal{G}(\A,\B)$ are parameterized by a function that associates each $(i,j) \in S_\C$ with a summation tree $T_{ij}$ whose leaves are $\{v_{ikj} : (\exists k)\, v_{ikj} \in \V\}$ and whose non-leaves are disjoint from $\V$ and the other trees.
Having associated disjoint summation trees with $S_\C$ in this manner, the generic member $(V,E) \in \mathcal{G}(\A,\B)$ is defined by
\begin{align*}
V &= \V^\A \cup \V^\B \cup \left\{V(T_{ij}) : (i,j) \in S_\C \right\} \\
E &= \quad\left\{(v^\A_{ik},v_{ikj}) : (\exists i,k,j)\, v^\A_{ik},v_{ikj} \in \V\right\} \\&\quad\cup\, \left\{(v^\B_{ik},v_{ikj}) : (\exists i,k,j)\, v^\B_{kj},v_{ikj} \in \V \right\} \\&\quad\cup\, \left\{E(T_{ij}) : (i,j) \in S_\C \right\}\text{.}
\end{align*}
Of the vertices $V$, we distinguish the \emph{inputs} $I$, which have no predecessors, and the \emph{outputs} $O$, which have no successors; by construction, these two sets are disjoint.  
We remark that $\mathcal{G}(\A,\B)$ can be specialized to obtain Hong and Kung's \SpGEMM{} DAGs, by restricting each $T_{ij}$ to be a binary tree with at least two leaves; this excludes, e.g., multiplication of diagonal matrices.

The key to Hong and Kung's communication analysis is the set of \emph{$S$-partitions} of a DAG $(V,E)$. 
A partition $\{V_1,\ldots,V_h\}$ of $V$ is an $S$-partition if additionally
\begin{enumerate}
\item Each $V_i$ has a \emph{dominator set} of size at most $S$ --- a dominator set of $V_i$ contains a vertex on every path from $I$ to $V_i$);
\item Each $V_i$ has a \emph{minimum set} of size at most $S$ --- the minimum set of $V_i$ is the set of all elements of $V_i$ with no successors in $V_i$); and
\item There is no cyclic dependence among $V_1,\ldots,V_h$ --- where $V_i$ depends on $V_j$ if $(V_i \times V_j) \cap E \ne \emptyset$.
\end{enumerate}
We exploit $\mathcal{G}(\A,\B)$'s members' common structure to derive lower bounds on dominator and minimum set sizes.
\begin{lemma}
\label[lemma]{lem:HK1}
Consider any $(V,E) \in \mathcal{G}(\A,\B)$ and any $U \subseteq V$.
Let $D_U$ be any dominator set of $U$ and let $M_U$ be the minimum set of $U$.
\begin{align*}
|D_U| &\ge \max\big(|\{(i,k) : (\exists j)\, v_{ikj} \in U\}|,\,|\{(k,j) : (\exists i)\, v_{ikj} \in U\}|\big)\\
|M_U| &\ge |\{(i,j) : (\exists k)\, v_{ikj} \in U\}|
\end{align*}
\end{lemma}
\begin{proof} 
There exists a set of disjoint paths, each starting from a distinct element of $\{v^\A_{ik} : (\exists j)\, v_{ikj} \in U\} \subset I$ and ending in $U$.
Similarly, there exists a set of disjoint paths, each starting from a distinct element of $\{v^\B_{kj} : (\exists i)\, v_{ikj} \in U\} \subset I$ and ending in $U$.  
By definition, $D_U$ must contain at least one vertex from every path in each of the two sets, so the first lower bound follows.
Now observe that for each pair $(i,j)$ such that $v_{ikj} \in U$ for some $k \in [K]$, $T_{ij} \cap U$ is nonempty and, since the summation trees are disjoint, has at least one element with no successors in $U$.
Thus each such $T_{ij}$ contributes at least one element to $M_U$, establishing the second lower bound.
\end{proof}
% \noindent 
We now relate the DAGs $\mathcal{G}(\A,\B)$ and the hypergraph $\HG(\A,\B)$.
\begin{definition}
\label[definition]{def:Wsets}
Consider any partition $\{\V_1,\ldots,\V_h\}$ of $\V$.
For each part $\V_i$, let $W^\A_i$, $W^\B_i$, and $W^\C_i$ be the subsets of $\N^\A$, $\N^\B$, and $\N^\C$, resp., having nonempty intersections with $\V_i$.
\end{definition}
\begin{lemma}
\label[lemma]{lem:HK2}
Consider any $G \in \mathcal{G}(\A,\B)$ and $S \in \{1,2,\ldots\}$.
For any $S$-partition $\{V_1,\ldots,V_h\}$ of $G$, there exists a partition $\{\V_1,\ldots,\V_h\}$ of $\V$ with $|W^\A_i|,|W^\B_i|,|W^\C_i| \le S$ for each $i \in [h]$. 
\end{lemma}
\begin{proof}
We construct $\{\V_1,\ldots,\V_h\}$ from $\{V_1,\ldots,V_h\}$ as follows.
Define a partition $\{V'_1,\ldots,V'_h\}$ of $\V^\mm$ by intersecting each subset $V_1,\ldots,V_h$ with $\V^\mm$.
Every element of $\N$ has a nontrivial intersection with at least one subset $V'_1,\ldots,V'_h$.
Moreover, applying \cref{lem:HK1} to each subset $V'_i$, we have that $|W'^\A_i|,|W'^\B_i|,|W'^\C_i| \le S$.
Now add each $v^\A_{ik} \in \V^\A$ to an arbitrarily chosen $V'_i$ that intersects the corresponding net $n^\A_{ik}$, and likewise for each $v^\B_{kj} \in \V^\B$ and each $v^\C_{ij} \in \V^\C$.
In this manner, we obtain a partition $\{\V_1,\ldots,\V_h\}$ of $\V$.
Moreover, augmenting $V'_1,\ldots,V'_h$ with elements of $\V^\nz$ in this manner does not change the intersected nets, i.e., $W^\A_i = W'^\A_i$, $W^\B_i = W'^\B_i$, and $W^\C_i = W'^\C_i$ for each $(\V_i,V'_i)$ pair.
\end{proof}

\begin{lemma}
\label[lemma]{lem:HK-attain}
Consider any $S \in \{1,2,\ldots\}$ and partition $\{\V_1,\ldots,\V_h\}$ of $\V$ with $|W^\A_i|,|W^\B_i|,|W^\C_i| \le S$ for each $i \in [h]$.
There exists an algorithm with communication cost $O(Mh(S/M+1)^3)$.
\end{lemma}
\begin{proof}
For each $i\in [h]$, partition $W^\A_i,W^\B_i,W^\C_i$ so that each part has size $m=\lfloor M/3 \rfloor$ except possibly for one smaller part.
These three partitions naturally induce a partition $\{\U_1,\ldots,\U_g\}$ of $\V$, whose parts are called \emph{blocks}.
Observe that $\{\U_1,\ldots,\U_g\}$ refines $\{\V_1,\ldots,\V_h\}$; in particular, each part $\V_i$ is partitioned into to at most $\lceil S/m \rceil^3$ blocks.
Thus, the total number of blocks $g \le h\lceil S/m \rceil^3$.

The algorithm considers $\U_j$ for each $j \in [g]$:
\begin{enumerate}
\item Load $\A$- and $\B$-matrix entries associated with $W^\A_{j}$ and $W^\B_{j}$, as well as any partial sums of $\C$-matrix entries $W^\C_{j}$ previously computed --- at most $3m$ loads.
\item Perform all possible nontrivial  multiplications and additions --- this can be done with no data movement by processing each $\C$-matrix entry in sequence.   
\item Store the updated $\C$-matrix entries --- at most $m$ stores.
\end{enumerate}
There are at most $4mg \le 4Mh(5S/M+1)^3/3$ loads and stores in total.
\end{proof}

Finally, we apply these results within Hong and Kung's communication lower bounds framework. 
\begin{theorem}
\label[theorem]{thm:HK}
The communication cost of an algorithm is at least $M(h-1)$, where $h$ is the minimum cardinality of a partition $\{\V_1,\ldots,\V_h\}$ of $\V$ with $|W^\A_i|,|W^\B_i|,|W^\C_i| \le 2M$ for each $i \in [h]$. 
When $h > 1$, this lower bound is tight up to a constant factor.
\end{theorem}
\begin{proof}
Any algorithm can be represented in Hong and Kung's model by some $G \in \mathcal{G}(\A,\B)$.
By Hung and Kung's argument~\citeyear[Lem.~3.1]{HK81}, the communication cost is at least $M(P(2M)-1)$, where $P(S)$ denotes the minimum cardinality of an $S$-partition of $G$.
We apply \cref{lem:HK2} to bound $P(2M)$ below by $h$.
Tightness follows from applying \cref{lem:HK-attain} with $S=2M$.
\end{proof}
% noindent
\cref{thm:HK} simplifies Hong and Kung's key lemma by requiring a single hypergraph partitioning problem, vs.\ $|\mathcal{G}(\A,\B)|$ graph partitioning problems.
Additionally, Hong and Kung did not address attainability of their key lemma.

We note that another of Hong and Kung's results~\citeyear[Thm.~6.1]{HK81} yields the lower bound $\Omega(|\V^\mm|/M^{1/2})$, a sequential analogue of the memory-dependent bound in \cref{eq:GLB}.
While this lower bound is attainable in the cases where $\A$ and $\B$ are both dense (see, e.g., \cite{BDHS11}), it is asymptotically loose in many sparse cases.
For example, if $I=J=K=n$ and $S_\A=S_\B=\{(i,i) : i \in [n]\}$, then any algorithm requires moving at least $3|\V^\mm|$ words between the two levels, the nonzero entries of $\A$, $\B$, and $\C$. 
However, a trivial lower bound of $|\V^\nz|$ also exists, under the assumption that fast memory must be empty before and after the computation.
It remains open to show whether or not the combination of this bound and the memory-dependent one is asymptotically loose.

We mention some other communication bounds that have appeared previously.
Pagh and St\"{o}ckel~\citeyear{PS14} proved matching worst-case upper and lower bounds based on the number of nonzeros in the input and output matrices.
That is, given a number of nonzeros, they showed that there exist input matrices that require a certain amount of communication; they also gave an algorithm that never requires more than a constant factor times this worst-case cost.
Greiner~\citeyear[Ch.~6]{Greiner12} considered restricted classes of algorithms and matrices and used a different approach to establish other worst-case lower and upper bounds, which are not always tight.
In contrast to these results, \cref{thm:HK} is input specific rather than worst case.

\section{Restricting \SpGEMM{} Algorithms}
\label{sec:simplify-HG}

Each \SpGEMM{} hypergraph (defined in \cref{sec:fg-model}) models a class of \SpGEMM{} algorithms via partitions of its vertices.
This hypergraph model can be simplified in the study of subclasses of algorithms and matrix nonzero structures. 
Here, we only consider the distributed-memory parallel case, as described in \cref{sec:parallelLB}.

In \cref{sec:coarsening}, we describe how to simplify \SpGEMM{} hypergraphs using a vertex coarsening technique while still correctly modeling algorithmic costs.
We use vertex coarsening to model subclasses of algorithms, restricting the parallelization (\cref{sec:restrict-par}), the data distribution (\cref{sec:restrict-distr}), and both (\cref{sec:restrict-alg}).
We can also use vertex coarsening when modeling subclasses of matrix nonzero structures, as we demonstrate in \cref{sec:SpMV} in the case of \emph{sparse matrix-vector multiplication} (\SpMV{}).
Lastly, in \cref{sec:generalize-HG}, we show how vertex coarsening can also be used to model \SpGEMM{}-like algorithms that violate two of the assumptions in \cref{sec:notation}, exploiting relations between the input matrix entries (\cref{sec:input-relations}) and only computing a subset of the output matrix entries (\cref{sec:output-sparsity}).

\subsection{Vertex Coarsening}
\label{sec:coarsening}

We can restrict the class of algorithms modeled by an \SpGEMM{} hypergraph by forcing certain subsets of vertices to be \emph{monochrome}, or all assigned to the same part, in the partition.
For example, we can enforce that a certain subset of multiplications will be performed by the same processor or that a particular  multiplication is performed by the processor that owns a corresponding nonzero entry.
To implement such restrictions in an \SpGEMM{} hypergraph, we \emph{coarsen} vertices, choosing a new vertex to represent each subset of \emph{constituent} vertices that must be monochrome. 
We will use the symbol ``$\sim$'' to mean ``represents'' in this sense.

After vertex coarsening, in order to model communication costs correctly, we define net membership as follows: a coarsened vertex is a member of a net if any of its constituent vertices was a member of that net.
This is because the computation corresponding to the coarsened vertex requires the matrix entries associated with any nets that contained a constituent vertex.

In order to model computational and memory costs correctly, the weights of a coarsened vertex equal the sum of the constituent vertices' weights.
This is because the processor assigned the coarsened vertex will perform the computation corresponding to all of the multiplication constituent vertices and will own all of the nonzero constituent vertices.

After coarsening the vertices and updating net memberships, there may be multiple \emph{coalesced} nets that contain identical sets of vertices, and there may be \emph{singleton} nets that contain only a single vertex.
To reduce the number of nets, we can combine coalesced nets into one coarsened net; the cost of the coarsened net equals the sum of the coalesced nets' costs.
Because singleton nets cannot be cut, we can omit them, a further simplification. 

\subsection{Restricted Parallelizations}
\label{sec:restrict-par}
We first apply vertex coarsening to the multiplication vertices $\V^\mm \subset \V$ of an \SpGEMM{} hypergraph $(\V,\N) = \HG(\A,\B)$, representing parallelizing computation across processors at a coarser granularity than individual  multiplications.
We represent each coarsened set by a new vertex $\hat v$, and let $\hat \V^\mm$ denote the set of these new $\hat v$ vertices. 
In this section, we leave the remaining vertices $\V^\nz$ untouched, so $\hat \V = \hat \V^\mm \cup \V^\nz$; we will address restricting data distribution subsequently in \cref{sec:restrict-distr,sec:restrict-alg}.
 Net memberships are updated in terms of the $\hat v$ vertices as described in \cref{sec:coarsening}, i.e., $\N$ induces a set $\hat \N$ of nets on $\hat \V$, thus obtaining the simplified hypergraph $(\hat \V,\hat \N)$.
Since each net contains a distinct nonzero vertex, no nets are coalesced/combined, so $|\hat \N| = |\N|$ and the nets' costs remain unit.
Additionally, the nonzero vertices are untouched, so the only weights that change are the computational weights of elements $\hat v$ of $\hat \V^\mm$: each $w_\comp(\hat v)$ equals the number of multiplication vertices it combined.   

The simplified hypergraph has fewer vertices than the original: $\V^\mm$ is replaced by $\hat \V^\mm$, whose cardinality is that of the given partition of $\V^\mm$. 
Additionally, while the number of nets is unchanged, the number of pins generally decreases.

In the remainder of \cref{sec:restrict-par}, we will study two types of coarsenings of $\V^\mm$: by slice and by fiber.
\Cref{fig:slicefiber} shows visualizations of each type of coarsening.

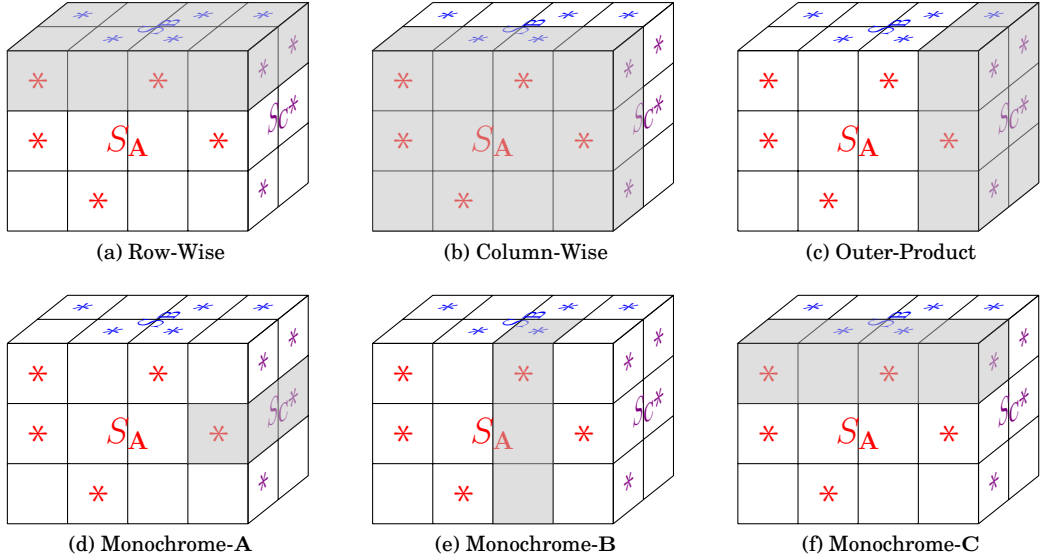
\begin{figure}
\newcommand{\sfscale}{.8}
\tikzstyle{vc} = [fill=gray!50,opacity=.5]

\begin{center}
\subfloat[][Row-Wise]{
	\begin{tikzpicture}[CCorient,scale=\sfscale]

	\compcube

	% draw shading
	\begin{scope}[Aface]
		\draw[vc] (-.5,.5) rectangle (\matdimK-.5,-.5);
	\end{scope}
	\begin{scope}[Bface]
		\draw[vc] (-.5,.5) rectangle (\matdimJ-.5,-\matdimK+.5);
	\end{scope}
	\begin{scope}[Cface]
		\draw[vc] (-.5,.5) rectangle (\matdimJ-.5,-.5);
	\end{scope}

	\end{tikzpicture}
	\label{fig:sf:row}
} \hfill
\subfloat[][Column-Wise]{
	\begin{tikzpicture}[CCorient,scale=\sfscale]

	\compcube

	% draw shading
	\begin{scope}[Aface]
		\draw[vc] (-.5,.5) rectangle (\matdimK-.5,-\matdimI+.5);
	\end{scope}
	\begin{scope}[Bface]
		\draw[vc] (-.5,.5) rectangle (.5,-\matdimK+.5);
	\end{scope}
	\begin{scope}[Cface]
		\draw[vc] (-.5,.5) rectangle (.5,-\matdimI+.5);
	\end{scope}

	\end{tikzpicture}
	\label{fig:sf:col}
} \hfill
\subfloat[][Outer-Product]{
	\begin{tikzpicture}[CCorient,scale=\sfscale]

	\compcube

	% draw shading
	\begin{scope}[Aface]
		\draw[vc] (\matdimK-1.5,.5) rectangle (\matdimK-.5,-\matdimI+.5);
	\end{scope}
	\begin{scope}[Bface]
		\draw[vc] (-.5,-\matdimK+1.5) rectangle (\matdimJ-.5,-\matdimK+.5);
	\end{scope}
	\begin{scope}[Cface]
		\draw[vc] (-.5,.5) rectangle (\matdimJ-.5,-\matdimI+.5);
	\end{scope}

	\end{tikzpicture}
	\label{fig:sf:out}
} \\
\subfloat[][Monochrome-$\A$]{
	\begin{tikzpicture}[CCorient,scale=\sfscale]

	\compcube

	% draw shading (specific to nonzero structure)
	\begin{scope}[Aface]
		\draw[vc] (\matdimK-1.5,-.5) rectangle (\matdimK-.5,-1.5);
	\end{scope}
	\begin{scope}[Cface]
		\draw[vc] (-.5,-.5) rectangle (\matdimJ-.5,-1.5);
	\end{scope}

	\end{tikzpicture}
	\label{fig:sf:monoa}
} \hfill
\subfloat[][Monochrome-$\B$]{
	\begin{tikzpicture}[CCorient,scale=\sfscale]

	\compcube

	% draw shading (specific to nonzero structure)
	\begin{scope}[Aface]
		\draw[vc] (1.5,.5) rectangle (2.5,-\matdimI+.5);
	\end{scope}
	\begin{scope}[Bface]
		\draw[vc] (-.5,-1.5) rectangle (.5,-2.5);
	\end{scope}

	\end{tikzpicture}
	\label{fig:sf:monob}
} \hfill
\subfloat[][Monochrome-$\C$]{
	\begin{tikzpicture}[CCorient,scale=\sfscale]

	\compcube

	% draw shading (specific to nonzero structure)
	\begin{scope}[Aface]
		\draw[vc] (-.5,.5) rectangle (\matdimK-.5,-.5);
	\end{scope}
	\begin{scope}[Bface]
		\draw[vc] (-.5,.5) rectangle (.5,-\matdimK+.5);
	\end{scope}
	\begin{scope}[Cface]
		\draw[vc] (-.5,.5) rectangle (.5,-.5);
	\end{scope}

	\end{tikzpicture}
	\label{fig:sf:monoc}
}
\end{center}
\caption{Visualizations of restricted parallelizations.  The top row (\cref{fig:sf:row,fig:sf:col,fig:sf:out}) corresponds to slice-wise coarsenings or 1D parallelizations; the bottom row (\cref{fig:sf:monoa,fig:sf:monob,fig:sf:monoc}) corresponds to fiber-wise coarsenings or 2D parallelizations.  In each figure, a particular coarsened vertex is highlighted as a shaded set of cubes; the coarsened vertex corresponds to all multiplication vertices in the shaded set.}
\label{fig:slicefiber}
\end{figure}

Slice-wise coarsenings, which model what we call \emph{1D parallelizations}, coarsen the \cc{} by slices parallel to one of the three coordinate planes.
Coarsening by slices means we identify all $v_{ikj}$ with the same $i$, $j$, or $k$ index --- the choice of index gives rise to three 1D models, called \emph{row-wise}, \emph{column-wise}, and \emph{outer-product}: 
\[ \begin{array}{l r c l r c l}
\text{row-wise} & \hat \V^\mm & =  & \{ \hat v_i : i \in [I] \}, & \hat v_i & \sim & \{ v_{its} : (\exists t,s)\, v_{its} \in \V \} \\ 
\text{column-wise} & \hat \V^\mm & =  & \{ \hat v_j : i \in [J] \}, & \hat v_j & \sim & \{ v_{rtj} : (\exists r,t)\, v_{rtj} \in \V \} \\ 
\text{outer-product} & \hat \V^\mm & =  & \{ \hat v_k : k \in [K] \}, & \hat v_k & \sim & \{ v_{rks} : (\exists r,s)\, v_{rks} \in \V \}.
\end{array}\]

Fiber-wise coarsenings, which model what we call \emph{2D parallelizations}, coarsen the \cc{} by fibers parallel to one of the three coordinate axes.
Coarsening by fibers means we identify all $v_{ikj}$ with the same pair $(i,k)$, $(k,j)$, or $(i,j)$, of indices --- this choice gives rise to three 2D models, called \emph{monochrome-$\A$}, \emph{monochrome-$\B$}, and \emph{monochrome-$\C$}: 
\[ \begin{array}{l r c l r c l}
\text{monochrome-$\A$} & \hat \V^\mm & =  & \{ \hat v_{ik} : (i,k) \in S_\A \}, & \hat v_{ik} & \sim & \{ v_{iks} : (\exists s)\, v_{iks} \in \V \} \\ 
\text{monochrome-$\B$} & \hat \V^\mm & =  & \{ \hat v_{kj} : (k,j) \in S_\B \}, & \hat v_{kj} & \sim & \{ v_{rkj} : (\exists r)\, v_{rkj} \in \V \} \\ 
\text{monochrome-$\C$} & \hat \V^\mm & =  & \{ \hat v_{ij} : (i,j) \in S_\C \}, & \hat v_{ij} & \sim & \{ v_{itj} : (\exists t)\, v_{itj} \in \V \}.
\end{array}\]

Each model is identified with a set of parallelizations. 
In a row-wise parallelization, every $\B$-slice is monochrome, and the parallelization corresponds to a partition of the row-wise model; column-wise and outer-product parallelizations are defined similarly.
In a monochrome-$\A$ parallelization, every $\A$-fiber is monochrome, and the parallelization corresponds to a partition of the monochrome-$\A$ model; monochrome-$\B$ and -$\C$ parallelizations are defined similarly.
A general parallelization is called fine-grained.

We now discuss relationships among these seven classes of parallelizations, as illustrated in \cref{fig:venn}.
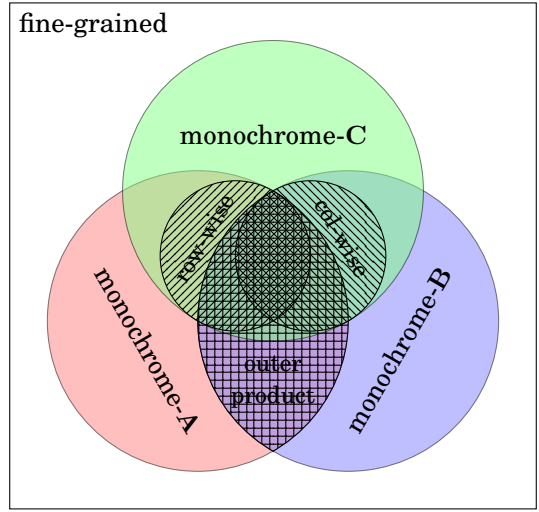
\begin{figure}
\newcommand{\sqrtthree}{1.732}
\newcommand{\fgborder}{.5}
\newcommand{\lblshift}{.5}

\begin{center}
\begin{tikzpicture}[]

	% DRAW 2D REGIONS TRANSLUCENT
	% fine-grained
	\draw (-3-\fgborder,-2-\fgborder) rectangle (3+\fgborder,2+\sqrtthree+\fgborder);
	% monochrome-A
	\draw[fill=red!50,opacity=.5] (-1,0) circle (2); 
	% monochrome-B
	\draw[fill=blue!50,opacity=.5] (1,0) circle (2);
	% monochrome-C
	\draw[fill=green!50,opacity=.5] (0,\sqrtthree) circle (2);
	
	% DRAW 1D REGIONS IN OUTLINE
	% outer-product
	\begin{scope}
		\clip (0,-2) rectangle (3,2);
		\draw[pattern=grid] (-1,0) circle (2); 
	\end{scope} 
	\begin{scope}
		\clip (0,-2) rectangle (-3,2);
		\draw[pattern=grid] (1,0) circle (2); 
	\end{scope} 
	% row-wise
	\draw[pattern=north east lines] (-.5,\sqrtthree/2) circle (1);
	% col-wise
	\draw[pattern=north west lines] (.5,\sqrtthree/2) circle (1);

	% DRAW LABELS
	% fine-grained
	\node[below right] at (-3-\fgborder,2+\sqrtthree+\fgborder) {fine-grained};
	% 2D 	
	\node[rotate=-60,below,shift={(0,-\lblshift)}] at (-1,0) {monochrome-$\A$};
	\node[rotate=60,below,shift={(0,-\lblshift)}] at (1,0) {monochrome-$\B$};
	\node[above,shift={(0,\lblshift)}] at (0,\sqrtthree) {monochrome-$\C$};
	% 1D
	\node[below,shift={(0,-\lblshift/2)}] at (0,0) {\small \begin{tabular}{c} outer \\ product \end{tabular}};
	\node[rotate=60,above,shift={(0,\lblshift/2)}] at (-.5,\sqrtthree/2) {\small row-wise};
	\node[rotate=-60,above,shift={(0,\lblshift/2)}] at (.5,\sqrtthree/2) {\small col-wise};
      
\end{tikzpicture}
\end{center}
\caption{Venn diagram showing relationships among restricted parallelizations.} %
\label{fig:venn}
\end{figure}
Let $\PF$ denote the set of all (fine-grained) parallelizations.
Let $\PR,\PL,\PU,\PA,\PB,\PC \subseteq \PF$ denote the row-wise, column-wise, outer-product, monochrome-$\A$, -$\B$, and -$\C$ parallelizations, respectively. 
\cref{fig:venn} depicts the inclusions that hold among these 7 sets.
Notice in particular that these 7 sets define a 13-way partition of $\PF$, whose parts are listed in the first column of \cref{tbl:venn};
%\begin{gather*}
%\PF \setminus (\PA \cup \PB \cup \PC) \text{,}\quad
%\PA \setminus (\PB \cup \PC) \text{,}\quad
%\PB \setminus (\PA \cup \PC) \text{,}\quad
%\PC \setminus (\PA \cup \PB) \text{,}\\
%(\PA \cap \PB)\setminus \PC \text{,}\quad
%((\PA \cap \PC) \setminus \PB)\cap \PR \text{,}\quad
%((\PB \cap \PC) \setminus \PA)\cap \PL \text{,}\\
%\PA \cap \PB \cap \PC \cap \PR \cap \PL \text{,}\quad
%(\PA \cap \PB \cap \PC \cap \PR) \setminus \PL \text{,}\quad
%(\PA \cap \PB \cap \PC \cap \PL) \setminus \PR \text{,}\\
%((\PB \cap \PC) \setminus \PA)\setminus \PL\text{,}\quad
%((\PA \cap \PC) \setminus \PB)\setminus \PR\text{,}\quad
%(\PA \cap \PB \cap \PC) \setminus (\PR \cup \PL) \text{.}
%\end{gather*}
%%
we now verify that these 13 parts indeed form a partition of $\PF$. 
Observe that $\PR \subseteq \PA \cap \PC$, $\PL \subseteq \PB \cap \PC$, and $\PU \subseteq \PA \cap \PB$.
Actually, $\PU = \PA \cap \PB$.
To see the converse inclusion, consider any $P \in \PA\cap\PB$, and observe that every $\C$-slice must be monochrome in $P$;
this is because $\C$-slices are Cartesian products of their projections onto the $\A$- and $\B$-faces, but a similar property does not generally hold for the $\A$- and $\B$-slices.
The preceding inclusions suffice to demonstrate correctness of the 13-way partition of $\PF$ in \cref{fig:venn}.

\cref{tbl:venn} uses the following four \SpGEMM{} instances to show that each of the 13 parts is not empty.
\begin{align}
\label{eq:mm1}
\left[\begin{array}{cc}
\nzsymm & \nzsymm\\
\nzsymm & \nzsymm
\end{array}\right]&=\left[\begin{array}{cc}
\nzsymm & \nzsymm\\
\nzsymm & \nzsymm
\end{array}\right]\left[\begin{array}{cc}
\nzsymm & \nzsymm \\
\nzsymm & \nzsymm
\end{array}\right]\\
\label{eq:mm2}
\left[\begin{array}{cc}
\nzsymm & \nzsymm\\
\nzsymm & \nzsymm
\end{array}\right]&=\left[\begin{array}{cc}
\nzsymm \\
 & \nzsymm
\end{array}\right]\left[\begin{array}{cc}
\nzsymm & \nzsymm \\
\nzsymm & \nzsymm
\end{array}\right]\\
\label{eq:mm3}
\left[\begin{array}{cc}
\nzsymm & \nzsymm\\
\nzsymm & \nzsymm
\end{array}\right]&=\left[\begin{array}{cc}
\nzsymm & \nzsymm\\
\nzsymm & \nzsymm
\end{array}\right]\left[\begin{array}{cc}
\nzsymm \\
& \nzsymm
\end{array}\right]\\
\label{eq:mm4}
\left[\begin{array}{cc}
\nzsymm & \nzsymm\\
\nzsymm & \nzsymm
\end{array}\right]&=\left[\begin{array}{cccc}
\nzsymm &  & \nzsymm\\
 & \nzsymm &  & \nzsymm
\end{array}\right]\left[\begin{array}{cc}
\nzsymm\\
\nzsymm\\
 & \nzsymm\\
 & \nzsymm
\end{array}\right]
\end{align}

\begin{table}%
\tbl{%
\label{tbl:venn}%
Examples showing that each of the 13 parts in \cref{fig:venn} is nonempty.
The ``finest'' parallelization is where each nontrivial multiplication is assigned to a distinct processor.
A parallelization ``by fiber'' means that each fiber (of a given type) is assigned to a distinct processor.
A parallelization ``by slice'' means that each slice (of a given type) is assigned to a distinct processor.
The ``coarsest'' parallelization is where all nontrivial multiplications are assigned to the same processor.
}%
{%
\begin{tabular}{ c c c }
\toprule
Part & \SpGEMM{} Instance & Parallelization \\
\midrule
$\PF \setminus (\PA \cup \PB \cup \PC)$ & \cref{eq:mm1} & finest  \\
$\PA \setminus (\PB \cup \PC)$ & \cref{eq:mm1} & by $\A$-fiber \\
$\PB \setminus (\PA \cup \PC)$ & \cref{eq:mm1} & by $\B$-fiber \\
$\PC \setminus (\PA \cup \PB)$ & \cref{eq:mm1} & by $\C$-fiber  \\
$((\PB \cap \PC) \setminus \PA)\cap \PL$ & \cref{eq:mm1} & by $\A$-slice \\
$((\PA \cap \PC) \setminus \PB)\cap \PR$ & \cref{eq:mm1} & by $\B$-slice \\
$(\PA \cap \PB)\setminus \PC$ &\cref{eq:mm1}  & by $\C$-slice \\
$\PA \cap \PB \cap \PC \cap \PR \cap \PL$ & \cref{eq:mm1} & coarsest \\
$((\PB \cap \PC) \setminus \PA)\setminus \PL$ & \cref{eq:mm2} & finest \\
$(\PA \cap \PB \cap \PC \cap \PR) \setminus \PL$ & \cref{eq:mm2} & by $\A$-fiber \\
$((\PA \cap \PC) \setminus \PB)\setminus \PR$ & \cref{eq:mm3} & finest \\
$(\PA \cap \PB \cap \PC \cap \PL) \setminus \PR$ & \cref{eq:mm3} & by $\B$-fiber \\
$(\PA \cap \PB \cap \PC) \setminus (\PR \cup \PL)$ & \cref{eq:mm4} & finest \\
\bottomrule
\end{tabular}
}%
\end{table}

\subsection{Restricted Data Distributions}
\label{sec:restrict-distr}

Now we apply vertex coarsening to the nonzero vertices $\V^\nz \subset \V$ of an \SpGEMM{} hypergraph $(\V,\N) = \HG(\A,\B)$, representing distributing data across processors at a coarser granularity than individual matrix entries. 
We represent each coarsened set by a new vertex $\hat v$, and let $\hat \V^\nz$ be the set of these new $\hat v$ vertices. 
In this section, we leave the remaining vertices $\V^\mm$ untouched, so $\hat \V = \V^\mm \cup \hat \V^\nz$; we addressed parallelizations previously in \cref{sec:restrict-par} and will synthesize the two approaches subsequently in \cref{sec:restrict-alg}.
Net memberships are updated in terms of the $\hat v$ vertices, i.e., $\N$ induces a set $\hat \N$ of nets on $\hat \V$, thus obtaining the simplified hypergraph $(\hat \V,\hat \N)$.
Now, it is possible that some nets are coalesced, and thus are combined in $\hat \N$; the cost of each net in $\hat \N$ is the number of coalesced nets it combines.  
Additionally, the multiplication vertices are untouched, so the only weights that change are the memory weights of elements $\hat v$ of $\hat \V^\nz$: each $w_\mem(\hat v)$ equals the number of nonzero vertices it combines.   

The simplified hypergraph has fewer vertices than the original: $\V^\nz$ is replaced by $\hat \V^\nz$, whose cardinality is that of the given partition of $\V^\nz$. 

In the remainder of \cref{sec:restrict-par}, we will study two classes of data distributions, \emph{row-wise} and \emph{column-wise}, meaning the matrix nonzeros are distributed by rows or columns.
We coarsen $\V^\A$, $\V^\B$, and $\V^\C$ separately, leading to sets of coarsened vertices $\hat \V^\A$, $\hat \V^\B$, and $\hat \V^\C$:
\[ \begin{array}{l r c l r c l}
\text{$\A$ row-wise} & \hat \V^\A & =  & \{ \hat v^\A_i : i \in [I] \}, & \hat v^\A_i & \sim & \{ v^\A_{it} : (\exists t)\, v^\A_{it} \in \V \} \\ 
\text{$\A$ column-wise} & \hat \V^\A & =  & \{ \hat v^\A_k : k \in [K] \}, & \hat v^\A_k & \sim & \{ v^\A_{rk} : (\exists r)\, v^\A_{rk} \in \V \} \\ 
\text{$\B$ row-wise} & \hat \V^\B & =  & \{ \hat v^\B_k : k \in [K] \}, & \hat v^\B_k & \sim & \{ v^\B_{ks} : (\exists s)\, v^\B_{ks} \in \V \} \\ 
\text{$\B$ column-wise} & \hat \V^\B & =  & \{ \hat v^\B_j : j \in [J] \}, & \hat v^\B_j & \sim & \{ v^\B_{tj} : (\exists t)\, v^\B_{tj} \in \V  \} \\ 
\text{$\C$ row-wise} & \hat \V^\C & =  & \{ \hat v^\C_i : i \in [I] \}, & \hat v^\C_i & \sim & \{ v^\C_{is} : (\exists s)\, v^\C_{is} \in \V  \} \\ 
\text{$\C$ column-wise} & \hat \V^\C & =  & \{ \hat v^\C_j : j \in [J] \}, & \hat v^\C_j & \sim & \{ v^\C_{rj} : (\exists r)\, v^\C_{rj} \in \V \}.
\end{array}\]
The unrestricted choices $\hat \V^\A = \V^\A$, $\hat \V^\B = \V^\B$, and $\hat \V^\C = \V^\C$, are called \emph{fine-grained}.
Having picked a distribution for each of $\A$, $\B$, and $\C$, we assemble $\hat \V^\nz = \hat \V^\A \cup \hat \V^\V \cup \hat \V^\C$.

We remark here that some choices of these distributions naturally correspond to the 1D parallelization schemes.
That is, a row-wise parallelization corresponds to a matching distribution of the rows of $\A$ and $\C$, a column-wise parallelization corresponds to a matching distribution of the columns of $\B$ and $\C$, and an outer-product parallelization corresponds to a matching distribution of the columns of $\A$ and rows of $\B$.
We revisit this connection next, in \cref{sec:restrict-alg}, when we consider pairing parallelizations and data distributions. 

\subsection{Restricted Algorithms}
\label{sec:restrict-alg}
In \cref{sec:restrict-par,sec:restrict-distr}, we considered vertex coarsening using partitions of $\V^\mm$ or $\V^\nz$.
Now we extend this idea to partitions of $\V$ whose parts may combine vertices from both $\V^\mm$ and $\V^\nz$.
This models the algorithmic constraint where a processor who performs a certain set of multiplications must also store a certain set of matrix entries.

Revisiting the parallelizations in \cref{sec:restrict-par}, when coarsening $\V^\mm$ by fiber, there is a natural choice of a nonzero vertex to associate with each fiber; likewise, when coarsening $\V^\mm$ by slice, there is a natural set of nonzero vertices to associate with each slice.
More generally, if all the multiplication vertices associated with a particular set of nonzero vertices are to be coarsened, it is natural to include those nonzero vertices.
(However, this assignment does impose additional constraints on load balance, which may be unnatural in certain applications.)
An expression $(xyz)$, where $x,y,z \in \{\text{r,c,f,R,C,F}\}$, indicates the data distributions (row-wise, column-wise, or fine-grained) used for $\A$, $\B$, and $\C$. Capital letters indicate that the data distribution is a natural choice, in the sense just described, and the corresponding nonzero vertices have been absorbed into coarsened vertices. 
We highlight the following coarsenings:
\[\begin{array}{l r c l r c l}
\text{row-wise (RfR):} & \hat \V & = & \bigcup\left\{ \begin{gathered} \{ \hat v_i : i \in [I] \} ,\\ \V^\B\end{gathered}\right\}, & \hat v_i & \sim & \bigcup\left\{\begin{gathered} \{ v_{its} : (\exists t,s)\, v_{its} \in \V\} ,\\ \{ v^\A_{it} : (\exists t)\, v^\A_{it} \in \V\} ,\\ \{ v^\C_{is} : (\exists s)\, v^\C_{is} \in \V \}\end{gathered}\right\}
\\
\text{column-wise (fCC):} & \hat \V  & = & \bigcup\left\{\begin{gathered} \{ \hat v_j : j \in [J] \} ,\\ \V^\A\end{gathered}\right\}, & \hat v_j & \sim & \bigcup\left\{\begin{gathered} \{ v_{rtj} : (\exists r,t)\, v_{rtj} \in \V\} ,\\ \{ v^\B_{tj} : (\exists t)\, v^\B_{tj} \in \V\} ,\\ \{ v^\C_{rj} : (\exists r)\, v^\C_{rj} \in \V \}\end{gathered}\right\} 
\\
\text{outer-product (CRf):} & \hat \V & = & \bigcup\left\{\begin{gathered} \{ \hat v_k : k \in [K] \} ,\\ \V^\C\end{gathered}\right\}, & \hat v_k & \sim & \bigcup\left\{\begin{gathered}\{ v_{rks} : (\exists r,s)\, v_{rks} \in \V\} ,\\ \{ v^\A_{rk} : (\exists r)\, v^\A_{rk} \in \V\} ,\\ \{ v^\B_{ks} : (\exists s)\, v^\B_{ks} \in \V \}\end{gathered}\right\} 
\end{array}\]
\[\begin{array}{l r c l r c l}
\text{monochrome-$\A$ (Fff):} & \hat \V  & = & \bigcup\left\{\begin{gathered} \{ \hat v_{ik} : (i,k) \in S_\A \} ,\\ \V^\B ,\; \V^\C\end{gathered}\right\}, & \hat v_{ik} & \sim &
\bigcup\left\{\begin{gathered} \{ v_{iks} : (\exists s)\, v_{iks} \in \V\} ,\\ \{ v^\A_{ik} \} \end{gathered}\right\} 
\\
\text{monochrome-$\B$ (fFf):} & \hat \V  & = & \bigcup\left\{\begin{gathered} \{ \hat v_{kj} : (k,j) \in S_\B \} ,\\ \V^\A ,\; \V^\C\end{gathered}\right\}, & \hat v_{kj} & \sim &
\bigcup\left\{\begin{gathered}\{ v_{rkj} : (\exists r)\, v_{rkj} \in \V\} ,\\ \{ v^\B_{kj} \} \end{gathered}\right\}
\\
\text{monochrome-$\C$ (ffF):} & \hat \V & = & \bigcup\left\{\begin{gathered}\{ \hat v_{ij} : (i,j) \in S_\C \} ,\\ \V^\A ,\; \V^\B\end{gathered}\right\}, & \hat v_{ij} & \sim &
\bigcup\left\{\begin{gathered}\{ v_{itj} : (\exists t)\, v_{itj} \in \V\} ,\\ \{ v^\C_{ij} \}\end{gathered}\right\}
\end{array}\]
% noindent
The coarsened vertices now have both nonzero computational and memory weights, indicating the number of multiplication and nonzero vertices, respectively, that they combined.
Note that the nets corresponding to the coarsened nonzero vertices become singletons, and are therefore omitted from the simplified hypergraph --- in the six cases above, the nets induced by $\N^\A \cup \N^\C$, $\N^\B\cup\N^\C$, $\N^\A\cup\N^\B$, $\N^\A$, $\N^\B$, and $\N^\C$, respectively, are omitted from $\hat \N$.

Further simplifications are possible, but perhaps less natural.
For instance, in the row-wise and monochrome-$\A$ cases, it is arguably natural to coarsen $\V^\B$ along rows of $\B$: every entry in a row of $\B$ is involved in  multiplications with the same entries of $\A$. 
After this coarsening, the nets induced by $\N^\B$ are coalesced along each row of $\B$, and so are combined into a single net whose cost equals the number of coalesced nets (also the sum of their costs).
Similarly, in the column-wise and monochrome-$\B$ cases, it is arguably natural to coarsen $\V^\A$ along columns of $\A$: every entry in a column of $\A$ multiplies the same entries of $\B$; the nets induced by $\N^\A$ are coalesced/combined along column of $\A$ as before.

These (arguably) natural simplifications don't always coarsen the nonzeros of all three matrices.
In particular, in the outer-product, monochrome-$\A$, and monochrome-$\B$ cases, there is no natural choice of rows or columns of $\C$, and in the monochrome-$\C$ case, there is no natural choice of rows or columns of $\A$ or $\B$.
Of course, nonzero vertex coarsening may be dictated by the application.   
For example, $\C$ may need to be distributed in a certain way for use after the \SpGEMM{}.
Or, in the row-wise case, if $\A$ and $\C$ are to be distributed row-wise, then additionally distributing $\B$ row-wise may be simplest from a software-design standpoint (and similarly distributing $\A$ column-wise in the column-wise case).

We have presented several restricted classes of algorithms that enable simplifications to the \SpGEMM{} hypergraph. For concreteness, we conclude with four self-contained examples, stated similarly to \cref{def:hp-spgemm}.
(We continue our simplifying assumption that $\A$ and $\B$ have no zero rows or columns.)

\begin{example}[Row-wise (RrR)]
\label[example]{ex:row}
Our first example pairs a row-wise parallelization with row-wise distributions of $\A$, $\B$, and $\C$, where the distributions of $\A$ and $\C$ are matched with the parallelization while the distribution of $\B$ is not.
\begin{gather*}
\begin{aligned}
\V &= \{ v_i : i \in [I] \} \cup \{v^\B_k : k \in [K] \}, \\
\N &= \{ n^\B_k = \{ v_i : (\exists i) \, (i,k) \in S_\A \}\cup\{ v^\B_k\} : k \in [K] \}, \\
w_\comp(v_i) &= |\{ (i,k,j) : (\exists k)\, ((i,k) \in S_\A \wedge (\exists j)\, (k,j) \in S_\B) \}|, \\
w_\mem(v_i) &= |\{ (i,k) : (\exists k)\, (i,k) \in S_\A\}| + |\{ (i,j) : (\exists j)\, (i,j) \in S_\C \}|, \\
\end{aligned} \\
\begin{aligned}
w_\comp(v^\B_k) &= 0, &
w_\mem(v^\B_k) = c(n^\B_k) &= |\{ (k,j) : (\exists j)\, (k,j) \in S_\B \}|. 
\end{aligned}
\end{gather*}
We have that $|\V| = I + K$ and $|\N| = K$, where each net has between $2$ and $I+1$ pins.
Note that computing the weights $v_\mem(v_i)$ requires determining $S_\C$. 
\end{example}

\begin{example}[Outer-product (CRf)] 
\label[example]{ex:outer}
Our second example pairs an outer-product parallelization with a column-wise distribution of $\A$, a row-wise distribution of $\B$, and a fine-grained distribution of $\C$, where the distributions of $\A$ and $\B$ are matched with the parallelization while the distribution of $\C$ is not.
\begin{gather*}
\begin{aligned}
\V &= \{ v_k : k \in [K] \} \cup \{v^\C_{ij} : (i,j) \in S_\C \}, \\
\N &= \{ n^\C_{ij} = \{ v_k : (\exists k) \, (i,k) \in S_\A \wedge (k,j) \in S_\B \}\cup\{ v^\C_{ij}\} : (i,j) \in S_\C \}, \\
w_\comp(v_k) &= |\{ i : (\exists k) \, (i,k) \in S_\A \}| \cdot |\{ j : (\exists k)\, (k,j) \in S_\B \}|, \\
w_\mem(v_k) &= |\{ i : (\exists k)\, (i,k) \in S_\A\}| + |\{ j : (\exists k)\, (k,j) \in S_\B \}|, \\
\end{aligned} \\
\begin{aligned}
w_\comp(v^\C_{ij}) &= 0, &
w_\mem(v^\C_{ij}) = c(n^\C_{ij}) &= 1.
\end{aligned}
\end{gather*}
We have that $|\V| = K + |S_\C|$ and $|\N| = S_\C$, where each net has between $2$ and $K+1$ pins.
Note that constructing $\V$ and $\N$ requires determining $S_\C$. 
A closely related hypergraph model for this example and two specializations (CRr and CRc) was previously studied~\cite{AA14}.
\end{example}

\begin{example}[Monochrome-$\A$ (Frf)]
\label[example]{ex:monoa}
Our third example pairs a monochrome-$\A$ parallelization with a row-wise distribution of $\B$ and fine-grained distributions of $\A$ and $\C$, where the distribution of $\A$ is matched with the parallelization while the distributions of $\B$ and $\C$ are not.
\begin{gather*}
\begin{aligned}
\V &= \{ v_{ik} : (i,k) \in S_\A \} \cup \{v^\B_k : k \in [K] \} \cup \{v^\C_{ij} : (i,j) \in S_\C\}, \\
\N &= \bigcup\left\{ \begin{gathered}
\{ n^\B_k = \{ v_{ik} : (\exists i) \, (i,k) \in S_\A \}\cup\{ v^\B_k\} : k \in [K] \} ,\\ 
\{ n^\C_{ij} = \{ v_{ik} : (\exists k) \, (i,k) \in S_\A \wedge (k,j) \in S_\B \} \cup \{ v^\C_{ij} \} : (i,j) \in S_\C \}
\end{gathered}\right\},
\end{aligned} \\
w_\comp(v_{ik}) = w_\mem(v^\B_k) = c(n^\B_k) = |\{ j : (\exists j) \, (k,j) \in S_\B \}|, \\
w_\mem(v_{ik}) = w_\mem(v^\C_{ij}) = c(n^\C_{ij})  = 1 ,\qquad
w_\comp(v^\B_k) = w_\comp(v^\C_{ij}) = 0.
\end{gather*}
We have that $|\V| = |S_\A| + K + |S_\C|$ and $|\N| = K + |S_\C|$, where each of the $K$ nets $n^\B_k$ has between $2$ and $I+1$ pins, and each of the $|S_\C|$ nets $n^\C_{ij}$ has between $2$ and $K+1$ pins.
Note that constructing $\V$ and $\N$ requires determining $S_\C$.
\end{example}

\begin{example}[Monochrome-$\C$ (ffF)] 
\label[example]{ex:monoc}
Our fourth example pairs a monochrome-$\C$ parallelization with fine-grained distributions of $\A$, $\B$, and $\C$, where the distribution of $\C$ is matched with the parallelization while the distributions of $\A$ and $\B$ are not.
\begin{gather*}
\begin{aligned}
\V &= \{ v_{ij} : (i,j) \in S_\C \} \cup \{v^\A_{ik} : (i,k) \in S_\A  \} \cup \{v^\B_{kj} : (k,j) \in S_\B\}, \\
\N &= \bigcup\left\{ \begin{gathered}
\{ n^\A_{ik} = \{ v_{ij} : (\exists j) \, (k,j) \in S_\B \}\cup\{ v^\A_{ik}\} : (i,k) \in S_\A \} ,\\ 
\{ n^\B_{kj} = \{ v_{ij} : (\exists i) \, (i,k) \in S_\A \} \cup \{ v^\B_{kj} \} : (k,j) \in S_\B \}
\end{gathered}\right\},
\end{aligned}\\
w_\comp(v_{ij}) = |\{ k : (\exists k) \, (i,k) \in S_\A \wedge (k,j) \in S_\B \}|, \\
w_\mem(v_{ij}) = w_\mem(v^\A_{ik}) = w_\mem(v^\B_{kj}) = c(n^\A_{ik}) = c(n^\B_{kj}) = 1, \\
w_\comp(v^\A_{ik}) = w_\comp(v^\B_{kj}) = 0.
\end{gather*}
We have that $|\V| = |S_\A| + |S_\B| + |S_\C|$ and $|\N| = |S_\A| + |S_\B|$, where each of the $|S_\A|$ nets $n^\A_{ik}$ has between $2$ and $J+1$ pins, and each of the $|S_\B|$ nets $n^\B_{kj}$ has between $2$ and $I+1$ pins.
Note that constructing $\V$ requires determining $S_\C$.
\end{example}

In each of these four examples, the simplified hypergraph is potentially much smaller than the original (fine-grained) hypergraph in terms of the numbers of vertices, nets, or pins.
Our experiments in \cref{sec:expt} using the PaToH hypergraph partitioning software demonstrated orders-of-magnitude improvements in runtime when partitioning simplified (vs.\ fine-grained) hypergraphs. 

We caution that constructing the simplified hypergraphs in these four examples requires determining $S_\C$, which can be as expensive as computing $\C$. 
This cost can be avoided, e.g., in the row-wise (RrR) case (\cref{ex:row}), if we omit memory weights.
% this simplified construction requires $O\left(|S_\A| + |S_\B|\right)$ operations. 

\subsection{Sparse Matrix-Vector Multiplication (\SpMV{})}
\label{sec:SpMV}
In the case of multiplying a sparse matrix by a dense vector (\SpMV{}), our hypergraph model simplifies in several ways.
If the input matrix $\B$ is a dense vector of length $K$, then the output matrix $\C$ is a dense vector of length $I$, since $\A$ has a nonzero in every row under the assumptions of \cref{sec:notation}.
In this case, there is only one multiplication performed for each nonzero of $\A$, so the multiplication vertices of $\V$ can be indexed by only two parameters, letting us write $\V^\mm=\{v_{ik} : (i,k) \in S_\A\}$. 
Furthermore, the $\B$ and $\C$ nonzero vertices can be indexed by only one parameter and simplify to $\V^\B=\{v_{k}^\B : k \in [K]\}$ and $\V^\C=\{v_{i}^\C : i \in [I]\}$.

In addition to the simplifications based on the inputs of \SpMV{}, we can also apply vertex coarsening, restricting the set of algorithms, to reproduce the ``fine-grain'' model for \SpMV{} of \c{C}ataly\"{u}rek and Aykanat~\citeyear{CA01a}; this model assumes that $I=K$, i.e., $\A$ is a square matrix.
There are three steps in the following derivation.

First, because each nonzero of $\A$ is involved in exactly one multiplication, we force $v_{ik}$ and $v_{ik}^\A$ to be monochrome for each $(i,k) \in S_\A$.
(This corresponds to monochrome-$\A$ (Fff), discussed in \cref{sec:restrict-alg}.)
We use the notation $\hat v_{ik}$ to refer to the coarsened vertex, which has weights $w_\comp(\hat v_{ik})=1$ and $w_\mem(\hat v_{ik})=1$.
Note that this coarsening implies that all nets in $\N^\A$ have exactly one vertex, and thus we omit them.

Second, we apply further coarsening to attain a symmetric partitioning of the input and output vector entries.
We do this to satisfy the ``consistency condition''~\cite{CA01a}, motivated by the problem of performing repeated \SpMV{} operations with the same matrix.
For every $i \in [I] = [K]$, we force $v_i^\B$, $v_i^\C$, and $\hat v_{ii}$ (if it exists) to be monochrome.
This corresponds to assigning vector entries to the processor that owns the corresponding diagonal entry of the matrix.
We maintain the notation $\hat v_{ii}$ to reference the coarsened vertex, which has weights $w_\comp(\hat v_{ii})=1$ and $w_\mem(\hat v_{ii})=3$, if $(i,i)\in S_\A$, and $w_\comp(\hat v_{ii})=0$ and $w_\mem(\hat v_{ii})=2$, if $(i,i) \not\in S_\A$.

Third, because the fine-grain SpMV model does not explicitly enforce memory balance, we drop the memory balance constraint on partitions by setting $\delta=p-1$.
Even though we removed the memory balance constraints, computational balance in this model guarantees a balanced partition of the nonzeros of $\A$. 
However, the memory allocated to the input and output vectors is not accounted for, and so the overall memory imbalance can potentially be higher than the computational imbalance.

The preceding derivation reproduces the ``fine-grain'' model for SpMV~\cite{CA01a}, which consists of a vertex for every nonzero in the matrix (and a zero-weight ``dummy'' vertex for every zero diagonal entry) and a net for each row and each column.
Similar simplifications to 1D \SpGEMM{} algorithms (see \cref{sec:restrict-alg}) yield row-wise and column-wise \SpMV{} algorithms~\cite{CA99}.
In particular, the row-wise (RrR) hypergraph (\cref{ex:row}) is identical to the ``column-net'' \SpMV{} hypergraph (modeling a row-wise algorithm), except for the presence of the $\B$-nonzero vertices and the memory weights.
Similarly, the outer-product (CRf) hypergraph (\cref{ex:outer}) is identical to the ``row-net'' \SpMV{} hypergraph (modeling a column-wise algorithm), except for the presence of the $\C$-nonzero vertices and the memory weights. 
On the other hand, the column-wise \SpGEMM{} parallelizations applied to \SpMV{} have no analogue in \c{C}ataly\"{u}rek and Aykanat's models~\citeyear{CA99,CA01a} since there is no parallelism. 
Additionally, the monochrome-$\A$, -$\B$, and -$\C$ \SpGEMM{} parallelizations, in the case of \SpMV{}, correspond to the aforementioned ``fine-grain'', ``row-net'', and ``column-net'' \SpMV{} hypergraphs. 

\subsection{Generalizing \SpGEMM{} Algorithms}
\label{sec:generalize-HG}

In this section, we explore two directions in which the class of \SpGEMM{} algorithms (defined in \cref{sec:model}) can be generalized.
The first generalization exploits relations among the input matrices' entries to reduce algorithmic costs --- recall that \cref{sec:notation} assumed $\A$- and $\B$-entries are unrelated.
A common application is exploiting symmetry of $\A$, $\B$, or $\C$.
The second generalization considers the case where only a subset of the output entries are desired, a task called \emph{masked} \SpGEMM{} --- recall that \cref{sec:notation} assumed all nonzero $\C$-entries are computed.

We show that both classes of algorithms can be modeled by straightforward simplifications to the \SpGEMM{} hypergraph (defined in \cref{sec:fg-model}).
In particular, we apply vertex coarsening similarly to as in \cref{sec:coarsening}, except here we adjust vertex weights and net costs in a simpler manner. 
The hypergraph-based communication bounds in \cref{sec:parallelLB,sec:sequentialLB}, as well as the restrictions studied in the previous subsections of \cref{sec:simplify-HG}, extend to both these classes of algorithms with only minor changes required to address the vertex coarsening.

\subsubsection{Exploiting Input Matrix Relations}
\label{sec:input-relations}
We consider modifying \SpGEMM{} algorithms to exploit known equivalence relations on the nonzero entries of $\A$ and $\B$ to reduce algorithmic costs.
For simplicity, we will consider equality relations: e.g., if $\A = \A^T$, then nonzeros in the upper triangle of $\A$ equal the corresponding entries in the lower triangle of $\A$. 
However, the following approach extends to cases like $\A=-\A^T$, or more generally to the case where each equivalence class is assigned a single value and each nonzero in that equivalence class is a function of that value. 

Consider any equality relation on the nonzeros of $\A$ and $\B$. 
We first consider \SpGEMM{}-like algorithms that store exactly one copy of repeated $\A$- and $\B$-entries, but still perform all nontrivial multiplications and compute all nonzero $\C$-entries, regardless of whether these values also include repetitions.
Such algorithms can be naturally modeled by altering the \SpGEMM{} hypergraph $(\V,\N) = \HG(\A,\B)$ according to a given partition of $\V^\A \cup \V^\B$, where each part is a set of nonzero entries that must have the same value.
For example, if an algorithm only inputs the unique entries of $\A = \A^T$, then the partition groups each pair of distinct vertices $v^\A_{ik}$ and $v^\A_{ki}$ while keeping the nonzero vertices associated with the diagonal of $\A$, as well as all of $\V^\B$, in singleton parts.
Now coarsen these vertex sets, following the approach in \cref{sec:coarsening} except setting the memory costs of the coarsened vertices to 1, rather than the number of coarsened vertices. 
Since only nonzero vertices are combined, no nets are coalesced by this coarsening.

We next extend the preceding class of \SpGEMM{}-like algorithms to avoid performing redundant nontrivial multiplications and computing redundant $\C$-entries.
We say that two nontrivial multiplications are redundant if their left operands are equal and their right operands are equal. 
For example, if $\A=\A^T$ and $\B=\B^T$, then any nontrivial multiplication $a_{lm}b_{ml}$ must equal $a_{ml}b_{lm}$.
More generally, if there exist $(i,k,j)$ and $(r,t,s)$ such that $a_{ik}=a_{rt}$ and $b_{kj}=b_{ts}$, then $a_{ik}b_{kj}=a_{rt}b_{ts}$ --- only one of these two multiplicationsies needs to be performed.
We say that two $\C$-entries are redundant if their summands can be matched as pairs of redundant nontrivial multiplications.   
For example, if all entries of $\A$ and $\B$ are equal to the same nonzero value, then all $\C$-entries are equal.  
More generally, if there exist $(i,j)$ and $(r,s)$ in $S_\C$ and a bijection
\[ \phi \colon \mathcal{K} = \{k : (\exists k) \, ((i,k) \in S_\A \wedge (k,j) \in S_\B)\} \rightarrow \{t : (\exists t)\, ((r,t)\in S_\A \wedge (t,s) \in S_\B)\}\] 
such that $a_{ik} = a_{r\phi(k)}$ and $b_{kj} = b_{\phi(k)s}$ for each $k \in \mathcal{K}$, then $c_{ij} = c_{rs}$, and only one of the two needs to be computed.  
We further alter the \SpGEMM{} hypergraph to exploit these savings as follows.
Let the notation $u \equiv v$ assert that $u,v \in \V$ are contained in the same part in some fixed partition of $\V$.
We extend the partition of $\V^\A \cup \V^\B$ to all of $\V$ in two steps, partitioning $\V^\mm$ and then $\V^\C$.
First, for each pair $v_{ikj},v_{rts} \in \V^\mm$, we specify that $v_{ikj} \equiv v_{rts}$ if $v^\A_{ik} \equiv v^\A_{rt}$ and $v^\B_{kj} \equiv v^\B_{ts}$.
Second, for each $v^\C_{ij},v^\C_{rs} \in \V^\C$, we specify that $v^\C_{ij} \equiv v^\C_{rs}$ if $n^\C_{ij}$ and $n^\C_{rs}$ intersect the same parts of $\V^\mm$.
We then coarsen $\V$ according to this partition: we follow the approach in \cref{sec:coarsening} except setting both the memory costs of the coarsened nonzero vertices and the computation costs of the coarsened multiplication vertices to 1, rather than to the numbers of coarsened nonzero and multiplication vertices. 
Now, it is possible that some nets are coalesced: coalesced nets can be combined without increasing net costs since only one nonzero needs to be stored/sent/received.

Lastly, we next extend the preceding class of \SpGEMM{}-like algorithms to exploit commutative multiplication, which is not guaranteed in our model (see \cref{sec:model}).
For example, if $\B=\A$, then $a_{lm} b_{ml} = a_{ml}b_{lm}$; or, if $\B=\A^T$, then $\C = \C^T$; neither of these redundancies necessarily occur without commutative multiplication. 
We model algorithms that avoid this larger class of redundant multiplications (and $\C$-entries) by augmenting the preceding construction in just one place: for each pair $v_{ikj},v_{rts} \in \V^\mm$, we additionally specify that $v_{ikj} \equiv v_{rts}$ if $v^\A_{ik} \equiv v^\B_{rt}$ and $v^\A_{kj} \equiv v^\B_{ts}$.

\subsubsection{Masked \SpGEMM{}}
\label{sec:output-sparsity}
We consider modifying \SpGEMM{} algorithms to only compute a subset of the entries of $\C=\A\cdot\B$; this is known as \emph{masked} \SpGEMM{} (see, e.g.,~\cite{ABG15}), where $S \subseteq S_\C$ indexes the desired subset of $\C$-entries and $S_\C \setminus S$ is called the \emph{mask}.
That is, only for each $(i,j) \in S$ are $c_{ij}$ and its associated nontrivial multiplications $a_{ik}b_{kj}$ computed.
We develop a hypergraph model starting with the usual \SpGEMM{} hypergraph, $(\V,\N)=\HG(\A,\B)$. 
We remove from $\N^\C$ each net $n^\C_{ij}$ where $(i,j) \not\in S$, as well as that net's elements (vertices) from $\V$.
In particular, for each $(i,j) \in S_\C \setminus S$, this removes from $\V^\C$ the vertex $v^\C_{ij}$ and from $\V^\mm$ each vertex $v_{ikj}$ with $k \in [K]$.

Depending on $S_\A$, $S_\B$, and $S$, it is possible that some entries of $\A$ and $\B$ are not involved in any nontrivial multiplications after masking.
Avoiding this triviality motivated our simplifying assumption in \cref{sec:model} that $\A$ and $\B$ have no zero rows or columns.
In terms of the hypergraph, this manifests as elements of $\N^\A$ and $\N^\B$ becoming singletons, containing only their associated nonzero vertices $v^\A_{ik}$ or $v^\B_{kj}$.
To model algorithms that reduce memory costs by not storing the associated nonzero entries, these nets and vertices can be removed from $\N^\A$ or $\N^\B$, and $\V^\A$ or $\V^\B$, respectively.

Next, we consider a generalization of masked \SpGEMM{}.
When the underlying set $X$ has a right multiplicative identity $1$, masking can be viewed as computing $(\A \cdot \B) \odot \M$, where $\odot$ denotes the Hadamard (or entrywise) matrix product and where the $\{0,1\}$-valued matrix $\M$ has the same dimensions as $\C=\A\cdot\B$ and $S_\M = S$.
To generalize, consider computing $\C = (\A \cdot \B) \odot \M$, where $\M$ is a general sparse matrix with the same dimensions as $\C$.
Let $(\V,\N)$ denote the hypergraph for masked $\SpGEMM$ with $S=S_\M$, as in the preceding paragraphs.
Suppose the algorithmic constraint that any processor who stores some $c_{ij}$ also must both store $m_{ij}$ and perform the  multiplication involving $m_{ij}$.
To model this, for each of the (unmasked) nonzero vertices $v^\C_{ij}$, we set its memory cost to 2 and its computation cost to 1.
It is possible to relax this algorithmic constraint by introducing additional multiplication vertices $v'_{ij}$ for the Hadamard product, additional nonzero vertices $v^\M_{ij}$ for the $\M$-entries, and additional nets $n^\M_{ij} = \{v'_{ij},v^\M_{ij},v^\C_{ik}\}$.

\section{Experimental Results}
\label{sec:expt}

We now compare our fine-grained hypergraph model from \cref{sec:fg-model} with the six restricted parallelizations developed in \cref{sec:restrict-par}, in the context of three applications: algebraic multigrid~(AMG, \cref{sec:AMG}), linear programming~(LP, \cref{sec:LP}), and Markov clustering~(MCL, \cref{sec:MC}).
A list of the \SpGEMMs{} that we study in our experiments is shown in \cref{tbl:matrixdims}. 
Our experiments consider only the distributed-memory parallel case, where we expect communication cost to be most closely correlated with execution time. 
Because hypergraph partitioning is NP-hard (in general), hypergraph partitioners use heuristics to approximate optimal partitions; results we present are not guaranteed to be optimal. 
However, it remains open whether partitioning \SpGEMM{} hypergraph instances are NP-hard.

The hypergraph partitioner we use is PaToH (version 3.2)~\cite{PaToH}.
% footnote{We used PaToH version 3.2, which was the most recent version at the time this paper was written.}
%
PaToH minimizes the ``connectivity metric,'' defined as the sum, over all nets, of the product of each net's cost and its number of incident parts minus one.
The communication costs shown below represent the maximum, over all parts, of the sum of each part's non-internal incident nets' costs (matching \cref{lem:per-proc-comm}). 
This cost is most closely aligned with the parallel running time devoted to communicating words of data, called the ``critical-path bandwidth cost'' by Ballard et al.~\citeyear{BCDH+14}.

We do not constrain memory load balance (i.e., $\delta=p-1$) --- and so we omit $\V^\nz$ from the hypergraphs --- but we do constrain computation load balance, $\epsilon = 0.01$. In a few cases, our partitioner failed to produce such a balanced partition. We discuss these cases below.

\begin{table}%
    \tbl{%
        \label{tbl:matrixdims}%
        The \SpGEMMs{} that we studied in our experiments and their parameters. The model AMG problem is represented by the rows 27-AP and 27-PTAP, corresponding to the $A\cdot P$ and $P^T\cdot (AP)$ \SpGEMMs{}, respectively, and the SA-$\rho$AMGe problem is similarly represented by the rows SA-AP and SA-PTAP. The AMG experiments were carried out in a weak scaling regime and are represented in this table by the largest instances.  The LP experiments have one input matrix and compute $\C=\A\cdot\A^T$, and the MCL experiments have one input matrix and compute $\C=\A\cdot \A$.  In addition to the the dimensions of the \SpGEMM, we provide the average number of nonzeros per row of each matrix and the ratio of nontrivial multiplications to output nonzeros.}%
    {%
        \begin{tabular}{ll*{3}{R{7}}*{3}{R{4.1}}R{3.1}}
\toprule
& Name & $I$ & $K$ & $J$ & $|S_\A|/I$ & $|S_\B|/K$ & $|S_\C|/I$ & $|\V^\mm|/|S_\C|$ \\
\midrule
\vertseclabel{4}{\textbf{AMG}}{sec:AMG}
& 27-AP & 970299 & 970299 & 35937 & 26.4582 & 4.4633 & 12.0551 & 9.8744 \\
& 27-PTAP & 35937 & 970299 & 35937 & 4.4633 & 12.055 & 25.3965 & 49.0196 \\
& SA-AP & 1088640 & 1088640 & 31496 & 26.406 & 20.145 & 38.529 & 13.935 \\
& SA-PTAP & 31496 & 1088640 & 31496 & 696.309 & 38.529 & 216.359 & 139.333 \\
\midrule
\vertseclabel{5}{\textbf{LP}}{sec:LP}
& fome21 & 67748 & 216350 & 67748 & 6.868 & 2.151 & 9.45 & 1.622 \\
& pds80 & 129181 & 434580 & 129181 & 7.182 & 2.135 & 9.669 & 1.64 \\
& pds100 & 156243 & 514577 & 156243 & 7.015 & 2.13 & 9.413 & 1.639 \\
& cont11l & 1468599 & 1961394 & 1468599 & 3.665 & 2.744 & 12.3 & 1.487 \\
& sgpf5y6 & 246077 & 312540 & 246077 & 3.381 & 2.662 & 11.284 & 1.234 \\
\midrule
\vertseclabel{7}{\textbf{MCL}}{sec:MC}
& biogrid11 & 5853 & 5853 & 5853 & 21.452 & 21.452 & 2105.727 & 1.626 \\
& dip & 5051 & 5051 & 5051 & 8.749 & 8.749 & 200.859 & 1.588 \\
& wiphi & 5955 & 5955 & 5955 & 8.396 & 8.396 & 85.616 & 1.525 \\
& dblp & 425957 & 425957 & 425957 & 4.929 & 4.929 & 64.816 & 1.654 \\
& enron & 36692 & 36692 & 36692 & 10.02 & 10.02 & 831.03 & 1.689 \\
& facebook & 4039 & 4039 & 4039 & 43.691 & 43.691 & 717.129 & 6.493 \\
& roadnetca & 1971281 & 1971281 & 1971281 & 2.807 & 2.807 & 6.548 & 1.358 \\
\bottomrule
\end{tabular}
    }%
\end{table}

We performed our experiments on two machines at the National Energy Research Supercomputing Center in Berkeley, California. The experiments with the AMG model problem were carried out on a node (1 TB memory) of the machine Carver; all other experiments were carried out on a node (128 GB memory) of the machine Cori.

Building the fine-grained hypergraph and partitioning it is as expensive as performing the associated \SpGEMM, and furthermore our partitioner is sequential.
Although we are not focused on partitioning times in this study, fast partitioning is important for a user who seeks to find the best \SpGEMM{} strategy for a particular problem. 
In our experiments, the partitioning times varied from a few seconds up to 5 hours, while the time to build the hypergraph was negligible in comparison. We observed that the partitioning time varied with the dimensions of the matrices (smaller was faster), the hypergraph model (1D was fastest while fine-grained was slowest), and on whether the matrices have a regular nonzero structure (the 27-point stencil matrices in \cref{sec:AMG} were quicker to partition, while the social-network matrices in \cref{sec:MC} were slower).

\subsection{Application: Algebraic Multigrid}
\label{sec:AMG}

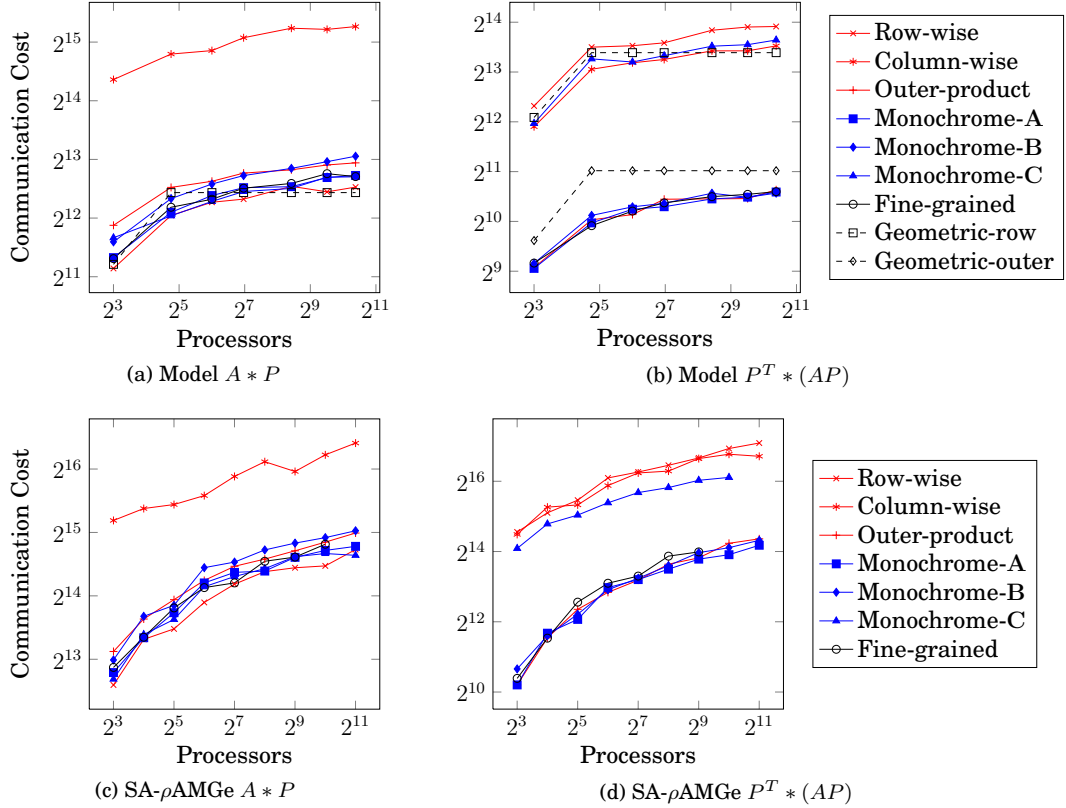
\begin{figure}
\centering
\renewcommand{\plotwidth}{2.65in}
\renewcommand{\plotheight}{2.65in}
\subfloat[][Model $A*P$]{
	\geomrowtrue
	\nolegendtrue
	\renewcommand{\datafile}{27ptAP.txt}
	\begin{tikzpicture}[scale=0.75]
\begin{axis}[
	label style={font=\large},
	tick label style={font=\large},
	legend style={at={(1.1,0.5)},anchor=west,font=\large},
	legend cell align=left,
	ylabel={\yaxistitle}, 
	xlabel={Processors},
        width=\plotwidth, 
        height=\plotheight,
        	xmode=log,
	log basis x ={2},
	ymode=log,
	log basis y={2},
]	
	% plot each column of data
	\addplot[color=red,mark=x] table[x=P, y=row] {\datafile};
	\addlegendentry{Row-wise}
	\ifata
		;
	\else
		\addplot[color=red,mark=asterisk] table[x=P, y=col] {\datafile};
		\addlegendentry{Column-wise}
	\fi
	\addplot[color=red,mark=+] table[x=P, y=outer] {\datafile};
	\addlegendentry{Outer-product}
	\addplot[color=blue,mark=square*] table[x=P, y=monoa] {\datafile};
	\addlegendentry{Monochrome-$\A$}
	\ifata
		;
	\else
		\addplot[color=blue,mark=diamond*] table[x=P, y=monob] {\datafile};
		\addlegendentry{Monochrome-$\B$}
	\fi
	\addplot[color=blue,mark=triangle*] table[x=P, y=monoc] {\datafile};
	\addlegendentry{Monochrome-$\C$}
	\addplot[color=black,mark=o] table[x=P, y=fine] {\datafile};
	\addlegendentry{Fine-grained}
	\ifgeomrow
		\addplot[dashed,color=black,mark=square,every mark/.append style={solid}] table[x=P, y=geom-row] {\datafile};
		\addlegendentry{Geometric-row}
	\fi
	\ifgeomouter
		\addplot[dashed,color=black,mark=diamond,every mark/.append style={solid}] table[x=P, y=geom-outer] {\datafile};
		\addlegendentry{Geometric-outer}
	\fi

	\ifnolegend
		% suppress legend
		\legend{}
	\fi
\end{axis}
\end{tikzpicture}
	\label{fig:model-AP}
}
\subfloat[][Model $P^T*(AP)$]{
	\geomrowtrue
	\geomoutertrue
	\nolegendfalse
	\renewcommand{\yaxistitle}{}
	\renewcommand{\datafile}{27ptPTAP.txt}
	\begin{tikzpicture}[scale=0.75]
\begin{axis}[
	label style={font=\large},
	tick label style={font=\large},
	legend style={at={(1.1,0.5)},anchor=west,font=\large},
	legend cell align=left,
	ylabel={\yaxistitle}, 
	xlabel={Processors},
        width=\plotwidth, 
        height=\plotheight,
        	xmode=log,
	log basis x ={2},
	ymode=log,
	log basis y={2},
]	
	% plot each column of data
	\addplot[color=red,mark=x] table[x=P, y=row] {\datafile};
	\addlegendentry{Row-wise}
	\ifata
		;
	\else
		\addplot[color=red,mark=asterisk] table[x=P, y=col] {\datafile};
		\addlegendentry{Column-wise}
	\fi
	\addplot[color=red,mark=+] table[x=P, y=outer] {\datafile};
	\addlegendentry{Outer-product}
	\addplot[color=blue,mark=square*] table[x=P, y=monoa] {\datafile};
	\addlegendentry{Monochrome-$\A$}
	\ifata
		;
	\else
		\addplot[color=blue,mark=diamond*] table[x=P, y=monob] {\datafile};
		\addlegendentry{Monochrome-$\B$}
	\fi
	\addplot[color=blue,mark=triangle*] table[x=P, y=monoc] {\datafile};
	\addlegendentry{Monochrome-$\C$}
	\addplot[color=black,mark=o] table[x=P, y=fine] {\datafile};
	\addlegendentry{Fine-grained}
	\ifgeomrow
		\addplot[dashed,color=black,mark=square,every mark/.append style={solid}] table[x=P, y=geom-row] {\datafile};
		\addlegendentry{Geometric-row}
	\fi
	\ifgeomouter
		\addplot[dashed,color=black,mark=diamond,every mark/.append style={solid}] table[x=P, y=geom-outer] {\datafile};
		\addlegendentry{Geometric-outer}
	\fi

	\ifnolegend
		% suppress legend
		\legend{}
	\fi
\end{axis}
\end{tikzpicture}
	\label{fig:model-PTAP}
} \\
\geomrowfalse
\geomouterfalse
\subfloat[][SA-$\rho$AMGe $A*P$]{
	\nolegendtrue
	\renewcommand{\datafile}{saamgeAP.txt}
	\begin{tikzpicture}[scale=0.75]
\begin{axis}[
	label style={font=\large},
	tick label style={font=\large},
	legend style={at={(1.1,0.5)},anchor=west,font=\large},
	legend cell align=left,
	ylabel={\yaxistitle}, 
	xlabel={Processors},
        width=\plotwidth, 
        height=\plotheight,
        	xmode=log,
	log basis x ={2},
	ymode=log,
	log basis y={2},
]	
	% plot each column of data
	\addplot[color=red,mark=x] table[x=P, y=row] {\datafile};
	\addlegendentry{Row-wise}
	\ifata
		;
	\else
		\addplot[color=red,mark=asterisk] table[x=P, y=col] {\datafile};
		\addlegendentry{Column-wise}
	\fi
	\addplot[color=red,mark=+] table[x=P, y=outer] {\datafile};
	\addlegendentry{Outer-product}
	\addplot[color=blue,mark=square*] table[x=P, y=monoa] {\datafile};
	\addlegendentry{Monochrome-$\A$}
	\ifata
		;
	\else
		\addplot[color=blue,mark=diamond*] table[x=P, y=monob] {\datafile};
		\addlegendentry{Monochrome-$\B$}
	\fi
	\addplot[color=blue,mark=triangle*] table[x=P, y=monoc] {\datafile};
	\addlegendentry{Monochrome-$\C$}
	\addplot[color=black,mark=o] table[x=P, y=fine] {\datafile};
	\addlegendentry{Fine-grained}
	\ifgeomrow
		\addplot[dashed,color=black,mark=square,every mark/.append style={solid}] table[x=P, y=geom-row] {\datafile};
		\addlegendentry{Geometric-row}
	\fi
	\ifgeomouter
		\addplot[dashed,color=black,mark=diamond,every mark/.append style={solid}] table[x=P, y=geom-outer] {\datafile};
		\addlegendentry{Geometric-outer}
	\fi

	\ifnolegend
		% suppress legend
		\legend{}
	\fi
\end{axis}
\end{tikzpicture}
	\label{fig:saamge-AP}
}
\subfloat[][SA-$\rho$AMGe $P^T*(AP)$]{
	\nolegendfalse
	\renewcommand{\yaxistitle}{}
	\renewcommand{\datafile}{saamgePTAP.txt}
	\begin{tikzpicture}[scale=0.75]
\begin{axis}[
	label style={font=\large},
	tick label style={font=\large},
	legend style={at={(1.1,0.5)},anchor=west,font=\large},
	legend cell align=left,
	ylabel={\yaxistitle}, 
	xlabel={Processors},
        width=\plotwidth, 
        height=\plotheight,
        	xmode=log,
	log basis x ={2},
	ymode=log,
	log basis y={2},
]	
	% plot each column of data
	\addplot[color=red,mark=x] table[x=P, y=row] {\datafile};
	\addlegendentry{Row-wise}
	\ifata
		;
	\else
		\addplot[color=red,mark=asterisk] table[x=P, y=col] {\datafile};
		\addlegendentry{Column-wise}
	\fi
	\addplot[color=red,mark=+] table[x=P, y=outer] {\datafile};
	\addlegendentry{Outer-product}
	\addplot[color=blue,mark=square*] table[x=P, y=monoa] {\datafile};
	\addlegendentry{Monochrome-$\A$}
	\ifata
		;
	\else
		\addplot[color=blue,mark=diamond*] table[x=P, y=monob] {\datafile};
		\addlegendentry{Monochrome-$\B$}
	\fi
	\addplot[color=blue,mark=triangle*] table[x=P, y=monoc] {\datafile};
	\addlegendentry{Monochrome-$\C$}
	\addplot[color=black,mark=o] table[x=P, y=fine] {\datafile};
	\addlegendentry{Fine-grained}
	\ifgeomrow
		\addplot[dashed,color=black,mark=square,every mark/.append style={solid}] table[x=P, y=geom-row] {\datafile};
		\addlegendentry{Geometric-row}
	\fi
	\ifgeomouter
		\addplot[dashed,color=black,mark=diamond,every mark/.append style={solid}] table[x=P, y=geom-outer] {\datafile};
		\addlegendentry{Geometric-outer}
	\fi

	\ifnolegend
		% suppress legend
		\legend{}
	\fi
\end{axis}
\end{tikzpicture}
	\label{fig:saamge-PTAP}
}
\renewcommand{\yaxistitle}{Communication Cost}
\caption{Communication costs of various parallelizations for performing the top-level triple product in the setup phase of algebraic multigrid (AMG --- \cref{sec:AMG}).}
\label{fig:amgplots}
\end{figure}

First, we consider the \SpGEMMs{} that arise in the context of algebraic multigrid (AMG). AMG methods are linear solvers typically applied to systems of linear equations arising from the discretization of partial differential equations (PDEs). AMG methods operate in two stages: setup and solve. The setup stage consists of forming a ``grid hierarchy,'' which is a sequence of matrices that represent the PDE at progressively coarser levels of discretization.

To build the grid hierarchy, we generate a sequence of ``prolongators'' $P_1,\dotsc, P_{\ell - 1}$, where $\ell$ is the number of levels in the hierarchy, an algorithm parameter. Prolongators are matrices with more rows than columns and the number of rows in each prolongator is equal to the number of columns in its predecessor. The grid hierarchy $A_1, \dotsc, A_\ell$ is then formed by successively computing
\begin{equation}
    A_2 = P_1^T A_1 P_1 \text{,}\qquad
    A_3 = P_2^T A_2 P_2 \text{,}\qquad
    \ldots\text{,}\qquad
    A_\ell = P_{\ell - 1}^T A_{\ell - 1} P_{\ell - 1} \text{,} \label{eq:PTAP}
\end{equation}
where $A_1$ is the fine-grained discretization of the PDE, an algorithm input.

Typically, AMG methods are designed such that all of the matrices in \cref{eq:PTAP} are sparse, so each triple product is  computed using two \SpGEMMs{}: one that forms $A_k \cdot P_k$, followed by another that forms $P_k^T \cdot (A_k P_k)$ for $k \in [\ell-1]$. In practice, the \SpGEMMs{} in \cref{eq:PTAP} can take a large part of the time spent in the setup phase, and the setup phase can take a large part of the total time. Considerable effort has been made recently to optimize these \SpGEMMs{} (see, e.g.,~\cite{BSH15-TR,PSYMD15}).

We consider the first level of the grid hierarchy for both a model problem and a realistic problem.
In the model problem, the rows of the input matrix $A_1$ correspond to points of an $N\times N \times N$ regular grid, and its nonzero structure corresponds to a 27-point stencil so that each point is adjacent to its (at most) 26 nearest neighbors.
The prolongator matrix $P_1$ is $N^3 \times (N/3)^3$ (we assume $N$ is divisible by 3) and is defined so that $3\times 3\times 3$ sub-grids correspond to single points in the coarser grid, and its values are computed using the technique of ``smoothed aggregation" (using damped Jacobi).
This model problem is one that has been previously studied~\cite{BDKS15,BSH15-TR}.
Because the grid is regular, we can compare our partitioning results with regular, geometric partitions.
For example, the natural partition of the rows of $A$ corresponds to assigning each processor a contiguous $(N/p^{1/3})\times (N/p^{1/3})\times (N/p^{1/3})$ subcube of points (assuming $N/p^{1/3}$ is an integer).

%While AMG is typically used in the case of irregular grids, this model problem is representative of applications that solve a PDE in three spatial dimensions using AMG with smoothed aggregation.
%Ballard et al.~\citeyear{BSH15-TR} compare this model problem with a more realistic application and argue that it is representative.

The second set of matrices comes from an oil-reservoir simulation called \textsc{spe10}~\cite{CB01}. In this problem, the Darcy equation, in dual form, is used to model fluid flow in porous media. 
The spatial domain is a rectangular prism, tessellated using a hexahedral mesh and discretized using linear finite elements~\cite{BKMVY15,CVV15}.
Brezina and Vassilevski~\citeyear{BV11} solved this problem using an AMG variant called SA-$\rho$AMGe, wherein $P_1$ has approximately $35\times$ more rows than columns (slightly more aggressive coarsening than the model problem) and is constructed using a polynomial smoother, giving more nonzeros.
For details of the implementation, see~\cite{K12,KKV13}.
%In this algorithm, we start the setup phase by partitioning the elements into coarse sets that we call the agglomerates. We restrict the PDE operator to each of the agglomerates, and we compute a basis of the near-nullspace of each restricted operator. We then arrange these bases as the diagonal blocks of a narrow rectangular block-diagonal matrix $\bar{P}$, and we produce the final $P$ from $\bar{P}$ by applying a polynomial smoother. 

%The most time-consuming part of the setup phase is the computation of the near-nullspaces -- this is accomplished by applying an eigensolver to each agglomerate. The algorithm is engineered so that the eigenproblems that we solve are sufficiently small to be solved on individual cores, and so this part of the algorithm requires no communication. For this reason, on machines with a sufficiently high core count, we expect the time spent on forming the \SpGEMM{} $P^T A P$ to dominate the setup phase.

We consider both \SpGEMMs{} for both problems, comparing the fine-grained model with all six restricted parallelizations from \cref{sec:restrict-par}.
In the case of the model problem, we also compare against 1D algorithms based on geometric partitions of the regular grid; these natural and efficient (but not necessarily optimal) instances can validate the quality of results obtained from the hypergraph partitioner.
For both problems, we present a weak-scaling study, maintaining $I/p$, where $I=K$ is the number of rows and columns of $A_1$.
The model problem ranges from dimension $18^3=5832$ on 8 processors up to dimension $99^3=970229$ on 1331 processors, and the SA-$\rho$AMGe problem ranges from dimension 4900 on 8 processors up to dimension 1088640 on 2048 processors.

In four cases involving the SA-$\rho$AMGe matrices, our partitioner failed with an out-of-memory error. The $P_1$ matrix in SA-$\rho$AMGe has more nonzero entries than in the model problem, and so the corresponding hypergraphs and their partitioning require more memory.
The \SpGEMMs{} that failed were the largest for the fine-grained model of $A\cdot P$, the largest for the monochrome-$\C$ model of $P^T \cdot (AP)$, and the two largest for the fine-grained model of $P^T \cdot (AP)$.

The results for the model problem are given in the top row of \cref{fig:amgplots}. 
We see in \cref{fig:model-AP} that for the first \SpGEMM, all parallelizations except for column-wise require about the same amount of communication (to within a factor of 2).
Note that the row-wise algorithm with geometric partitioning (labeled ``Geometric-row'') achieves the least communication on 1331 processors; this implies that PaToH is not finding an optimal partition at that scale (at least in the row-wise, monochrome-$\A$, monochrome-$\C$, or fine-grained cases), but we believe that the optimal partition is not much better than the geometric one.
We can also compare the relative costs discovered by hypergraph partitioning with Ballard et al.'s theoretical analysis~\citeyear[Table 2]{BSH15-TR} based on geometric partitioning.
Theoretically, for geometric partitions, the row-wise algorithm requires factors of 2.0 and 4.9 less communication than outer-product and column-wise algorithms, respectively; the ratios for partitions discovered by PaToH are 1.3 and 6.7 on 1331 processors.
We conclude from this data that a row-wise parallelization is the simplest parallelization that is nearly optimal for computing $A\cdot P$; this agrees with previous results~\cite{BDKS15,BSH15-TR}.

For the second \SpGEMM, we see that the outer-product, monochrome-$\A$, and monochrome-$\B$ parallelizations all perform as well as the fine-grained model (recall from \cref{fig:venn} that monochrome-$\A$ and -$\B$ parallelizations refine outer product parallelizations).
The other parallelizations all perform similarly to each other, and require about $10\times$ more communication than the fine-grained parallelization.
We note that the outer-product algorithm based on geometric partitioning (labeled ``Geometric-outer'') is (slightly) outperformed by the partitions discovered by PaToH.
For geometric partitions, the outer-product algorithm communicates a factor of $5\times$ less than the row-wise algorithm (matching previous theoretical analyses~\cite[Tab.~2]{BSH15-TR}); using hypergraph partitioning that difference increases to $7\times$ on 1331 processors.
We conclude that an outer-product parallelization is the simplest parallelization that is nearly optimal for computing $P^T\cdot(AP)$.
In our earlier paper~\cite{BDKS15}, for the second \SpGEMM{}, we compared Geometric-row with the fine-grained model and concluded that Geometric-row was suboptimal. 
The present experiments add the course-grained models and Geometric-outer and draw stronger conclusions.

The results for the SA-$\rho$AMGe problem are shown in the bottom row of \cref{fig:amgplots}.
For both \SpGEMMs, the results are qualitatively very similar to those of the model problem.
Because the mesh's geometry is not available within the AMG code, we do not have parallelizations based on geometric partitions with which to compare.
The conclusions are the same: the row-wise algorithm is nearly optimal for the case of $A \cdot P$, and the outer-product parallelization is nearly optimal for the case of $P^T \cdot (AP)$.

\subsection{Application: Linear Programming Normal Equations}
\label{sec:LP}

We next consider \SpGEMMs{} within a linear programming application, studied previously by Akbudak and Aykanat~\citeyear{AA14}, wherein a linear program is solved by an interior-point method.
During each iteration of this method, normal equations of the form $A\cdot D^2 \cdot A^T$ are constructed via \SpGEMM{}, and the resulting system is solved either directly (via Cholesky factorization) or iteratively (via a Krylov subspace method).
Boman, Parekh, and Phillips~\citeyear{BPP05} reported that these \SpGEMM{}s can dominate the overall runtime.
The matrix $A$ encodes the linear constraints and remains fixed throughout the iterations, while the positive diagonal matrix $D$ varies with each iteration.
Therefore, the \SpGEMM{}s always have the form $\A\cdot\B$, where $S_\B = S_{\A^T}$ and both $S_\A$ and $S_\B$ remain unchanged throughout the iterations.
This means there is potential to amortize the cost of forming and partitioning the \SpGEMM{} hypergraph.

Akbudak and Aykanat considered \SpGEMM{} algorithms similar to what we call outer-product (CRf), (CRr), and (CRc) --- see \cref{ex:outer}).

We consider the same five linear programming constraint matrices as Akbudak and Aykanat --- fome21, pds-80, pds-100, sgpf5y6, and cont11l --- all from the University of Florida collection~\cite{DH11}. In each case, the \SpGEMM{} dimensions $I=J < K$, i.e., the output matrix is smaller in dimension than the input matrices. 
In \cref{fig:lpplots}, we compare the fine-grained hypergraph model with the six restricted parallelizations from \cref{sec:restrict-par} for these five problems; note that since $S_\B = S_{\A^T}$, column-wise is equivalent to row-wise and monochrome-$\B$ is equivalent to monochrome-$\A$, so we omit these curves.
We perform a strong-scaling study, increasing the number of processors for fixed matrices.

In the first four problems (\cref{fig:fome21,fig:pds-80,fig:pds-100,fig:sgpf5y6}), the fine-grained, outer-product, and monochrome-$\A$ parallelizations were the most communication efficient, while the row-wise and monochrome-$\C$ parallelizations were the least.
The largest difference between row-wise and outer-product parallelizations was observed in the fourth problem~(\cref{fig:sgpf5y6}), with a communication cost ratio of $23\times$ (on 128 processors).
This trend is less clear in the fifth problem (\cref{fig:cont11l}), where all parallelizations' costs varied within a factor of two.
For all problems, we observed the differences between parallelization costs to be much less sensitive to processor count than in the following experiment (see \cref{fig:mcplots}).

Since monochrome-$\A$ includes outer-product and monochrome-$\C$ includes row-wise (in the sense of \cref{fig:venn}), these results suggest that 2D parallelizations don't provide much advantage over 1D parallelizations; on the other hand, it is important to use outer-product instead of row-wise. 
The observation in all cases that the outer-product curve closely tracks the fine-grained curve supports Akbudak and Aykanat's decision to consider only outer-product parallelizations, at least for these problems.

\begin{figure}
\renewcommand{\plotwidth}{3.3in}
\renewcommand{\plotheight}{2.3in}
\atatrue
\nolegendtrue
\centering
\subfloat[][fome21]{
	\renewcommand{\datafile}{fome21.txt}
	\begin{tikzpicture}[scale=0.75]
\begin{axis}[
	label style={font=\large},
	tick label style={font=\large},
	legend style={at={(1.1,0.5)},anchor=west,font=\large},
	legend cell align=left,
	ylabel={\yaxistitle}, 
	xlabel={Processors},
        width=\plotwidth, 
        height=\plotheight,
        	xmode=log,
	log basis x ={2},
	ymode=log,
	log basis y={2},
]	
	% plot each column of data
	\addplot[color=red,mark=x] table[x=P, y=row] {\datafile};
	\addlegendentry{Row-wise}
	\ifata
		;
	\else
		\addplot[color=red,mark=asterisk] table[x=P, y=col] {\datafile};
		\addlegendentry{Column-wise}
	\fi
	\addplot[color=red,mark=+] table[x=P, y=outer] {\datafile};
	\addlegendentry{Outer-product}
	\addplot[color=blue,mark=square*] table[x=P, y=monoa] {\datafile};
	\addlegendentry{Monochrome-$\A$}
	\ifata
		;
	\else
		\addplot[color=blue,mark=diamond*] table[x=P, y=monob] {\datafile};
		\addlegendentry{Monochrome-$\B$}
	\fi
	\addplot[color=blue,mark=triangle*] table[x=P, y=monoc] {\datafile};
	\addlegendentry{Monochrome-$\C$}
	\addplot[color=black,mark=o] table[x=P, y=fine] {\datafile};
	\addlegendentry{Fine-grained}
	\ifgeomrow
		\addplot[dashed,color=black,mark=square,every mark/.append style={solid}] table[x=P, y=geom-row] {\datafile};
		\addlegendentry{Geometric-row}
	\fi
	\ifgeomouter
		\addplot[dashed,color=black,mark=diamond,every mark/.append style={solid}] table[x=P, y=geom-outer] {\datafile};
		\addlegendentry{Geometric-outer}
	\fi

	\ifnolegend
		% suppress legend
		\legend{}
	\fi
\end{axis}
\end{tikzpicture}
	\label{fig:fome21}
}
\subfloat[][pds-80]{
	\renewcommand{\yaxistitle}{}
	\renewcommand{\datafile}{pds80.txt}
	\begin{tikzpicture}[scale=0.75]
\begin{axis}[
	label style={font=\large},
	tick label style={font=\large},
	legend style={at={(1.1,0.5)},anchor=west,font=\large},
	legend cell align=left,
	ylabel={\yaxistitle}, 
	xlabel={Processors},
        width=\plotwidth, 
        height=\plotheight,
        	xmode=log,
	log basis x ={2},
	ymode=log,
	log basis y={2},
]	
	% plot each column of data
	\addplot[color=red,mark=x] table[x=P, y=row] {\datafile};
	\addlegendentry{Row-wise}
	\ifata
		;
	\else
		\addplot[color=red,mark=asterisk] table[x=P, y=col] {\datafile};
		\addlegendentry{Column-wise}
	\fi
	\addplot[color=red,mark=+] table[x=P, y=outer] {\datafile};
	\addlegendentry{Outer-product}
	\addplot[color=blue,mark=square*] table[x=P, y=monoa] {\datafile};
	\addlegendentry{Monochrome-$\A$}
	\ifata
		;
	\else
		\addplot[color=blue,mark=diamond*] table[x=P, y=monob] {\datafile};
		\addlegendentry{Monochrome-$\B$}
	\fi
	\addplot[color=blue,mark=triangle*] table[x=P, y=monoc] {\datafile};
	\addlegendentry{Monochrome-$\C$}
	\addplot[color=black,mark=o] table[x=P, y=fine] {\datafile};
	\addlegendentry{Fine-grained}
	\ifgeomrow
		\addplot[dashed,color=black,mark=square,every mark/.append style={solid}] table[x=P, y=geom-row] {\datafile};
		\addlegendentry{Geometric-row}
	\fi
	\ifgeomouter
		\addplot[dashed,color=black,mark=diamond,every mark/.append style={solid}] table[x=P, y=geom-outer] {\datafile};
		\addlegendentry{Geometric-outer}
	\fi

	\ifnolegend
		% suppress legend
		\legend{}
	\fi
\end{axis}
\end{tikzpicture}
	\label{fig:pds-80}
} \\
\subfloat[][pds-100]{
	\renewcommand{\datafile}{pds100.txt}
	\begin{tikzpicture}[scale=0.75]
\begin{axis}[
	label style={font=\large},
	tick label style={font=\large},
	legend style={at={(1.1,0.5)},anchor=west,font=\large},
	legend cell align=left,
	ylabel={\yaxistitle}, 
	xlabel={Processors},
        width=\plotwidth, 
        height=\plotheight,
        	xmode=log,
	log basis x ={2},
	ymode=log,
	log basis y={2},
]	
	% plot each column of data
	\addplot[color=red,mark=x] table[x=P, y=row] {\datafile};
	\addlegendentry{Row-wise}
	\ifata
		;
	\else
		\addplot[color=red,mark=asterisk] table[x=P, y=col] {\datafile};
		\addlegendentry{Column-wise}
	\fi
	\addplot[color=red,mark=+] table[x=P, y=outer] {\datafile};
	\addlegendentry{Outer-product}
	\addplot[color=blue,mark=square*] table[x=P, y=monoa] {\datafile};
	\addlegendentry{Monochrome-$\A$}
	\ifata
		;
	\else
		\addplot[color=blue,mark=diamond*] table[x=P, y=monob] {\datafile};
		\addlegendentry{Monochrome-$\B$}
	\fi
	\addplot[color=blue,mark=triangle*] table[x=P, y=monoc] {\datafile};
	\addlegendentry{Monochrome-$\C$}
	\addplot[color=black,mark=o] table[x=P, y=fine] {\datafile};
	\addlegendentry{Fine-grained}
	\ifgeomrow
		\addplot[dashed,color=black,mark=square,every mark/.append style={solid}] table[x=P, y=geom-row] {\datafile};
		\addlegendentry{Geometric-row}
	\fi
	\ifgeomouter
		\addplot[dashed,color=black,mark=diamond,every mark/.append style={solid}] table[x=P, y=geom-outer] {\datafile};
		\addlegendentry{Geometric-outer}
	\fi

	\ifnolegend
		% suppress legend
		\legend{}
	\fi
\end{axis}
\end{tikzpicture}
	\label{fig:pds-100}
} 
\subfloat[][sgpf5y6]{
	\renewcommand{\yaxistitle}{}
	\renewcommand{\datafile}{sgpf5y6.txt}
	\begin{tikzpicture}[scale=0.75]
\begin{axis}[
	label style={font=\large},
	tick label style={font=\large},
	legend style={at={(1.1,0.5)},anchor=west,font=\large},
	legend cell align=left,
	ylabel={\yaxistitle}, 
	xlabel={Processors},
        width=\plotwidth, 
        height=\plotheight,
        	xmode=log,
	log basis x ={2},
	ymode=log,
	log basis y={2},
]	
	% plot each column of data
	\addplot[color=red,mark=x] table[x=P, y=row] {\datafile};
	\addlegendentry{Row-wise}
	\ifata
		;
	\else
		\addplot[color=red,mark=asterisk] table[x=P, y=col] {\datafile};
		\addlegendentry{Column-wise}
	\fi
	\addplot[color=red,mark=+] table[x=P, y=outer] {\datafile};
	\addlegendentry{Outer-product}
	\addplot[color=blue,mark=square*] table[x=P, y=monoa] {\datafile};
	\addlegendentry{Monochrome-$\A$}
	\ifata
		;
	\else
		\addplot[color=blue,mark=diamond*] table[x=P, y=monob] {\datafile};
		\addlegendentry{Monochrome-$\B$}
	\fi
	\addplot[color=blue,mark=triangle*] table[x=P, y=monoc] {\datafile};
	\addlegendentry{Monochrome-$\C$}
	\addplot[color=black,mark=o] table[x=P, y=fine] {\datafile};
	\addlegendentry{Fine-grained}
	\ifgeomrow
		\addplot[dashed,color=black,mark=square,every mark/.append style={solid}] table[x=P, y=geom-row] {\datafile};
		\addlegendentry{Geometric-row}
	\fi
	\ifgeomouter
		\addplot[dashed,color=black,mark=diamond,every mark/.append style={solid}] table[x=P, y=geom-outer] {\datafile};
		\addlegendentry{Geometric-outer}
	\fi

	\ifnolegend
		% suppress legend
		\legend{}
	\fi
\end{axis}
\end{tikzpicture}
	\label{fig:sgpf5y6}
} \\
\nolegendfalse
\subfloat[][cont11l]{
	\renewcommand{\datafile}{cont11l.txt}
	\begin{tikzpicture}[scale=0.75]
\begin{axis}[
	label style={font=\large},
	tick label style={font=\large},
	legend style={at={(1.1,0.5)},anchor=west,font=\large},
	legend cell align=left,
	ylabel={\yaxistitle}, 
	xlabel={Processors},
        width=\plotwidth, 
        height=\plotheight,
        	xmode=log,
	log basis x ={2},
	ymode=log,
	log basis y={2},
]	
	% plot each column of data
	\addplot[color=red,mark=x] table[x=P, y=row] {\datafile};
	\addlegendentry{Row-wise}
	\ifata
		;
	\else
		\addplot[color=red,mark=asterisk] table[x=P, y=col] {\datafile};
		\addlegendentry{Column-wise}
	\fi
	\addplot[color=red,mark=+] table[x=P, y=outer] {\datafile};
	\addlegendentry{Outer-product}
	\addplot[color=blue,mark=square*] table[x=P, y=monoa] {\datafile};
	\addlegendentry{Monochrome-$\A$}
	\ifata
		;
	\else
		\addplot[color=blue,mark=diamond*] table[x=P, y=monob] {\datafile};
		\addlegendentry{Monochrome-$\B$}
	\fi
	\addplot[color=blue,mark=triangle*] table[x=P, y=monoc] {\datafile};
	\addlegendentry{Monochrome-$\C$}
	\addplot[color=black,mark=o] table[x=P, y=fine] {\datafile};
	\addlegendentry{Fine-grained}
	\ifgeomrow
		\addplot[dashed,color=black,mark=square,every mark/.append style={solid}] table[x=P, y=geom-row] {\datafile};
		\addlegendentry{Geometric-row}
	\fi
	\ifgeomouter
		\addplot[dashed,color=black,mark=diamond,every mark/.append style={solid}] table[x=P, y=geom-outer] {\datafile};
		\addlegendentry{Geometric-outer}
	\fi

	\ifnolegend
		% suppress legend
		\legend{}
	\fi
\end{axis}
\end{tikzpicture}
	\label{fig:cont11l}
}
\caption{Communication costs of various hypergraph models for constructing the coefficient matrix for the normal equations in the context of the interior point methods for Linear Programming (LP --- \cref{sec:LP}).}
\label{fig:lpplots}
\atafalse
\end{figure}
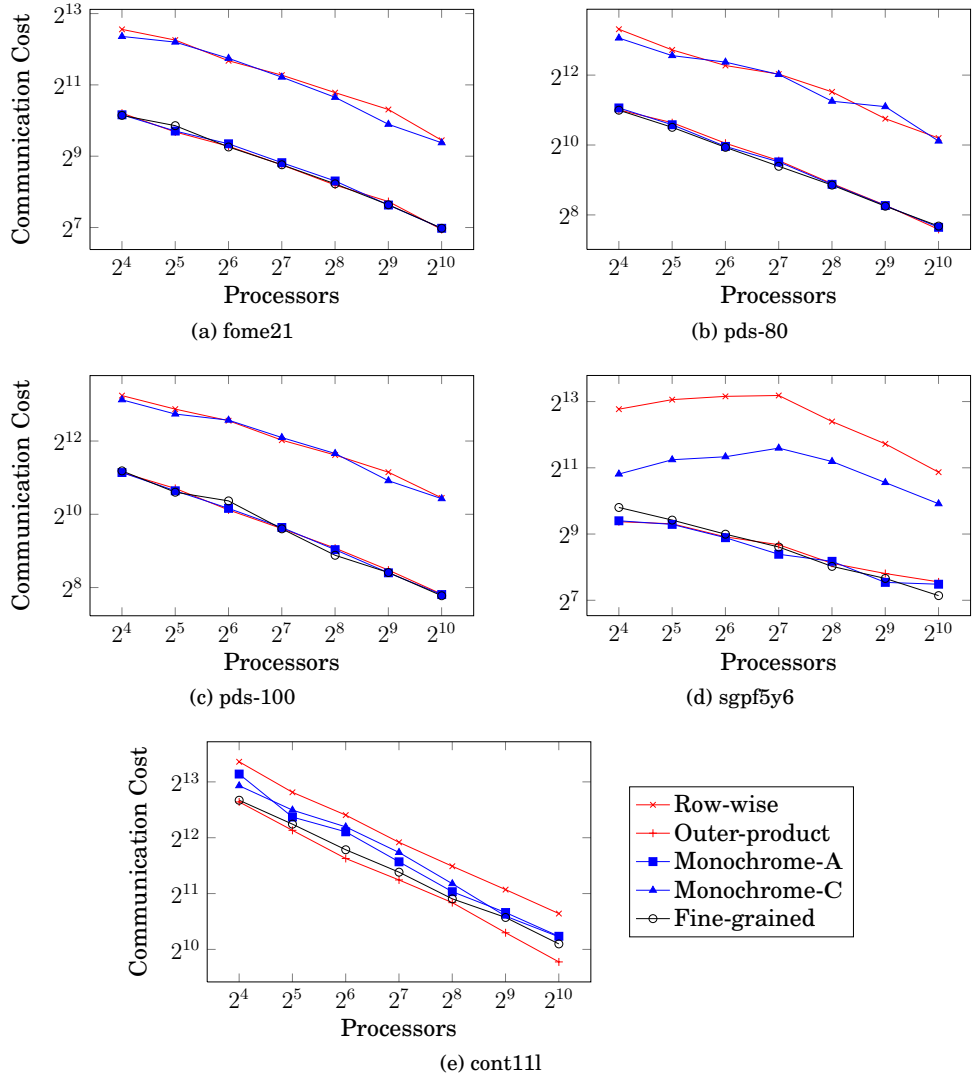

\subsection{Application: Markov Clustering}
\label{sec:MC}

Van Dongen's Markov Clustering Algorithm (MCL)~\citeyear{vanDongen00} is an iterative scheme for determining clusters in graphs.
The basic idea is to square the adjacency matrix of a graph, ``inflate'' and prune the result based on its values to maintain sparsity, and then iterate on the result.
The computational bottleneck for the algorithm is using \SpGEMM{} to square the sparse matrix, so in this section we explore squaring several different matrices coming from social-network and protein-interaction graphs, where Markov clustering is  commonly used.

We consider only the first iteration of MCL, squaring the original adjacency matrix, as a representative example of the iteration.
There have been several proposed variants of MCL, including (multi-level) regularized MCL~\cite{SP09,NLFPS14}, that perform slightly different \SpGEMM{}s.
We believe that the present experiments can help inform algorithmic choices for parallelizing any clustering algorithm applied to scale-free graphs that uses \SpGEMM{} as its computational workhorse.

Of the 11 non-synthetic data sets considered by Niu et al.~\citeyear{NLFPS14}, we present results for the 7 matrices that are publicly available and whose fine-grained hypergraphs fit in memory on our machine.
The matrices dblp, enron, and facebook are social network matrices, roadnetca is a graph representing roads and intersections, and dip, wiphi, and biogrid11 are protein-protein interaction matrices.
\Cref{fig:mcplots} presents the results.

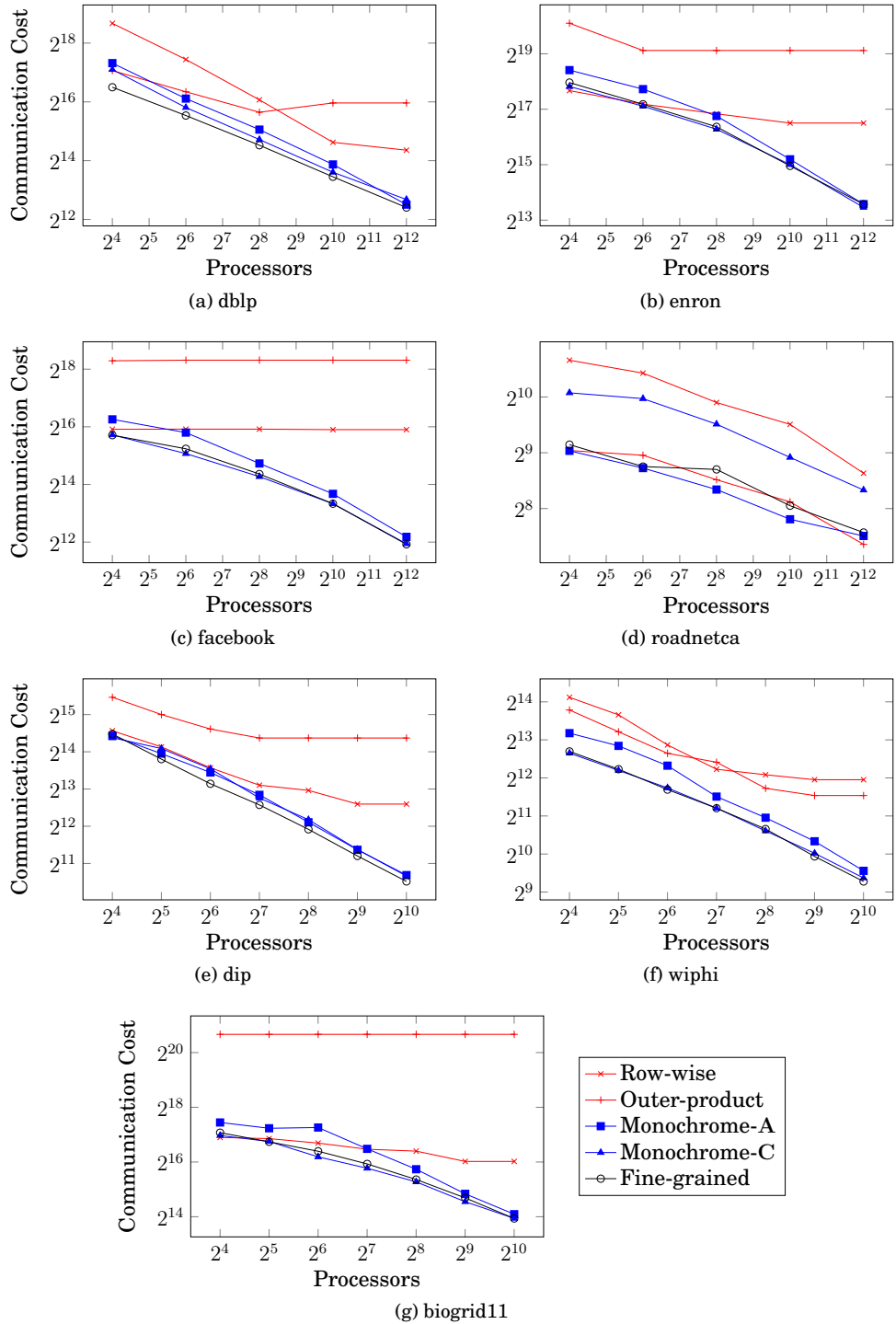
\begin{figure}
\renewcommand{\plotwidth}{3.3in}
\renewcommand{\plotheight}{2.3in}
\atatrue
\nolegendtrue
\centering
\subfloat[][dblp]{
	\renewcommand{\datafile}{dblp.txt}
	\begin{tikzpicture}[scale=0.75]
\begin{axis}[
	label style={font=\large},
	tick label style={font=\large},
	legend style={at={(1.1,0.5)},anchor=west,font=\large},
	legend cell align=left,
	ylabel={\yaxistitle}, 
	xlabel={Processors},
        width=\plotwidth, 
        height=\plotheight,
        	xmode=log,
	log basis x ={2},
	ymode=log,
	log basis y={2},
]	
	% plot each column of data
	\addplot[color=red,mark=x] table[x=P, y=row] {\datafile};
	\addlegendentry{Row-wise}
	\ifata
		;
	\else
		\addplot[color=red,mark=asterisk] table[x=P, y=col] {\datafile};
		\addlegendentry{Column-wise}
	\fi
	\addplot[color=red,mark=+] table[x=P, y=outer] {\datafile};
	\addlegendentry{Outer-product}
	\addplot[color=blue,mark=square*] table[x=P, y=monoa] {\datafile};
	\addlegendentry{Monochrome-$\A$}
	\ifata
		;
	\else
		\addplot[color=blue,mark=diamond*] table[x=P, y=monob] {\datafile};
		\addlegendentry{Monochrome-$\B$}
	\fi
	\addplot[color=blue,mark=triangle*] table[x=P, y=monoc] {\datafile};
	\addlegendentry{Monochrome-$\C$}
	\addplot[color=black,mark=o] table[x=P, y=fine] {\datafile};
	\addlegendentry{Fine-grained}
	\ifgeomrow
		\addplot[dashed,color=black,mark=square,every mark/.append style={solid}] table[x=P, y=geom-row] {\datafile};
		\addlegendentry{Geometric-row}
	\fi
	\ifgeomouter
		\addplot[dashed,color=black,mark=diamond,every mark/.append style={solid}] table[x=P, y=geom-outer] {\datafile};
		\addlegendentry{Geometric-outer}
	\fi

	\ifnolegend
		% suppress legend
		\legend{}
	\fi
\end{axis}
\end{tikzpicture}
	\label{fig:dblp}
}
\subfloat[][enron]{
	\renewcommand{\yaxistitle}{}
	\renewcommand{\datafile}{enron.txt}
	\begin{tikzpicture}[scale=0.75]
\begin{axis}[
	label style={font=\large},
	tick label style={font=\large},
	legend style={at={(1.1,0.5)},anchor=west,font=\large},
	legend cell align=left,
	ylabel={\yaxistitle}, 
	xlabel={Processors},
        width=\plotwidth, 
        height=\plotheight,
        	xmode=log,
	log basis x ={2},
	ymode=log,
	log basis y={2},
]	
	% plot each column of data
	\addplot[color=red,mark=x] table[x=P, y=row] {\datafile};
	\addlegendentry{Row-wise}
	\ifata
		;
	\else
		\addplot[color=red,mark=asterisk] table[x=P, y=col] {\datafile};
		\addlegendentry{Column-wise}
	\fi
	\addplot[color=red,mark=+] table[x=P, y=outer] {\datafile};
	\addlegendentry{Outer-product}
	\addplot[color=blue,mark=square*] table[x=P, y=monoa] {\datafile};
	\addlegendentry{Monochrome-$\A$}
	\ifata
		;
	\else
		\addplot[color=blue,mark=diamond*] table[x=P, y=monob] {\datafile};
		\addlegendentry{Monochrome-$\B$}
	\fi
	\addplot[color=blue,mark=triangle*] table[x=P, y=monoc] {\datafile};
	\addlegendentry{Monochrome-$\C$}
	\addplot[color=black,mark=o] table[x=P, y=fine] {\datafile};
	\addlegendentry{Fine-grained}
	\ifgeomrow
		\addplot[dashed,color=black,mark=square,every mark/.append style={solid}] table[x=P, y=geom-row] {\datafile};
		\addlegendentry{Geometric-row}
	\fi
	\ifgeomouter
		\addplot[dashed,color=black,mark=diamond,every mark/.append style={solid}] table[x=P, y=geom-outer] {\datafile};
		\addlegendentry{Geometric-outer}
	\fi

	\ifnolegend
		% suppress legend
		\legend{}
	\fi
\end{axis}
\end{tikzpicture}
	\label{fig:enron}
} \\
\subfloat[][facebook]{
	\renewcommand{\datafile}{facebook.txt}
	\begin{tikzpicture}[scale=0.75]
\begin{axis}[
	label style={font=\large},
	tick label style={font=\large},
	legend style={at={(1.1,0.5)},anchor=west,font=\large},
	legend cell align=left,
	ylabel={\yaxistitle}, 
	xlabel={Processors},
        width=\plotwidth, 
        height=\plotheight,
        	xmode=log,
	log basis x ={2},
	ymode=log,
	log basis y={2},
]	
	% plot each column of data
	\addplot[color=red,mark=x] table[x=P, y=row] {\datafile};
	\addlegendentry{Row-wise}
	\ifata
		;
	\else
		\addplot[color=red,mark=asterisk] table[x=P, y=col] {\datafile};
		\addlegendentry{Column-wise}
	\fi
	\addplot[color=red,mark=+] table[x=P, y=outer] {\datafile};
	\addlegendentry{Outer-product}
	\addplot[color=blue,mark=square*] table[x=P, y=monoa] {\datafile};
	\addlegendentry{Monochrome-$\A$}
	\ifata
		;
	\else
		\addplot[color=blue,mark=diamond*] table[x=P, y=monob] {\datafile};
		\addlegendentry{Monochrome-$\B$}
	\fi
	\addplot[color=blue,mark=triangle*] table[x=P, y=monoc] {\datafile};
	\addlegendentry{Monochrome-$\C$}
	\addplot[color=black,mark=o] table[x=P, y=fine] {\datafile};
	\addlegendentry{Fine-grained}
	\ifgeomrow
		\addplot[dashed,color=black,mark=square,every mark/.append style={solid}] table[x=P, y=geom-row] {\datafile};
		\addlegendentry{Geometric-row}
	\fi
	\ifgeomouter
		\addplot[dashed,color=black,mark=diamond,every mark/.append style={solid}] table[x=P, y=geom-outer] {\datafile};
		\addlegendentry{Geometric-outer}
	\fi

	\ifnolegend
		% suppress legend
		\legend{}
	\fi
\end{axis}
\end{tikzpicture}
	\label{fig:facebook}
} 
\subfloat[][roadnetca]{
	\renewcommand{\yaxistitle}{}
	\renewcommand{\datafile}{roadnetca.txt}
	\begin{tikzpicture}[scale=0.75]
\begin{axis}[
	label style={font=\large},
	tick label style={font=\large},
	legend style={at={(1.1,0.5)},anchor=west,font=\large},
	legend cell align=left,
	ylabel={\yaxistitle}, 
	xlabel={Processors},
        width=\plotwidth, 
        height=\plotheight,
        	xmode=log,
	log basis x ={2},
	ymode=log,
	log basis y={2},
]	
	% plot each column of data
	\addplot[color=red,mark=x] table[x=P, y=row] {\datafile};
	\addlegendentry{Row-wise}
	\ifata
		;
	\else
		\addplot[color=red,mark=asterisk] table[x=P, y=col] {\datafile};
		\addlegendentry{Column-wise}
	\fi
	\addplot[color=red,mark=+] table[x=P, y=outer] {\datafile};
	\addlegendentry{Outer-product}
	\addplot[color=blue,mark=square*] table[x=P, y=monoa] {\datafile};
	\addlegendentry{Monochrome-$\A$}
	\ifata
		;
	\else
		\addplot[color=blue,mark=diamond*] table[x=P, y=monob] {\datafile};
		\addlegendentry{Monochrome-$\B$}
	\fi
	\addplot[color=blue,mark=triangle*] table[x=P, y=monoc] {\datafile};
	\addlegendentry{Monochrome-$\C$}
	\addplot[color=black,mark=o] table[x=P, y=fine] {\datafile};
	\addlegendentry{Fine-grained}
	\ifgeomrow
		\addplot[dashed,color=black,mark=square,every mark/.append style={solid}] table[x=P, y=geom-row] {\datafile};
		\addlegendentry{Geometric-row}
	\fi
	\ifgeomouter
		\addplot[dashed,color=black,mark=diamond,every mark/.append style={solid}] table[x=P, y=geom-outer] {\datafile};
		\addlegendentry{Geometric-outer}
	\fi

	\ifnolegend
		% suppress legend
		\legend{}
	\fi
\end{axis}
\end{tikzpicture}
	\label{fig:roadnetca}
} \\
\subfloat[][dip]{
	\renewcommand{\datafile}{dip.txt}
	\begin{tikzpicture}[scale=0.75]
\begin{axis}[
	label style={font=\large},
	tick label style={font=\large},
	legend style={at={(1.1,0.5)},anchor=west,font=\large},
	legend cell align=left,
	ylabel={\yaxistitle}, 
	xlabel={Processors},
        width=\plotwidth, 
        height=\plotheight,
        	xmode=log,
	log basis x ={2},
	ymode=log,
	log basis y={2},
]	
	% plot each column of data
	\addplot[color=red,mark=x] table[x=P, y=row] {\datafile};
	\addlegendentry{Row-wise}
	\ifata
		;
	\else
		\addplot[color=red,mark=asterisk] table[x=P, y=col] {\datafile};
		\addlegendentry{Column-wise}
	\fi
	\addplot[color=red,mark=+] table[x=P, y=outer] {\datafile};
	\addlegendentry{Outer-product}
	\addplot[color=blue,mark=square*] table[x=P, y=monoa] {\datafile};
	\addlegendentry{Monochrome-$\A$}
	\ifata
		;
	\else
		\addplot[color=blue,mark=diamond*] table[x=P, y=monob] {\datafile};
		\addlegendentry{Monochrome-$\B$}
	\fi
	\addplot[color=blue,mark=triangle*] table[x=P, y=monoc] {\datafile};
	\addlegendentry{Monochrome-$\C$}
	\addplot[color=black,mark=o] table[x=P, y=fine] {\datafile};
	\addlegendentry{Fine-grained}
	\ifgeomrow
		\addplot[dashed,color=black,mark=square,every mark/.append style={solid}] table[x=P, y=geom-row] {\datafile};
		\addlegendentry{Geometric-row}
	\fi
	\ifgeomouter
		\addplot[dashed,color=black,mark=diamond,every mark/.append style={solid}] table[x=P, y=geom-outer] {\datafile};
		\addlegendentry{Geometric-outer}
	\fi

	\ifnolegend
		% suppress legend
		\legend{}
	\fi
\end{axis}
\end{tikzpicture}
	\label{fig:dip}
} 
\subfloat[][wiphi]{
	\renewcommand{\yaxistitle}{}
	\renewcommand{\datafile}{wiphi.txt}
	\begin{tikzpicture}[scale=0.75]
\begin{axis}[
	label style={font=\large},
	tick label style={font=\large},
	legend style={at={(1.1,0.5)},anchor=west,font=\large},
	legend cell align=left,
	ylabel={\yaxistitle}, 
	xlabel={Processors},
        width=\plotwidth, 
        height=\plotheight,
        	xmode=log,
	log basis x ={2},
	ymode=log,
	log basis y={2},
]	
	% plot each column of data
	\addplot[color=red,mark=x] table[x=P, y=row] {\datafile};
	\addlegendentry{Row-wise}
	\ifata
		;
	\else
		\addplot[color=red,mark=asterisk] table[x=P, y=col] {\datafile};
		\addlegendentry{Column-wise}
	\fi
	\addplot[color=red,mark=+] table[x=P, y=outer] {\datafile};
	\addlegendentry{Outer-product}
	\addplot[color=blue,mark=square*] table[x=P, y=monoa] {\datafile};
	\addlegendentry{Monochrome-$\A$}
	\ifata
		;
	\else
		\addplot[color=blue,mark=diamond*] table[x=P, y=monob] {\datafile};
		\addlegendentry{Monochrome-$\B$}
	\fi
	\addplot[color=blue,mark=triangle*] table[x=P, y=monoc] {\datafile};
	\addlegendentry{Monochrome-$\C$}
	\addplot[color=black,mark=o] table[x=P, y=fine] {\datafile};
	\addlegendentry{Fine-grained}
	\ifgeomrow
		\addplot[dashed,color=black,mark=square,every mark/.append style={solid}] table[x=P, y=geom-row] {\datafile};
		\addlegendentry{Geometric-row}
	\fi
	\ifgeomouter
		\addplot[dashed,color=black,mark=diamond,every mark/.append style={solid}] table[x=P, y=geom-outer] {\datafile};
		\addlegendentry{Geometric-outer}
	\fi

	\ifnolegend
		% suppress legend
		\legend{}
	\fi
\end{axis}
\end{tikzpicture}
	\label{fig:wiphi}
} \\
\nolegendfalse
\subfloat[][biogrid11]{
	\renewcommand{\datafile}{biogrid11.txt}
	\begin{tikzpicture}[scale=0.75]
\begin{axis}[
	label style={font=\large},
	tick label style={font=\large},
	legend style={at={(1.1,0.5)},anchor=west,font=\large},
	legend cell align=left,
	ylabel={\yaxistitle}, 
	xlabel={Processors},
        width=\plotwidth, 
        height=\plotheight,
        	xmode=log,
	log basis x ={2},
	ymode=log,
	log basis y={2},
]	
	% plot each column of data
	\addplot[color=red,mark=x] table[x=P, y=row] {\datafile};
	\addlegendentry{Row-wise}
	\ifata
		;
	\else
		\addplot[color=red,mark=asterisk] table[x=P, y=col] {\datafile};
		\addlegendentry{Column-wise}
	\fi
	\addplot[color=red,mark=+] table[x=P, y=outer] {\datafile};
	\addlegendentry{Outer-product}
	\addplot[color=blue,mark=square*] table[x=P, y=monoa] {\datafile};
	\addlegendentry{Monochrome-$\A$}
	\ifata
		;
	\else
		\addplot[color=blue,mark=diamond*] table[x=P, y=monob] {\datafile};
		\addlegendentry{Monochrome-$\B$}
	\fi
	\addplot[color=blue,mark=triangle*] table[x=P, y=monoc] {\datafile};
	\addlegendentry{Monochrome-$\C$}
	\addplot[color=black,mark=o] table[x=P, y=fine] {\datafile};
	\addlegendentry{Fine-grained}
	\ifgeomrow
		\addplot[dashed,color=black,mark=square,every mark/.append style={solid}] table[x=P, y=geom-row] {\datafile};
		\addlegendentry{Geometric-row}
	\fi
	\ifgeomouter
		\addplot[dashed,color=black,mark=diamond,every mark/.append style={solid}] table[x=P, y=geom-outer] {\datafile};
		\addlegendentry{Geometric-outer}
	\fi

	\ifnolegend
		% suppress legend
		\legend{}
	\fi
\end{axis}
\end{tikzpicture}
	\label{fig:biogrid11}
}
\caption{Communication costs of various hypergraph models for squaring symmetric matrices in the context of the Markov Clustering Algorithm (MCL --- \cref{sec:MC}).}
\label{fig:mcplots}
\atafalse
\end{figure}

We consider the fine-grained hypergraph model as well as four of the restricted parallelizations: row-wise, outer-product, monochrome-$\A$, and monochrome-$\B$.
Because we're squaring symmetric matrices in this section, column-wise and monochrome-$\B$ parallelizations are equivalent.
Note that we do not exploit symmetry in these experiments.
As in \cref{sec:LP}, we perform a strong-scaling study.

The overall conclusion from \cref{fig:mcplots} is that 2D algorithms seem to require significantly less communication than 1D algorithms, particularly at high processor counts.
The largest difference was for the facebook matrix, where the ratio of communication costs between monochrome-$\C$ and outer-product was $83\times$ (on 4096 processors).
Furthermore, the downward trend of 2D and 3D plot lines show that per-processor communication costs decreased as the number of processors increased, indicating scalability.
The 1D plot lines are mostly horizontal, implying that 1D algorithms are not scalable for these problems. 
The partitions of the 1D models we obtained from the partitioner violated our load-balance constraint ($\epsilon = 0.01$) and were  increasingly imbalanced as we increased the number of parts. 
We attribute this to the presence of ``heavy'' vertices whose weights exceeded the maximum per-part weight prescribed by our constraint.
The exception to these trends is the roadnetca matrix, which is qualitatively different from the social network and protein-protein interaction matrices (we include it for completeness). 

We mention that Faisal, Parthasarathy, and Sadayappan~\citeyear{FPS14} parallelized regularized MCL within a graph framework.
Their MPI-based algorithm used a 1D (row-wise) parallelization, but the data from these experiments suggests that a 2D parallelization would likely perform better for these types of matrices.
This conclusion agrees with Boman, Devine, and Rajamanickam's premise~\citeyear{BDR13} in the context of SpMV: when $S_\A$ is a scale-free graph, 1D partitions of the \SpMV{} hypergraph are often sub-optimal.
% We note that the 2D partitioning technique proposed in that paper avoids the overhead of partitioning the fine-grain SpMV hypergraph model; here we are not as concerned with the running time of the hypergraph partitioner for \SpGEMM.

\section{Conclusion}
\label{sec:conclusion}

We have proposed a framework for studying \SpGEMM{} algorithms, called the fine-grained hypergraph model, which captures communication and computation costs, as well as memory usage (see \cref{sec:model}).
We have defined hypergraph partitioning problems whose solutions yield lower bounds on communication costs on both parallel and sequential machines, and algorithms that attain these lower bounds within constant factors (see \cref{sec:LB}).
We have also shown how to simplify the fine-grained hypergraph model, applying vertex coarsening to focus on special classes of \SpGEMM{} algorithms such as 1D and 2D classes (see \cref{sec:simplify-HG}). 
We applied these theoretical tools to real-world \SpGEMM{} problems, experimentally confirming that hypergraph partitioning helps identify algorithms with reduced communication costs (see \cref{sec:expt}).

A key feature of our approach to \SpGEMM{} analysis is that it is sparsity dependent.
The fact that our communication lower bounds, parameterized by the nonzero structures of the input matrices, are tight suggests that a simpler parameterization --- in terms of numbers of nonzeros, number of arithmetic operations, etc.\ --- is unlikely. 
However, for specific classes of instances, simpler tight lower bounds are possible: e.g., in the case where the input matrices are both dense, tight lower bounds are parameterized by the matrix dimensions (see, e.g.,~\cite{BCDH+14}).

The class of \SpGEMM{} algorithms studied in this paper is quite general, and efficiently implementing a generic member of this class remains a practical challenge.
On the other hand, implementations of 1D and 2D algorithms are much simpler, and our experiments demonstrated that in some cases these simpler approaches suffice to minimize communication. 
For example, in the AMG and LP applications (\cref{sec:AMG,sec:LP}), we observed that some (but not all) of the 1D algorithms could minimize communication.
On the other hand, in the MCL application (\cref{sec:MC}), 2D algorithms minimized communication and 1D algorithms did not, assuming that the partition of the fine-grained model is optimal, or nearly so.

We have proposed hypergraph analysis as a practical tool for \SpGEMM{} algorithm design.
We do not advocate solving hypergraph partitioning problems on the fly: this cost could dominate that of the actual \SpGEMM{} computation.
Our approach would be most practical when the preprocessing overhead can be amortized over many \SpGEMM{}s, like when a sequence of matrices with the same (or similar) nonzero structure are to be multiplied (see, e.g.,~\cite{BPP05,KKV13}).

One direction for future work is refining our parallel communication cost model to capture latency cost (number of messages) in addition to bandwidth cost (number of words).
It is straightforward to derive a lower bound on latency cost by modifying \cref{lem:per-proc-comm} to count the number of adjacent parts instead of the number of adjacent nets.
On the other hand, the algorithm in \cref{lem:parUB} may not attain the latency lower bound, due to sending messages of containing single words. 

Another interesting extension to our model is considering parallel machines with shared memory.
A promising approach is to generalize our sequential model to have multiple processors, each with their own fast memory.
Of course, Hong and Kung's communication lower bound result, invoked in \cref{thm:HK}, must be replaced by a shared-memory extension, of which several have appeared in the literature. 

Another direction for future work is considering additional coarsenings of the fine-grained model.
In the case of \SpMV{}, ``regular'' partitions have been proposed to limit latency costs~\cite{VB05,BDR13}.
In the case of \SpGEMM{},  SpSUMMA~\cite{BG12} corresponds to a regular partition, a special case of monochrome-$\C$. 
%
%and it would be interesting to see if similar latency savings apply here as well.
%
Coarsening the fine-grained model and restricting to regular partitions yields simpler algorithms that may also exhibit reduced latency cost.

\begin{acks}
We thank Erik Boman and Siva Rajamanickam for helpful discussions, particularly their suggestion to augment the hypergraph model with nonzero vertices.

% Grey's ack
This research was supported in part by an appointment to the Sandia National Laboratories Truman Fellowship in National Security Science and Engineering, sponsored by Sandia Corporation (a wholly owned subsidiary of Lockheed Martin Corporation) as Operator of Sandia National Laboratories under its U.S.\ Department of Energy Contract No.\ DE-AC04-94AL85000.

% Oded's acks.
Research is supported by grants 1878/14 and 1901/14 from the Israel Science Foundation (founded by the Israel Academy of Sciences and Humanities) and grant 3-10891 from the Ministry of Science and Technology, Israel. Research is also supported by the Einstein Foundation and the Minerva Foundation. 
This paper is supported by the Intel Collaborative Research Institute for Computational Intelligence (ICRI-CI). This research was supported by a grant from the United States-Israel Binational Science Foundation (BSF), Jerusalem, Israel.

% The rules on ack'ing the DOE are:
% http://science.energy.gov/funding-opportunities/acknowledgements/

% Alex's ack
This material is based upon work supported by the U.S.\ Department of Energy, Office of Science, Office of Advanced Scientific Computing Research.

% NERSC's ack.
This research used resources of the National Energy Research Scientific Computing Center, which is a DOE Office of Science User Facility.
\end{acks}

% Bibliography
\bibliographystyle{ACM-Reference-Format-Journals}
%\bibliography{paper}

\begin{thebibliography}{00}

%%% the Web address.

\ifx \showCODEN    \undefined \def \showCODEN     #1{\unskip}     \fi
\ifx \showDOI      \undefined \def \showDOI       #1{{\tt DOI:}\penalty0{#1}\ }
  \fi
\ifx \showISBNx    \undefined \def \showISBNx     #1{\unskip}     \fi
\ifx \showISBNxiii \undefined \def \showISBNxiii  #1{\unskip}     \fi
\ifx \showISSN     \undefined \def \showISSN      #1{\unskip}     \fi
\ifx \showLCCN     \undefined \def \showLCCN      #1{\unskip}     \fi
\ifx \shownote     \undefined \def \shownote      #1{#1}          \fi
\ifx \showarticletitle \undefined \def \showarticletitle #1{#1}   \fi
\ifx \showURL      \undefined \def \showURL       #1{#1}          \fi

\bibitem[\protect\citeauthoryear{Akbudak and Aykanat}{Akbudak and
  Aykanat}{2014}]%
        {AA14}
{K. Akbudak} {and} {C. Aykanat}. 2014.
\newblock \showarticletitle{Simultaneous input and output matrix partitioning
  for outer-product--parallel sparse matrix-matrix multiplication}.
\newblock {\em SIAM Journal on Scientific Computing\/} {36}, 5 (2014),
  C568--C590.
\newblock
\showDOI{%
\url{http://dx.doi.org/10.1137/13092589X}}


\bibitem[\protect\citeauthoryear{Akbudak, Kayaaslan, and Aykanat}{Akbudak
  et~al\mbox{.}}{2013}]%
        {AKA13}
{K. Akbudak}, {E. Kayaaslan}, {and} {C. Aykanat}. 2013.
\newblock \showarticletitle{Hypergraph partitioning based models and methods
  for exploiting cache locality in sparse matrix-vector multiplication}.
\newblock {\em SIAM Journal on Scientific Computing\/} {35}, 3 (2013),
  C237--C262.
\newblock
\showDOI{%
\url{http://dx.doi.org/10.1137/100813956}}


\bibitem[\protect\citeauthoryear{Azad, Ballard, Buluc, Demmel, Grigori,
  Schwartz, Toledo, and Williams}{Azad et~al\mbox{.}}{2015a}]%
        {AB+15-TR}
{A. Azad}, {G. Ballard}, {A. Buluc}, {J. Demmel}, {L. Grigori}, {O. Schwartz},
  {S. Toledo}, {and} {S. Williams}. 2015a.
\newblock {\em Exploiting multiple levels of parallelism in sparse
  matrix-matrix multiplication}.
\newblock {T}echnical {R}eport 1510.00844. arXiv.
\newblock
\showURL{%
\url{http://arxiv.org/abs/1510.00844}}


\bibitem[\protect\citeauthoryear{Azad, Bulu\c{c}, and Gilbert}{Azad
  et~al\mbox{.}}{2015b}]%
        {ABG15}
{A. Azad}, {A. Bulu\c{c}}, {and} {J. Gilbert}. 2015b.
\newblock \showarticletitle{Parallel triangle counting and enumeration using
  matrix algebra}. In {\em Proceedings of the IPDPSW, Workshop on Graph
  Algorithm Building Blocks} {\em (GABB '15)}. 804 -- 811.
\newblock
\showDOI{%
\url{http://dx.doi.org/10.1109/IPDPSW.2015.75}}


\bibitem[\protect\citeauthoryear{Ballard, Bulu\c{c}, Demmel, Grigori, Lipshitz,
  Schwartz, and Toledo}{Ballard et~al\mbox{.}}{2013}]%
        {BBDG+13}
{G. Ballard}, {A. Bulu\c{c}}, {J. Demmel}, {L. Grigori}, {B. Lipshitz}, {O.
  Schwartz}, {and} {S. Toledo}. 2013.
\newblock \showarticletitle{Communication optimal parallel multiplication of
  sparse random matrices}. In {\em Proceedings of the 25th ACM Symposium on
  Parallelism in Algorithms and Architectures} {\em (SPAA '13)}. ACM, 222--231.
\newblock
\showISBNx{978-1-4503-1572-2}
\showDOI{%
\url{http://dx.doi.org/10.1145/2486159.2486196}}


\bibitem[\protect\citeauthoryear{Ballard, Carson, Demmel, Hoemmen, Knight, and
  Schwartz}{Ballard et~al\mbox{.}}{2014}]%
        {BCDH+14}
{G. Ballard}, {E. Carson}, {J. Demmel}, {M. Hoemmen}, {N. Knight}, {and} {O.
  Schwartz}. 2014.
\newblock \showarticletitle{Communication lower bounds and optimal algorithms
  for numerical linear algebra}.
\newblock {\em Acta Numerica\/}  {23} (May 2014), 1--155.
\newblock
\showISSN{1474-0508}
\showDOI{%
\url{http://dx.doi.org/10.1017/S0962492914000038}}


\bibitem[\protect\citeauthoryear{Ballard, Demmel, Holtz, Lipshitz, and
  Schwartz}{Ballard et~al\mbox{.}}{2012}]%
        {BDHLS12}
{G. Ballard}, {J. Demmel}, {O. Holtz}, {B. Lipshitz}, {and} {O. Schwartz}.
  2012.
\newblock \showarticletitle{Brief announcement: strong scaling of matrix
  multiplication algorithms and memory-independent communication lower bounds}.
  In {\em Proceedings of the 24th ACM Symposium on Parallelism in Algorithms
  and Architectures} {\em (SPAA '12)}. ACM, New York, NY, USA, 77--79.
\newblock
\showISBNx{978-1-4503-1213-4}
\showDOI{%
\url{http://dx.doi.org/10.1145/2312005.2312021}}


\bibitem[\protect\citeauthoryear{Ballard, Demmel, Holtz, and Schwartz}{Ballard
  et~al\mbox{.}}{2011}]%
        {BDHS11}
{G. Ballard}, {J. Demmel}, {O. Holtz}, {and} {O. Schwartz}. 2011.
\newblock \showarticletitle{Minimizing communication in numerical linear
  algebra}.
\newblock {\em {SIAM Journal on Matrix Analysis and Applications}\/} {32}, 3
  (September 2011), 866--901.
\newblock
\showDOI{%
\url{http://dx.doi.org/10.1137/090769156}}


\bibitem[\protect\citeauthoryear{Ballard, Druinsky, Knight, and
  Schwartz}{Ballard et~al\mbox{.}}{2015a}]%
        {BDKS15}
{G. Ballard}, {A. Druinsky}, {N. Knight}, {and} {O. Schwartz}. 2015a.
\newblock \showarticletitle{Brief announcement: hypergraph partitioning for
  parallel sparse matrix-matrix multiplication}. In {\em Proceedings of the
  27th ACM Symposium on Parallelism in Algorithms and Architectures} {\em (SPAA
  '15)}. ACM, New York, NY, USA, 86--88.
\newblock
\showISBNx{978-1-4503-3588-1}
\showDOI{%
\url{http://dx.doi.org/10.1145/2755573.2755613}}


\bibitem[\protect\citeauthoryear{Ballard, Siefert, and Hu}{Ballard
  et~al\mbox{.}}{2015b}]%
        {BSH15-TR}
{G. Ballard}, {C. Siefert}, {and} {J. Hu}. 2015b.
\newblock {\em Reducing communication costs for sparse matrix multiplication
  within algebraic multigrid}.
\newblock {T}echnical {R}eport SAND2015-3275. Sandia National Laboratories.
\newblock
\showURL{%
\url{http://prod.sandia.gov/techlib/access-control.cgi/2015/153275.pdf}}


\bibitem[\protect\citeauthoryear{Barker, Kalchev, Mishev, Vassilevski, and
  Yang}{Barker et~al\mbox{.}}{2015}]%
        {BKMVY15}
{A. Barker}, {D. Kalchev}, {I. Mishev}, {P. Vassilevski}, {and} {Y. Yang}.
  2015.
\newblock \showarticletitle{Accurate coarse-scale {AMG}-based finite volume
  reservoir simulations in highly heterogeneous media}. In {\em SPE Reservoir
  Simulation Symposium}. SPE--173277--MS.
\newblock
\showDOI{%
\url{http://dx.doi.org/10.2118/173277-MS}}


\bibitem[\protect\citeauthoryear{Boman, Devine, Fisk, Heaphy, Hendrickson,
  Vaughan, {\c C}ataly{\"u}rek, Bozdag, Mitchell, and Teresco}{Boman
  et~al\mbox{.}}{2007}]%
        {BDFHH07}
{E. Boman}, {K. Devine}, {L. Fisk}, {R. Heaphy}, {B. Hendrickson}, {C Vaughan},
  {{\"U}. {\c C}ataly{\"u}rek}, {D. Bozdag}, {W. Mitchell}, {and} {J. Teresco}.
  2007.
\newblock {\em Zoltan 3.0: Parallel Partitioning, Load Balancing, and
  Data-Management Services; User's Guide}.
\newblock {T}echnical {R}eport SAND2007-4748W. Sandia National Laboratories.
\newblock
\showURL{%
\url{http://www.cs.sandia.gov/Zoltan/ug_html/ug.html}}


\bibitem[\protect\citeauthoryear{Boman, Devine, and Rajamanickam}{Boman
  et~al\mbox{.}}{2013}]%
        {BDR13}
{E. Boman}, {K. Devine}, {and} {S. Rajamanickam}. 2013.
\newblock \showarticletitle{Scalable matrix computations on large scale-free
  graphs using {2D} graph partitioning}. In {\em Proceedings of the
  International Conference on High Performance Computing, Networking, Storage
  and Analysis} {\em (SC '13)}. ACM, New York, NY, USA, Article 50, 12 pages.
\newblock
\showISBNx{978-1-4503-2378-9}
\showDOI{%
\url{http://dx.doi.org/10.1145/2503210.2503293}}


\bibitem[\protect\citeauthoryear{Boman, Parekh, and Phillips}{Boman
  et~al\mbox{.}}{2005}]%
        {BPP05}
{E. Boman}, {O. Parekh}, {and} {C. Phillips}. 2005.
\newblock {\em {LDRD} final report on massively-parallel linear programming:
  the {parPCx} system}.
\newblock {T}echnical {R}eport SAND2004-6440. Sandia National Laboratories.
\newblock
\showURL{%
\url{http://prod.sandia.gov/techlib/access-control.cgi/2004/046440.pdf}}


\bibitem[\protect\citeauthoryear{Bor\v{s}tnik, VandeVondele, Weber, and
  Hutter}{Bor\v{s}tnik et~al\mbox{.}}{2014}]%
        {BVWH14}
{U. Bor\v{s}tnik}, {J. VandeVondele}, {V. Weber}, {and} {J. Hutter}. 2014.
\newblock \showarticletitle{Sparse matrix multiplication: the distributed
  block-compressed sparse row library}.
\newblock {\em {Parallel Computing}\/} {40}, 5--6 (2014), 47--58.
\newblock
\showISSN{0167-8191}
\showDOI{%
\url{http://dx.doi.org/10.1016/j.parco.2014.03.012}}


\bibitem[\protect\citeauthoryear{Brezina and Vassilevski}{Brezina and
  Vassilevski}{2011}]%
        {BV11}
{M. Brezina} {and} {P. Vassilevski}. 2011.
\newblock \showarticletitle{Smoothed aggregation spectral element agglomeration
  {AMG}: {SA-$\rho$AMGe}}. In {\em Proceedings of the 8th International
  Conference on Large-Scale Scientific Computing (LSSC)} {\em (Lecture Notes in
  Computer Science)}, {I.~Lirkov}, {S.~Margenov}, {and} {J.~Wa{\'s}niewski}
  (Eds.), Vol. 7116. Springer, 3--15.
\newblock
\showDOI{%
\url{http://dx.doi.org/10.1007/978-3-642-29843-1_1}}


\bibitem[\protect\citeauthoryear{Buluc and Gilbert}{Buluc and Gilbert}{2008}]%
        {BG08a}
{A. Buluc} {and} {J. Gilbert}. 2008.
\newblock \showarticletitle{On the representation and multiplication of
  hypersparse matrices}. In {\em IEEE International Symposium on Parallel and
  Distributed Processing} {\em (IPDPS '08)}. 1--11.
\newblock
\showISSN{1530-2075}
\showDOI{%
\url{http://dx.doi.org/10.1109/IPDPS.2008.4536313}}


\bibitem[\protect\citeauthoryear{Bulu{\c c} and Gilbert}{Bulu{\c c} and
  Gilbert}{2012}]%
        {BG12}
{A. Bulu{\c c}} {and} {J. Gilbert}. 2012.
\newblock \showarticletitle{Parallel sparse matrix-{\allowbreak}matrix
  multiplication and indexing: implementation and experiments}.
\newblock {\em SIAM Journal on Scientific Computing\/} {34}, 4 (2012),
  C170--C191.
\newblock
\showDOI{%
\url{http://dx.doi.org/10.1137/110848244}}


\bibitem[\protect\citeauthoryear{{\c C}ataly{\"u}rek and Aykanat}{{\c
  C}ataly{\"u}rek and Aykanat}{1999}]%
        {PaToH}
{{\"U}. {\c C}ataly{\"u}rek} {and} {C. Aykanat}. 1999.
\newblock {\em {PaToH}: Partitioning Tool for Hypergraphs}.
\newblock {T}echnical {R}eport. (Revised March 2011).
\newblock
\showURL{%
\url{http://bmi.osu.edu/umit/PaToH/manual.pdf}}


\bibitem[\protect\citeauthoryear{{\c C}ataly{\"u}rek and Aykanat}{{\c
  C}ataly{\"u}rek and Aykanat}{2001}]%
        {CA01a}
{{\" U}. {\c C}ataly{\"u}rek} {and} {C. Aykanat}. 2001.
\newblock \showarticletitle{A fine-grain hypergraph model for {2D}
  decomposition of sparse matrices}. In {\em Proceedings of the 15th
  International Parallel and Distributed Processing Symposium} {\em (IPDPS
  '01)}. 118--123.
\newblock
\showISBNx{0-7695-0990-8}
\showDOI{%
\url{http://dx.doi.org/10.1109/IPDPS.2001.925093}}


\bibitem[\protect\citeauthoryear{{\c C}ataly{\"u}rek, Aykanat, and
  U{\c{c}}ar}{{\c C}ataly{\"u}rek et~al\mbox{.}}{2010}]%
        {CAU10}
{{\"U}. {\c C}ataly{\"u}rek}, {C. Aykanat}, {and} {B. U{\c{c}}ar}. 2010.
\newblock \showarticletitle{On two-dimensional sparse matrix partitioning:
  models, methods, and a recipe}.
\newblock {\em SIAM Journal on Scientific Computing\/} {32}, 2 (2010),
  656--683.
\newblock
\showDOI{%
\url{http://dx.doi.org/10.1137/080737770}}


\bibitem[\protect\citeauthoryear{\c{C}ataly\"{u}rek and
  Aykanat}{\c{C}ataly\"{u}rek and Aykanat}{1999}]%
        {CA99}
{U. \c{C}ataly\"{u}rek} {and} {C. Aykanat}. 1999.
\newblock \showarticletitle{Hypergraph-partitioning-based decomposition for
  parallel sparse-matrix vector multiplication}.
\newblock {\em IEEE Transactions on Parallel and Distributed Systems\/} {10}, 7
  (Jul 1999), 673--693.
\newblock
\showISSN{1045-9219}
\showDOI{%
\url{http://dx.doi.org/10.1109/71.780863}}


\bibitem[\protect\citeauthoryear{Christensen, Villa, and
  Vassilevski}{Christensen et~al\mbox{.}}{2015}]%
        {CVV15}
{M. Christensen}, {U. Villa}, {and} {P. Vassilevski}. 2015.
\newblock \showarticletitle{Multilevel techniques lead to accurate numerical
  upscaling and scalable robust solvers for reservoir simulation}. In {\em SPE
  Reservoir Simulation Symposium}. SPE--173257--MS.
\newblock
\showDOI{%
\url{http://dx.doi.org/10.2118/173257-MS}}


\bibitem[\protect\citeauthoryear{Christie and Blunt}{Christie and
  Blunt}{2001}]%
        {CB01}
{M. Christie} {and} {M. Blunt}. 2001.
\newblock \showarticletitle{Tenth {SPE} comparative solution project: a
  comparison of upscaling techniques}.
\newblock {\em SPE Reservoir Evaluation \& Engineering\/} {4}, 04 (August
  2001), 308--317.
\newblock
\showDOI{%
\url{http://dx.doi.org/10.2118/72469-PA}}


\bibitem[\protect\citeauthoryear{Davis}{Davis}{2006}]%
        {Davis06}
{T. Davis}. 2006.
\newblock {\em Direct Methods for Sparse Linear Systems}.
\newblock Society for Industrial and Applied Mathematics.
\newblock
\showDOI{%
\url{http://dx.doi.org/10.1137/1.9780898718881}}


\bibitem[\protect\citeauthoryear{Davis and Hu}{Davis and Hu}{2011}]%
        {DH11}
{T. Davis} {and} {Y. Hu}. 2011.
\newblock \showarticletitle{The {University of Florida} sparse matrix
  collection}.
\newblock {\em {ACM Transactions on Mathematical Software}\/}  {38} (2011),
  1:1--1:25.
\newblock


\bibitem[\protect\citeauthoryear{Faisal, Parthasarathy, and Sadayappan}{Faisal
  et~al\mbox{.}}{2014}]%
        {FPS14}
{S. Faisal}, {S. Parthasarathy}, {and} {P. Sadayappan}. 2014.
\newblock \showarticletitle{Global graphs: a middleware for large scale graph
  processing}. In {\em 2014 IEEE International Conference on Big Data (Big
  Data)}. 33--40.
\newblock
\showDOI{%
\url{http://dx.doi.org/10.1109/BigData.2014.7004369}}


\bibitem[\protect\citeauthoryear{Greiner}{Greiner}{2012}]%
        {Greiner12}
{G. Greiner}. 2012.
\newblock {\em Sparse Matrix Computations and Their I/O Complexity}.
\newblock Dissertation. Technische Universit\"{a}t M\"{u}nchen, M\"{u}nchen.
\newblock
\showURL{%
\url{http://mediatum.ub.tum.de?id=1113167}}


\bibitem[\protect\citeauthoryear{Gustavson}{Gustavson}{1978}]%
        {Gustavson78}
{F. Gustavson}. 1978.
\newblock \showarticletitle{Two fast algorithms for sparse matrices:
  multiplication and permuted transposition}.
\newblock {\em {ACM Transactions on Mathematical Software}\/} {4}, 3 (Sept.
  1978), 250--269.
\newblock
\showISSN{0098-3500}
\showDOI{%
\url{http://dx.doi.org/10.1145/355791.355796}}


\bibitem[\protect\citeauthoryear{Hong and Kung}{Hong and Kung}{1981}]%
        {HK81}
{J. Hong} {and} {H. Kung}. 1981.
\newblock \showarticletitle{{I/O} complexity: the red-blue pebble game}. In
  {\em Proceedings of the Thirteenth Annual ACM Symposium on Theory of
  Computing} {\em (STOC '81)}. ACM, 326--333.
\newblock
\showDOI{%
\url{http://dx.doi.org/10.1145/800076.802486}}


\bibitem[\protect\citeauthoryear{Irony, Toledo, and Tiskin}{Irony
  et~al\mbox{.}}{2004}]%
        {ITT04}
{D. Irony}, {S. Toledo}, {and} {A. Tiskin}. 2004.
\newblock \showarticletitle{Communication lower bounds for distributed-memory
  matrix multiplication}.
\newblock {\em {Journal of Parallel and Distributed Computing}\/} {64}, 9
  (2004), 1017--1026.
\newblock
\showDOI{%
\url{http://dx.doi.org/10.1016/j.jpdc.2004.03.021}}


\bibitem[\protect\citeauthoryear{Kalchev}{Kalchev}{2012}]%
        {K12}
{D. Kalchev}. 2012.
\newblock {\em Adaptive Algebraic Multigrid for Finite Element Elliptic
  Equations with Random Coefficients}.
\newblock Master's\ thesis. Sofia University.
\newblock
\showURL{%
\url{https://e-reports-ext.llnl.gov/pdf/594392.pdf}}


\bibitem[\protect\citeauthoryear{Kalchev, Ketelsen, and Vassilevski}{Kalchev
  et~al\mbox{.}}{2013}]%
        {KKV13}
{D. Kalchev}, {C. Ketelsen}, {and} {P. Vassilevski}. 2013.
\newblock \showarticletitle{Two-level adaptive algebraic multigrid for a
  sequence of problems with slowly varying random coefficients}.
\newblock {\em SIAM Journal on Scientific Computing\/} {35}, 6 (2013),
  B1215--B1234.
\newblock
\showDOI{%
\url{http://dx.doi.org/10.1137/120895366}}


\bibitem[\protect\citeauthoryear{Krishnamoorthy, {\c C}ataly{\" u}rek,
  Nieplocha, Rountev, and Sadayappan}{Krishnamoorthy et~al\mbox{.}}{2006}]%
        {KCNRS06}
{S. Krishnamoorthy}, {{\" U}. {\c C}ataly{\" u}rek}, {J. Nieplocha}, {A.
  Rountev}, {and} {P. Sadayappan}. 2006.
\newblock \showarticletitle{Hypergraph partitioning for automatic memory
  hierarchy management}. In {\em Proceedings of the ACM/IEEE SC 2006
  Conference} {\em (SC '06)}. 34--46.
\newblock
\showDOI{%
\url{http://dx.doi.org/10.1109/SC.2006.36}}


\bibitem[\protect\citeauthoryear{Niu, Lai, Faisal, Parthasarathy, and
  Sadayappan}{Niu et~al\mbox{.}}{2014}]%
        {NLFPS14}
{Q. Niu}, {P.-W. Lai}, {S. Faisal}, {S. Parthasarathy}, {and} {P. Sadayappan}.
  2014.
\newblock \showarticletitle{A fast implementation of {MLR-MCL} algorithm on
  multi-core processors}. In {\em 21st International Conference on High
  Performance Computing} {\em (HiPC '14)}. 1--10.
\newblock
\showDOI{%
\url{http://dx.doi.org/10.1109/HiPC.2014.7116888}}


\bibitem[\protect\citeauthoryear{Pagh and St\"{o}ckel}{Pagh and
  St\"{o}ckel}{2014}]%
        {PS14}
{R. Pagh} {and} {M. St\"{o}ckel}. 2014.
\newblock \showarticletitle{The input/output somplexity of sparse matrix
  multiplication}.
\newblock In {\em Algorithms - ESA 2014}, {A.~Schulz} {and} {D.~Wagner} (Eds.).
  Lecture Notes in Computer Science, Vol. 8737. Springer Berlin Heidelberg,
  750--761.
\newblock
\showISBNx{978-3-662-44776-5}
\showDOI{%
\url{http://dx.doi.org/10.1007/978-3-662-44777-2_62}}


\bibitem[\protect\citeauthoryear{Park, Smelyanskiy, Meier~Yang, Mudigere, and
  Dubey}{Park et~al\mbox{.}}{2015}]%
        {PSYMD15}
{J. Park}, {M. Smelyanskiy}, {U. Meier~Yang}, {D. Mudigere}, {and} {P. Dubey}.
  2015.
\newblock \showarticletitle{High-performance algebraic multigrid solver
  optimized for multi-core based distributed parallel systems}. In {\em
  Proceedings of the International Conference on High Performance Computing,
  Networking, Storage and Analysis} {\em (SC '15)}. Article 54, 12 pages.
\newblock
\showDOI{%
\url{http://dx.doi.org/10.1145/2807591.2807603}}


\bibitem[\protect\citeauthoryear{Rabin and Vazirani}{Rabin and
  Vazirani}{1989}]%
        {RV89}
{M. Rabin} {and} {V. Vazirani}. 1989.
\newblock \showarticletitle{Maximum matchings in general graphs through
  randomization}.
\newblock {\em Journal of Algorithms\/} {10}, 4 (1989), 557 -- 567.
\newblock
\showISSN{0196-6774}
\showDOI{%
\url{http://dx.doi.org/10.1016/0196-6774(89)90005-9}}


\bibitem[\protect\citeauthoryear{Satuluri and Parthasarathy}{Satuluri and
  Parthasarathy}{2009}]%
        {SP09}
{V. Satuluri} {and} {S. Parthasarathy}. 2009.
\newblock \showarticletitle{Scalable graph clustering using stochastic flows:
  applications to community discovery}. In {\em Proceedings of the 15th ACM
  SIGKDD International Conference on Knowledge Discovery and Data Mining} {\em
  (KDD '09)}. ACM, New York, NY, USA, 737--746.
\newblock
\showISBNx{978-1-60558-495-9}
\showDOI{%
\url{http://dx.doi.org/10.1145/1557019.1557101}}


\bibitem[\protect\citeauthoryear{van Dongen}{van Dongen}{2000}]%
        {vanDongen00}
{S. van Dongen}. 2000.
\newblock {\em Graph Clustering by Flow Simulation}.
\newblock Ph.D. Dissertation. University of Utrecht.
\newblock
\showURL{%
\url{http://www.library.uu.nl/digiarchief/dip/diss/1895620/full.pdf}}


\bibitem[\protect\citeauthoryear{Vastenhouw and Bisseling}{Vastenhouw and
  Bisseling}{2005}]%
        {VB05}
{B. Vastenhouw} {and} {R. Bisseling}. 2005.
\newblock \showarticletitle{A two-dimensional data distribution method for
  parallel sparse matrix-vector multiplication}.
\newblock {\em {SIAM Review}\/} {47}, 1 (2005), 67--95.
\newblock
\showDOI{%
\url{http://dx.doi.org/10.1137/S0036144502409019}}


\bibitem[\protect\citeauthoryear{Yamazaki and Li}{Yamazaki and Li}{2011}]%
        {YL11}
{I. Yamazaki} {and} {X. Li}. 2011.
\newblock \showarticletitle{On techniques to improve robustness and scalability
  of a parallel hybrid linear solver}. In {\em 9th International Conference on
  High Performance Computing for Computational Science} {\em (VECPAR '10)},
  {Jos{\'e} M. Laginha~M. Palma}, {Michel Dayd{\'e}}, {Osni Marques}, {and}
  {Jo{\~a}o~Correia Lopes} (Eds.). Springer Berlin Heidelberg, 421--434.
\newblock
\showISBNx{978-3-642-19328-6}
\showDOI{%
\url{http://dx.doi.org/10.1007/978-3-642-19328-6_38}}


\end{thebibliography}

\end{document}